\documentclass[sigconf]{acmart}
\AtBeginDocument{%
  }

\setcopyright{acmlicensed}
\copyrightyear{2026}
\acmYear{2026}
\acmDOI{XXXXXXX.XXXXXXX}
\acmConference[E-Energy '26]{The 17th ACM International Conference on Future and Sustainable Energy Systems}{June 22--25, 2026}{Banff, Canada}

\acmISBN{978-1-4503-XXXX-X/2018/06}

\usepackage{amsthm}
\usepackage{bbold}
\usepackage{graphicx}
\usepackage{subcaption}
\usepackage{comment}
\usepackage{amsmath}
\usepackage{booktabs} % For professional-looking tables
\usepackage{nomencl}
\makenomenclature
\usepackage{etoolbox} % 为了 ifstrequal 用在 \nomgroup

\usepackage{etoolbox}
\renewcommand{\nomgroup}[1]{%
  \ifstrequal{#1}{P}{\item[\textbf{Parameters}]}{%
  \ifstrequal{#1}{A}{\item[\textbf{Variables}]}{}}}

\newtheorem{remark}{Remark}
\newtheorem{theorem}{{Theorem}}%[section]
\newtheorem{lemma}[theorem]{{Lemma}}
\newtheorem{definition}{{Definition}}

\newtheorem{proposition}[theorem]{{Proposition}}

\newtheorem{assumption}{{Assumption}}

\newcommand{\addcite}[0]{\ifthenelse{\boolean{showcomments}}
{\textcolor{purple}{(add cite(s)) }}{}}%

\begin{document}

\title[Coherency Analysis in Nonlinear Heterogeneous Power Networks]{Coherency Analysis in Nonlinear Heterogeneous Power Networks: A Blended Dynamics Approach}

\author{Yixuan Liu}
% \author{Yingzhu Liu}
% \email{@pku.edu.cn}
\affiliation{%
  \institution{Peking University}
  \city{Haidian}
  \state{Beijing}
  \country{China}
}
\email{yixuanliu2024@outlook.com}
\orcid{0009-0003-6630-6125}
% School of Advanced Manufacturing and Robotics

\author{Yingzhu Liu}
\affiliation{%
  \institution{Peking University}
  \city{Haidian}
  \state{Beijing}
  \country{China}}
\email{yzliucoe@pku.edu.cn}

\author{Pengcheng You}
\authornote{Corresponding author.}
\affiliation{%
  \institution{Peking University}
  \city{Haidian}
  \state{Beijing}
  \country{China}}
\email{pcyou@pku.edu.cn}

\renewcommand{\shortauthors}{Liu et al.}

\begin{abstract}
  %%%%%% Background %%%%%%
Power system coherency refers to the phenomenon that machines in a power network exhibit similar frequency responses after disturbances, and is foundational for model reduction and control design.
Despite abundant empirical observations, 
the understanding of coherence in complex power networks remains incomplete where the dynamics could be highly heterogeneous, nonlinear, and increasingly affected by persistent disturbances such as renewable energy fluctuations.
To bridge this gap, this paper extends the blended dynamics approach, originally rooted in consensus analysis of multi-agent systems, to develop a novel coherency analysis in power networks. 
We show that the frequency responses of coherent machines coupled by nonlinear power flow %nodes in a highly coherent network  % coherent = similar; coherence = similarity (of frequency responses)
can be approximately represented by the blended dynamics, which is 
% We propose to represent the coherent frequency response of all nodes % 
a weighted average of nonlinear heterogeneous nodal dynamics, even under time-varying disturbances. 
% 很大的一个问题是，我去掉了coherent response这个词后，就难以简洁地指代 wb(t)了，完整来讲是the trajectory of the blended dynamics 或 the frequency response governed by the blended dynamics
Specifically, by developing novel bounds on the difference between the trajectories of nodal dynamics and the blended dynamics, we identify two key factors---either high network connectivity or small time-variation rate of disturbances---that contribute to coherence. 
% They enable the nodal frequencies to both rapidly approach the blended dynamics for arbitrary initial states and closely follow the blended dynamics in the long term even when the frequencies cannot settle to an equilibrium. 
They enable the nodal frequencies to rapidly approach the blended-dynamics trajectory from arbitrary initial state. Furthermore, they ensure the frequencies closely follow this trajectory in the long term, even when the system does not settle to an equilibrium.
% we derive upper bounds on the difference between actual nodal responses and the coherent response, which we refer to as coherence error.  跳过定义error 的步骤了，直接描述含义吧！
% Our results reveal that the time-variation rate of persistent disturbances and the network connectivity are two key factors that govern the evolution of the coherence error and also determine its limiting behavior. In particular, we identify that in the transient phase, while higher network connectivity expedites the error decay, it sometimes amplifies the coherence error in the early stage before it starts to decay, suggesting an intriguing tradeoff that has been unnoticed.
These insights contribute to the understanding of power system coherency and are further supported by simulation results.
\end{abstract}







\maketitle

% \begin{center}
%   \textsc{List of Symbols}
% \end{center}

% ========== Parameters ==========
\nomenclature[P01]{$\mathcal{N}$}{Set of nodes (corresponding to buses) in the network graph.}
\nomenclature[P02]{$N$}{Number of nodes (buses).}
\nomenclature[P03]{$\mathcal{E}$}{Set of edges (corresponding to lines) in the network graph.}
\nomenclature[P04]{$E$}{Number of edges (lines).}
\nomenclature[P05]{$M_i$}{Inertia constant of the generator at bus $i$.}
\nomenclature[P06]{$f_i(\cdot)$}{Frequency-dependent terms at bus $i$.}
\nomenclature[P07]{$\xi_i(t)$}{Net power injection (local generation minus local load) at bus $i$ and time $t$.}
\nomenclature[P08]{$B_{ij}$}{Sensitivity of the line flow to the phase angle difference between bus $i$ and bus $j$.}
\nomenclature[P09]{$L_B$}{Weighted Laplacian matrix of the network graph.}
\nomenclature[P10]{$M_b$}{System-wide average of the inertia constants.}
\nomenclature[P11]{$f_b(\cdot)$}{System-wide average of the frequency-dependent terms.}
\nomenclature[P12]{$\xi_b(t)$}{System-wide average of the local net power injections at time $t$.}
\nomenclature[P13]{$\mu$}{Minimum level of damping effect over all frequencies and all buses.}
\nomenclature[P14]{$L$}{Maximum level of damping effect over all frequencies and all buses.}
\nomenclature[P15]{$\lambda_2$}{Second-smallest eigenvalue of the matrix $M^{-1}L_B$.}
\nomenclature[P16]{$\lambda_2^L$}{Second-smallest eigenvalue of the matrix $L_B$.}
\nomenclature[P17]{$C$}{Maximum rate of change of $\xi_i(t)$ over all buses and all time.}
\nomenclature[P18]{$C_{\lim}$}{Maximum rate of change of $\xi_i(t)$ over all buses as time goes to infinity.}
\nomenclature[P19]{$\Delta \xi$}{Vector of initial abrupt changes in the values of $\xi_i(t)$.}
\nomenclature[P20]{$k$}{Uniform scaling factor for all line sensitivities $B_{ij}$'s.}

% ========== Variables ==========
\nomenclature[A01]{$\theta_i$}{Voltage phase angle at bus $i$.}
\nomenclature[A02]{$\omega_i$}{Voltage frequency at bus $i$.}
\nomenclature[A03]{$p_{e,i}$}{Real power injected from bus $i$ into the network.}
\nomenclature[A04]{$\omega_b$}{State variable of the blended dynamics.}
\nomenclature[A05]{$\bar\omega$}{Center of Inertia (COI) frequency.}

\printnomenclature

\section{Introduction}
% Maintaining frequency synchronization in power systems is crucial for stability, 
Stable operation of power systems requires machines to operate at closely synchronized frequencies, 
% relies on the frequency synchronization among its components,
and loss of synchrony may lead to inter-area oscillations, power flow instability, and even cascading failures~\cite{kundur2007power}. Empirical observations have suggested that connected machines in power networks tend to exhibit similar frequency responses to external disturbances, a phenomenon referred to as \emph{power system coherency}~\cite{chow2013power}. Coherence has been widely exploited to support model reduction~\cite{germond1978dynamic} and control design~\cite{jiang2021grid}, simplifying large-scale power system analysis while preserving the dominant dynamics.

Extensive efforts have been made to understand the rationale behind power system coherency. Classic slow coherency analyses~\cite{romeres2013novel,tyuryukanov2020slow} identify groups of 
coherent machines, i.e., machines with highly similar responses, 
% nodal machines that exhibit coherent behavior 这里刻意解释了"coherent machines"的意思，方便后面说the network is highly coherent 或者 the nodes are highly coherent.
based on the structure of power networks. Then each group of coherent machines is aggregated into a larger equivalent machine. 
Although it is shown that coherence emerges from strong interconnections within each group, quantifying such relationships remains a challenge, 
particularly in establishing theoretical bounds on the differences between the nodal and aggregate responses.
% 为什么加了这句话：因为difficult to quantify对应的是贡献的其中一层（刻画现象），而贡献的另一层是聚合近似的效果好坏，在这层上slow coherency也提供了很多聚合方法，而形成对比的点在于它们缺乏这样的bounds。

Another line of work uses $H_2$-norm~\cite{bamieh2012coherence,tegling2015price,andreasson2017coherence} and $H_\infty$-norm~\cite{pirani2017system} to characterize the differences in nodal angle or frequency responses in a power network, which enables an explicit evaluation of how coherence is influenced by network connectivity~\cite{bamieh2012coherence}, line parameters~\cite{tegling2015price}, machine parameters~\cite{pirani2017system} and controller types~\cite{andreasson2017coherence}. 
However, these results are predicated on the assumption of homogeneous nodal dynamics and are not applicable to more practical scenarios.

More recent studies have attempted to relax the homogeneity assumption. \cite{paganini2019global} adopts a milder proportionality assumption instead, and provides a first-cut approximation to heterogeneity. 
Notably, in the presence of heterogeneous nodal dynamics, the frequency response of the full network is represented by the trajectory of the Center of Inertia (COI) frequency, defined as a weighted average of nodal frequencies. \cite{min2025frequency} takes a step further and develops a frequency-domain analysis framework for heterogeneous linear time-invariant network dynamical systems.
The study reveals that coherence is influenced not only by network connectivity, but also by the potential frequency composition of disturbances through the harmonic mean of nodal transfer functions. Despite the remarkable progress, the analysis remains restricted to approximated linear models. %或者linearized models?

One promising alternative is to develop a time-domain analysis of power system coherency via \emph{blended dynamics}. Such a notion stems from multi-agent systems and is used to characterize the consensus and emergent behavior of agents with strong couplings among them. % coupling应该不难理解？strong具体定义为合适的类型+足够的增益，应该不需要说的很具体？（Blended dynamics文章里标题就是 Robustness of Synchronization of Heterogeneous Agents by Strong Coupling）
% strongly coupled agents. 
% via the weighted average of their individual dynamics. 
The blended dynamics approach inherently accommodates heterogeneous nonlinear dynamics~\cite{kim2015robustness}. In general, consensus-enforcing network couplings are necessary for blended dynamics to emerge, leading to a variety of control designs driven by neighborhood communication~\cite{kim2015robustness,lee2020tool,lee2022blended}. 
However, to the best of our knowledge, none of the structures directly fit the nonlinear physical coupling of power flows between buses in a power network. 
As a result, it remains unclear whether blended dynamics and coherence are correlated for power systems.
% However, although existing results have covered various coupling forms, they generally assume the 

% 本来感觉是偏向刻画自然现象的：自然现象里潜在地可能有一个coherent response作为target，我猜了一个，并证明它确实命中了target，从而解密了这个自然现象。但是，自然现象里到底哪里有这样一个target，有一个rigorously defined的coherent response等待我来寻找？
% 但如果写成我先构建一个近似，然后检验这个近似的效果不错，it is a well approximation of the full system behavior under ... conditions，就变成了一个工程的、人为的事情了。
% ❓到底有先有，我再试图寻找；还是先就没有，我构造一个并证明不错？
%  现在的写法还在摇摆不定，同时带上了"reduced-order model"做工程近似和"emergent network behavior"做现象刻画这两个方面，试图融合在一起吧。
To fill these gaps, we extend the blended dynamics approach to develop a coherency analysis in nonlinear heterogeneous power networks subject to persistent time-varying disturbances. 
Specifically, we propose a reduced-order model based on a particular weighted average of (possibly nonlinear) nodal dynamics, namely the blended dynamics. Then we show that such blended dynamics indeed characterizes the behavior of the whole network that emerges when all nodes are coherent. 
% 没有再用coherent response这个词，而是直接解释成了"the behavior of the whole network which emerges when all nodes are highly coherent"，并强调涌现，强调coherent（coherent这个词前面定义过）.
% and show that it is indeed the emergent network behavior, i.e., coherent response, 
% due to the consensus-enforcing power flow coupling. % 这里先漏掉了consensus-enforcing power flow coupling。毕竟后面contributions里也可以继续强调。这里句子太长了，写不下了。
Basically, we quantify the difference between the trajectories of the nodal dynamics and the blended dynamics, which reflects the level of coherence and is referred to as the \emph{coherence error}. 
We first analyze nonlinear nodal dynamics under linearized power flows, where we derive time-dependent bounds on the coherence error. These bounds reveal the regimes in which the error remains small for all $t > 0$ or decays exponentially to a small level. 
Then we further show that similar insights carry over to the nonlinear power flow case, under mild additional conditions. %✅though additional restrictions on disturbances and initial conditions are required. 

In summary, our results contribute to the understanding of coherence in power systems in the following ways:

\noindent
% \textbf{Characterizing coherent response:} 
\textbf{Coherency-based reduced-order approximation:}
% provide a coherency-based reduced-order approximation of network frequency responses based on blended dynamics. 
We formally show that the physical coupling of power flows can serve as a consensus-enforcing input to power system nodal dynamics, thus driving all nodal frequencies toward the trajectory of the reduced-order blended dynamics. 
% derive the emergent blended dynamics as coherent response. 
Such a time-domain coherency analysis complements the literature by accommodating heterogeneous power networks with possibly nonlinear nodal dynamics and nonlinear power flows. 
    %Our characterization accommodates heterogeneous nodal dynamics which may include some nonlinearities beyond typical linear damping behavior. By removing the homogeneity/proportionality assumptions, our results highlight how heterogeneous nodal dynamics collectively shape the overall coherent response.

\noindent    
\textbf{%Quantitative analysis of the
Characterization of the level of coherence through explicit bounds: %  Characterizing degree of coherence:
  } We develop novel bounds for all $t>0$ on the difference between nodal frequency trajectories and the blended-dynamics trajectory. These bounds shed light on how network connectivity could enhance the coherence level in both the limiting and the transient phase, 
  % strike a tradeoff between early amplification and decay rates of coherence error, 
  as well as how the time-variation rate of persistent disturbances plays a critical role in coherence, as compared with the prior work~\cite{min2025frequency} that only provides finite-time bounds and cannot handle step disturbances.%✅cannot handle
    %which enables us to reveal several key factors that have been insufficiently analyzed. 
    %In particular, our analysis incorporates a class of persistent time-varying disturbances and highlights the influence of the time-variation rate of the disturbances. %Besides, we demonstrate a tradeoff in how network connectivity impacts coherence over time. These findings provide new insights into the mechanism that leads to power system coherence.
    %These analyses complement prior work~\cite{min2025frequency} which explores similar factors but provides time domain coherence bounds only within some finite time interval $[0,T]$ when subjected to persistent disturbances.

The remainder of the paper is organized as follows. Section~\ref{sec:preli} introduces necessary preliminaries on notations and reviews the general framework of the blended dynamics approach. Section~\ref{sec:problem} defines the nonlinear heterogeneous power network model and formulates our problem of coherency analysis. To address this problem, Section~\ref{sec:methods} constructs the specific blended dynamics tailored for the power network. Then the corresponding coherence error is characterized in Section~\ref{sec:results}, with explicit error bounds under both linearized and nonlinear power flows. Section~\ref{sec:simulations} provides numerical simulations that validate our theoretical results, and Section~\ref{sec:conclusion} concludes the paper. 
% “formulates our problem of coherency analysis”这里没想到怎么具体概括这个问题，所以这个section显得跟后面两个sections逻辑上跟的不紧。为了显得逻辑连贯一点，加了"To address this problem" "Then the corresponding"等连词
    
\section{Preliminaries}\label{sec:preli}
\subsection{Notations}
Let $x = x(t)$ denote the system state at time $t$. Its time derivative is written as $
\dot{x} := \frac{\partial x}{\partial t}.
$ For a differentiable function $f : \mathbb{R}^n \to \mathbb{R}$, such as a Lyapunov function, the time derivative along the system trajectory is \( \dot{f} := (\nabla_x f(x))^T \dot{x} \) and its Hessian is denoted by $\nabla^2 f(x)$. For any time-dependent signal $u(t)$, we denote the left and right limits at $t=0$ by $
u(0_-) := \lim_{t \to 0^-} u(t), \ u(0_+) := \lim_{t \to 0^+} u(t),
$ whenever these limits exist. If $u(t)$ is continuous at $t=0$, we simply write $u(0) = u(0_-) = u(0_+)$. Unless otherwise specified, expressions involving $t > 0$ start from $t = 0_+$ (i.e., after any possible initial discontinuity), while expressions involving $t \geq 0$ start from $t = 0_-$.
% 注意不要使用重复的符号

The vector of all ones is denoted by $\mathbb{1}$ and the identity matrix is denoted by $I$; their dimensions are given as a subscript when necessary. For a vector $x$, let $|x|$ denote its Euclidean norm, i.e., $
|x| := \sqrt{x^T x}$. For a matrix $A$, let $|A|$ denote its induced 2-norm, i.e., $|A| := \sup_{x \neq 0} \frac{|Ax|}{|x|}$. Let $\sigma_m(A)$ denote the minimum singular value of $A$. In particular, if $A$ is real symmetric, $\lambda_i(A)$ denotes the $i$-th smallest eigenvalue of $A$. For two real symmetric matrices $A$ and $B$ of the same dimension, 
% symmetric matrices $A,B \in \mathbb{R}^{N \times N}$, 
the relation $
A \leq B$ means that $B-A$ is positive semidefinite, while $ A < B$ means that $B-A$ is positive definite. 
The functions $\sin(\cdot)$ and $\cos(\cdot)$ are applied element-wise when used with vector arguments. 
For a set of scalars $\{x_i : i \in \mathcal{N}\}$ with an index set $\mathcal{N} := \{1,2,\dots,N\}$, the diagonal matrix is written as $
\mathrm{diag}(x_1,\dots,x_N) $ or equivalently, as $\mathrm{diag}(x_i,\ i \in \mathcal{N})$. 
For two functions $g$ and $h$, we write $g(x) = \Theta(h(x))$ if there exist positive constants $c_1, c_2$ and $x_0$ such that $c_1 h(x) \le |g(x)| \le c_2 h(x)$ for all $x \ge x_0$. 
% (resp. $ A < B$) means that $B-A$ is positive semidefinite (resp. positive definite).

% ✅notion
\subsection{Notion of Blended Dynamics} 
Here we briefly review the core idea of the \emph{blended dynamics} approach~\cite{kim2015robustness} in a general setting, which characterizes the consensus and group behavior in multi-agent systems. 
% model setup
Specifically, consider a group of heterogeneous agents with the dynamics of agent $i$ given by $\dot{x}_i = h_i(x_i,t) + u_i(t)$, where $x_i$ is the agent's state. $h_i$ is a vector field representing the agent's (possibly nonlinear) local dynamics, which may include time-varying signals and disturbances. $u_i(t)$ is a coupling signal, typically designed based on neighborhood communication, to enforce consensus among agents. 
% Main conclusion
When (approximate) consensus is achieved, i.e., when all agents follow highly similar trajectories, their collective behavior can be approximated by the so-called blended dynamics, which is constructed using a weighted average of the vector fields of all agents:%
\begin{align}\label{eq:pre blended dyn}
% \notag \\
% \dot {x}_b = \sum_i \beta_i h_i(x_b,t)\Big/\sum_i \beta_i,
    \dot{{x}}_{b} = \textstyle \sum_i \beta_i h_i(x_b,t)\big/ \textstyle \sum_i \beta_i,
\end{align}
where the weights $\beta_i$ are chosen such that $\sum_i \beta_iu_i= 0$, reflecting the fact that the coupling signals $u_i$ are designed as internal exchanges and should not contribute to the group's net motion. Such dynamics captures an emergent behavior that may not be exhibited in any individual agent but generated by a mixture of all agent dynamics, thus referred to as blended dynamics.

\section{Problem Statement}\label{sec:problem}
Consider a power network with a connected undirected graph $(\mathcal{N},\mathcal{E})$, where $\mathcal{N}:=\{1,\dots,N\}$ is the set of nodes and $\mathcal{E} \subseteq\{\{i, j\} : i, j \in \mathcal{N}, i \neq j\}$ % \mathcal{E} \subset \mathcal{N}\times \mathcal{N}$ 
is the set of edges with $|\mathcal{E}| = E$. Each node is usually a bus, while each edge describes a connection between two buses, such as a transmission line. 
Without loss of generality, we assume there is only one (aggregate) controllable generator at each bus. The dynamical model of bus $i$ is given by 
 \begin{subequations}\label{eq:nonlinear dynamics}
        \begin{align}
           \dot \theta_i &= \omega_i,\label{eq:angle}\\
         \dot \omega_i & =\frac{1}{ M_i}\left( f_i(\omega_i)+ \xi_i(t) - p_{e,i}\right).\label{eq:agent}
        \end{align}
   \end{subequations}
Here $\theta_i$ and $\omega_i$ are the voltage phase angle and frequency relative to the utility frequency given by $2 \pi 50$ or $2 \pi 60$ Hz. % 这句话写法参考了"Robust Decentralized Secondary Frequency  Control in Power Systems: Merits and Tradeoffs", TAC, 2019, Florian Do ̈ rfler
$M_i >0$ represents the inertia constant of the generator. $f_i:\mathbb{R}\to\mathbb{R}$ is a (possibly nonlinear) function of local frequency deviation, summarizing all frequency-dependent terms such as generator damping. 
% We assume $f_i(0) = 0$, since any constant offset $f_i(0)$ can be absorbed in 
$\xi_i(t)$ is the real-time net power injection (local generation minus local load) at bus $i$. Since it includes the disturbances from fluctuating demand, variable generation, etc, we will simply refer to $\xi_i(t)$ as the "disturbance" (with a slight abuse of terminology). 
% net power injection at bus i
% defined as \textbf{the local generation minus the local load demand} at bus
% that represents the real-time net power imbalance at bus $i$ arising from fluctuating demand, variable generation, etc. 
$p_{e,i}$ denotes the real electrical power injected from bus $i$ into the network, 
which can be expressed as follows under the assumptions of constant voltage magnitudes and lossless lines: 
% assume the voltage magnitude |Vj | is constant for each bus; each line is lossless, x_ij is a reactance;
% approximated using the DC Power Flow equation, given by
        \begin{align}
        % \label{eq:power flow}
         p_{e,i} &= \sum_{j \in \mathcal{N}_i} B_{ij} \sin(\theta_i - \theta_j), \tag{2c} \label{eq:power flow}
        \end{align} where $\mathcal{N}_i$ is the set of buses directly connected to bus $i$. $B_{ij}$ characterizes the sensitivity of the line flow to the phase angle difference between bus $i$ and bus $j$, given by $B_{ij} = |V_i||V_j|/x_{ij}$, where $|V_i|$ and $|V_j|$ are the constant voltage magnitudes, and $x_{ij}$ is the line reactance.
        % (cf. the appendix of~\cite{you2018stabilization} for the derivation of $B_{ij}$). 
        % Note that the variables $\theta_i$ and $\omega_i$ are \emph{deviations} from their nominal values.
% Note that the variables $\theta_i$ and $\omega_i$ represent, respectively, the phase angle relative to a reference frame rotating at the nominal synchronous frequency, and the frequency deviation from this nominal value. 
% Wikipedia: synchronously rotating reference frame
Note that we set $f_i(0) = 0$, since any constant offset $f_i(0)$ can be absorbed in $\xi_i(t)$. 
Besides, we model $\xi_i(t)$ to be continuously differentiable for $t \in(0,\infty)$, in order to capture the impact of their rate of change $\dot \xi_i(t)$, while allowing a finite jump at $t = 0$ to accommodate possible abrupt changes.
%原来的：assume...; to simplify the representation of their time-variation rate,  使用 "model" 能更准确地反映我们作为研究者的主观能动性，表明这是一个经过深思熟虑后为了反映系统本质而做出的选择，而不是一个随意的假设；在阐述你的设定或假设时，最好从“问题本身”或“系统特性”出发，说明其物理意义或对后续分析的重要性，而不是从“作者的便利性”出发。

\begin{remark}
The term $f_i(\omega_i)$ in \eqref{eq:agent} captures nonlinear frequency-dependent behaviors beyond classical linear damping or linear droop control.
Examples include the nonlinear frequency-sensitive loads~\cite{pradhan2016frequency} and nonlinear droop controllers with saturation~\cite{kundur2007power}.
In addition, nonlinear local primary controllers are increasingly adopted to enhance transient performance, such as the load-side controllers in~\cite{zhao2014design} and the inverter-based controllers in~\cite{cui2022reinforcement}.
\end{remark}
% These practical motivations make it essential to include such nonlinearities in the coherency analysis.

\begin{figure}[t]
    \centering
    \includegraphics[width=0.4\textwidth]{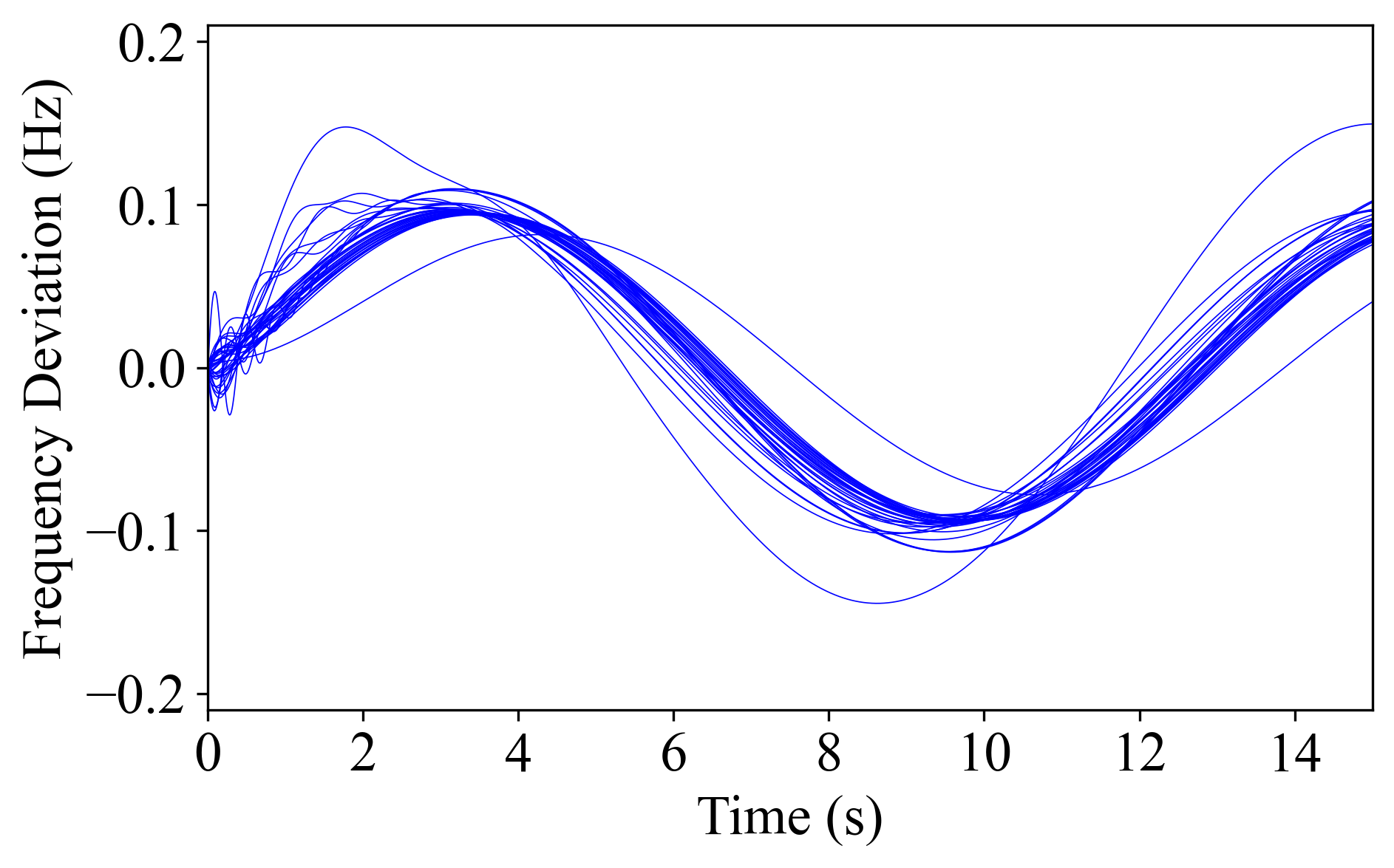} % 调整宽高比
    \caption{Frequency responses of $35$ generators in the Icelandic power grid~\cite{edinburgh_power_systems} following a disturbance.}
    \label{fig:question}
\end{figure}
While each node in the network graph has its own dynamics \eqref{eq:agent}, it is often observed that the frequencies $\omega_i(t)$ of different nodes evolve in highly similar patterns, as shown in the illustrative example in Fig.~\ref{fig:question}. Such observations, known as the power system coherency, are often exploited to approximate the frequency trajectories $\omega_i(t)$'s of the full $N$-th order system by a lower-order system trajectory $\omega_b$~\cite{chow2013power}. However, it remains an open problem \emph{how to construct a reduced-order frequency model such that its behavior approximates that of % to approximate 
a coherent nonlinear system \eqref{eq:nonlinear dynamics} with bounded error, especially when the nodal dynamics \eqref{eq:agent} are heterogeneous and are subject to persistent time-varying disturbances $\xi_i(t)$.}
% 主语对应：approximate trajectories by a trajectory；open question--open problem
% 这样的叙述会有一个问题，如果是近似整个nonlinear system，而不仅仅是频率响应，读者可能会奇怪角度的动力学到哪里去了。需要再补充解释？解决：a reduced-order frequency model such that its behavior approximates that of a coherent nonlinear system 加更多限定词；approximate的主语和宾语更准确。
% The former causes the frequencies to keep changing rather than settling to an equilibrium, while the latter tends to make each nodal frequency change in a different way. 

% 一直说approximation的话，感觉让陈述变得太偏工程了，而不是刻画一个大家自然都收敛到一个低阶动力学的自然现象。
One possibility is to draw inspiration from the blended dynamics framework, which captures the collective behavior of nonlinear heterogeneous agents under approximate consensus. 
The role of blended dynamics is strikingly aligned with our problem in the power system context.
However, the challenge remains whether the physical coupling between nodes via power flows can be consensus-enforcing. % The original paragraph follows a simple, linear path: "Here is an idea" → "Let's use it" → "We will analyze the error." This is logical, but it lacks tension. It doesn't explain \textit{why} this is a difficult or interesting problem.  %原先的写法掩盖了构造blended dynamics的贡献，只写了有这么一个东西，很自然地拿过来用，没有突出此间的联系是很巧妙地，但是拿过来又是有难度的，所以explicitly model the blended dynamics也是一个工作量。（之前选择不把构造wb放进goal是因为，不方便严格地说要构造一个达到什么效果的wb；但是构造一个motivated by blended dynamics的wb也是合理的目标，之前没理顺）
%Given its close resemblance to the notion of coherency, it is natural to extend it in the context of power networks and consider the blended dynamics as a candidate for the above question. 
%Based on this novel construction of $\omega_b$, 
To this end, the focus of this paper is to explicitly model the blended dynamics for power networks and develop theoretical characterizations of the resulting approximation error 
$$\max_{i \in \mathcal{N}} |\omega_i - \omega_b|,$$
which is later referred to as the \emph{coherence error}. 
This error is a proxy for nodal frequency differences $|\omega_i - \omega_j|$ and a measure of coherency for a power network. % 这里需要点明“|wi - wb|的上界也是|wi - wj|的上界”的这层关系吗？（这两个error不是互相等价的，|wi - wb|很大的时候|wi  - wj|也可能很小）
We particularly aim to identify conditions under which this error is small, 
%---either for all $t>0$ or asymptotically as $t \to \infty$---
despite the heterogeneity and disturbances. In the following sections, we will first introduce the constructed form of the blended dynamics for power networks and then present our bounds on the corresponding coherence error.

\section{Construction of Blended Dynamics}\label{sec:methods}
% \section{Blended Dynamics} % 不需要the
In this section, we are inspired by the notion of blended dynamics to propose a candidate reduced-order model tailored for the swing dynamics \eqref{eq:nonlinear dynamics}. As shown in \eqref{eq:pre blended dyn}, blended dynamics builds on a weighted average of individual vector fields, where the weights $\beta_i$ are designed such that the coupling signals $u_i(t)$ are canceled out, i.e., $\sum_i \beta_i u_i(t) = 0$. Here $u_i$ corresponds to the power flow coupling $-M_i^{-1}p_{e,i}$ given in \eqref{eq:power flow}, thus we can choose $\beta_i = M_i$ due to the inherent power flow balance $\sum_i p_{e,i} = 0$. Based on this idea, we arrive at the following definition of blended dynamics.
% 本来直接写construct the blended dynamics for the nonlinear system (1)-(2)也可以。但是总觉得需要多强调一下是拓展这个框架，不是直接套用这个框架，所以加入了adapt
% are inspired by the notion of  ；reduced-order model tailored for the swing dynamics；builds on a weighted 

% \begin{definition}
% [Blended Dynamics of Power Networks]
%     Let $\omega_b\in \mathbb{R}$ be the frequency of the blended dynamics for the swing dynamics \eqref{eq:nonlinear dynamics}. Then starting from the initial point $$\omega_b(0)=\sum_{i=1}^N M_i \omega_i(0) \big/\sum_{i=1}^N M_i,$$ 
%     $\omega_b$ evolves according to
%     % For a given networked system \eqref{eq:nonlinear dynamics}, define its blended dynamics as
%        \begin{align}\label{eq: blended dyn}
%             \dot \omega_b &= \frac{\sum_{i=1}^N M_i \left(M_i^{-1}(f_i(\omega_b)+ \xi_i)\right)}{\sum_{i=1}^N M_i}   = \frac{f_b (\omega_b) + \xi_b}{M_b},
%        \end{align}
%     where $M_b :=(1/N)\sum_{i=1}^N M_i$, $f_b(\omega_b):= (1/N) \sum_{i=1}^N f_i(\omega_b)$, and $\xi_b:= (1/N) \sum_{i=1}^N \xi_i$.
% \end{definition}
% 分开是为了避免where后面需要定义的notation太多太长吗？感觉这个看上去没有那么像一个定义；初始值其实也不需要重点突出
\begin{definition}
[Blended Dynamics of Power Networks]
    Let $\omega_b\in \mathbb{R}$ be the state of the blended dynamics for the swing dynamics \eqref{eq:nonlinear dynamics}.  
    $\omega_b$ evolves according to
    % For a given networked system \eqref{eq:nonlinear dynamics}, define its blended dynamics as
       \begin{align}\label{eq: blended dyn}
            \dot \omega_b &= \frac{\sum_{i=1}^N M_i \left(M_i^{-1}(f_i(\omega_b)+ \xi_i)\right)}{\sum_{i=1}^N M_i}   = \frac{f_b (\omega_b) + \xi_b}{M_b},
       \end{align}
       starting from the initial point $$\omega_b(0)=\sum_{i=1}^N M_i \omega_i(0) \big/\sum_{i=1}^N M_i.$$
    Here, $M_b :=(1/N)\sum_{i=1}^N M_i$, $f_b(\omega_b):= (1/N) \sum_{i=1}^N f_i(\omega_b)$, and $\xi_b:= (1/N) \sum_{i=1}^N \xi_i$.
\end{definition}

The existence and uniqueness of the solution to \eqref{eq: blended dyn} usually requires local Lipschitz continuity on $f_b(\cdot)$. This will be satisfied by Assumption~\ref{assump1} to be introduced later. 
\begin{remark}
% 这段话应该第一句先点明COI也用于以单条轨迹代表整体。然后再说定义直觉一样，但是问题在于xxx
The typical COI frequency $\bar \omega:= (\sum_i M_i \omega_i)/(\sum_i M_i)$ is defined based on the same intuition, taking weighted average and assigning larger weights to nodes with larger inertia. The COI trajectory was also used to represent system-wide frequency evolution in prior works,e.g.,~\cite{paganini2019global}.
% to serve as a coherent response of swing dynamics, 
But these works rely on the proportionality assumption, which imposes a uniform damping-to-inertia ratio across all nodes. This assumption enables the the dynamics of $\bar \omega$ to decouple neatly from the nodal deviations $\omega_i - \bar \omega$ and reduce to a simple first-order equation of the same form as \eqref{eq: blended dyn}. However, it becomes restrictive in modern power systems with a high penetration of distributed resources that are highly heterogeneous. In such heterogeneous systems, the dynamics of $\bar{\omega}$ is generally high-order and entangled with full network states. Instead, our blended dynamics provides an approximation of the COI dynamics while preserving the simple structure.
% 另一种可能的顺序是：依赖假设——假设很局限——但是过去工作用这个假设是为了分析简便——在异质下，这个简便丢失——我们提供了近似。
\end{remark}

In the blended dynamics literature, the accuracy of the blended dynamics representation crucially relies on the presence of consensus-enforcing couplings $u_i$. Different from the large variety of coupling types in the existing works, the physical coupling of power flows \eqref{eq:power flow} has a unique structure: It is nonlinear (sinusoidal mapping) and purely integral-based (phase angle difference). 
%It depends on the sine of angle differences between neighboring nodes, and since angles are the integrals of frequencies, the coupling drives nodal frequencies only through their integrated differences. ✅括号标注的精简写法
This calls for a systematic analysis of the coherence error $\max_{i \in \mathcal{N}}|\omega_i - \omega_b|$ under the power flow coupling, which will be presented in the next section.

\section{Characterization of Coherence Error}\label{sec:results} %标题不需要the
% 我在这里非要加一句bound的作用和意义，是因为我特别害怕别人不知道wb不是一个启发式的提法而是wi 确实会收敛到的一个一阶动力学，我也害怕别人get不到这些bounds是有理论价值的，是能给coherency提供insights的。
% 不敢一直反复用reduced-order approximation这样的词，怕让整个叙述的风格更加偏向工程构造而不是物理现象刻画。
% 既想强调验证wb，又想强调发掘关键条件（影响因素）并提供insights，但难以在句子里同时融入这两个意思
In this section, we derive time-dependent upper bounds on the coherence error $\max_{i \in \mathcal{N}}|\omega_i(t) - \omega_b(t)|$ over the entire time horizon. These bounds show that under proper conditions of network parameters and disturbance properties, all nodal frequencies $\omega_i$ can be effectively approximated by the blended dynamics \eqref{eq: blended dyn}. For illustration purposes, we will first consider only nonlinear nodal dynamics with linearized power flows in Section~\ref{subsec:linearflow}. Then the analysis will be extended to the nonlinear power flow setting in Section~\ref{subsec:nonlinearflow}.
% 尽量拆分短句(These bounds; Then the analysis);  For illustration purposes; 对于长的section，加上章节导航，便于阅读。

\begin{comment} CDC version
which not only confirm the effectiveness of the blended dynamics characterization (answering Q1), but also quantify how closely the actual nodal frequencies follow it (answering Q2).
% present our results on bounding the differences $|\omega_i(t) - \omega_b(t)|$. 
Throughout the following discussion, we refer to $|\omega_i(t) - \omega_b(t)|$ as the \emph{coherence error}.
% which quantifies how closely each node’s frequency follows the blended dynamics. 上一段刚讲过了

% \textcolor{olive}{To be rewritten.}We first establish a limiting upper bound on the coherence error, 
% and then develop a transient upper bound that explicitly quantifies how fast the coherence error shrinks towards its limiting bound. Based on these bounds, we further investigate the key factors that influence the degree of coherence.
%Specifically, we will derive explicit bounds on the differences $|\omega_i(t) - \omega_b(t)|$, whichf
\end{comment}

Before proceeding to the results, we make an assumption on the functions $f_i(\omega_i)$ in \eqref{eq:agent}.
\begin{assumption}
    \label{assump1}
$f_i(\omega_i)$ is continuously differentiable and there exist constants $\mu >0$ and $L>0$ such that, for $\forall i \in \mathcal{N}$ and $\forall \omega_i \in \mathbb{R}$,
      \begin{align}\label{eq:contractive}
            -L \leq \frac{1}{M_i}\frac{ d f_i(\omega_i)}{d \omega_i} \leq -\mu.
        \end{align}
    % \end{enumerate}
\end{assumption}
\noindent
This assumption basically requires a positive damping effect and limits excessively fast responses to frequency variations. Classical linear damping, which is often ensured by primary frequency control mechanisms, readily satisfies this assumption as a special case. %加中间这句是为了回应reviewer说的Discuss exceptions or justify how control design ensures compliance.  
% ✅记得数学环境后加noindent

\subsection{Results under Linearized Power Flows}\label{subsec:linearflow}
% a more insightful characterization 这个用语好像学术论文里不常见？
\textcolor{black}{To obtain more insights into the coherence error, we first establish upper bounds using linearized power flow equations while retaining the nonlinear nodal dynamics.} Specifically, we replace \eqref{eq:power flow} with the following DC power flow equations: 
\begin{align}\label{eq:power flow linear}
    p_{e,i} = \sum_{j \in \mathcal{N}_i} B_{ij} (\theta_i - \theta_j), \ \forall i \in \mathcal{N},
\end{align}
which is a standard approximation as the phase angle differences are typically small~\cite{wood2013power}.
% Changed: reasonableto standard %Deleted: of the power network coupling %Changed: in the regime of smallto as the phase

For brevity, define $\theta := [\theta_1,\dots,\theta_N]^T$, 
$\omega := [\omega_1,\dots,\omega_N]^T$, 
$\xi := [\xi_1,\dots,\xi_N]^T$, and $M := \operatorname{diag}(M_1,\dots,M_N)$.
Let $L_B$ denote a weighted Laplacian matrix with entries $(L_B)_{ij} = -B_{ij}$ for $i \ne j$, $(L_B)_{ii} = \sum_{j \ne i} B_{ij}$. Further denote the second-smallest eigenvalue of the matrix $M^{-1}L_B$ as $\lambda_2(M^{-1}L_B)$, or simply $\lambda_2$. 
Then we are ready to present the main results in the following theorem.

% 把定理里两个部分共用的符号（K、lambda_2）拿到外面、前面来。
\begin{theorem}\label{thm: thm1 general}
Let Assumption \ref{assump1} hold. Given  
\begin{equation}\label{eq:constant_K}
    K :=\frac{32 NM_b (1+\frac L \mu)^2 }{(\min_{i\in\mathcal{N}} M_i)  \mu^2},
\end{equation}
the following results hold.
\begin{enumerate}
    \item  If $C := \max_{i \in \mathcal{N}} \sup_{t > 0}|\dot{\xi}_i(t)|/M_i$ is finite, 
% Let $\omega_b(0_+) =  f_b^{-1}(-\xi_b(0_+))$.
% Assume that there exists a constant $C\geq 0$ such that $|\dot \xi_i(t)|/M_i\leq C, \forall t \textcolor{red}{>} 0,\forall i \in \mathcal{N}$. 
% Given any initial conditions $\omega(0)$ and $\theta(0)$ of system \eqref{eq:nonlinear dynamics},
% Then, there exist positive constants $\alpha_1$ and $\alpha_2$ such that for $\forall i \in \mathcal{N}$ and $\forall t \textcolor{red}{>} 0$,
then there exists a positive constant $\alpha$ such that for $\forall t > 0$,
\begin{align}\label{eq:thm1 general conclusion 1}
   \max_{i \in \mathcal{N}}|\omega_i(t) - \omega_b(t) |^2 \leq \alpha  e^{-ct} + K C^2\frac{(\lambda_2+ 4L^2)^2}{ \lambda_2^3},
\end{align}
where
\begin{equation}\label{eq:exponential_rate_linear}
c := \frac{\mu\lambda_2}{4(\lambda_2+ 4L^2)}.
\end{equation}
% in which $M:= \operatorname{diag}(M_1,\dots,M_N)$ and $L_B$ is a weighted Laplacian matrix with entries $(L_B)_{ij} = -B_{ij}$ for $i \ne j$, $(L_B)_{ii} = \sum_{j \ne i} B_{ij}$.
\item If $C_{\lim} := \max_{i \in \mathcal{N}}\limsup_{t \to \infty}|\dot{\xi}_i(t)|/M_i$ is finite, then
% Let $\omega_b(0_+) =  f_b^{-1}(-\xi_b(0_+))$.
% Assume that there exists a constant $C\geq 0$ such that $|\dot \xi_i(t)|/M_i\leq C, \forall t \textcolor{red}{>} 0,\forall i \in \mathcal{N}$. 
% Given any initial conditions $\omega(0)$ and $\theta(0)$ of system \eqref{eq:nonlinear dynamics},
% Then, there exist positive constants $\alpha$ and $\alpha_2$ such that for $\forall i \in \mathcal{N}$ and $\forall t \textcolor{red}{>} 0$,
\begin{align}\label{eq:thm1 general conclusion 2}
  \max_{i \in \mathcal{N}}\limsup_{t \to \infty}|\omega_i(t) - \omega_b(t) |^2 \leq  K C_{\lim}^2\frac{(\lambda_2+ 4L^2)^2}{ \lambda_2^3}.
\end{align}
%where $K$ and $\lambda_2$ are as defined above.
\end{enumerate}
% \[
% \varepsilon(\lambda_2)= \frac{\mu}{8\lambda_2 }\left(L^2 + \sqrt{L^4 + 8 \lambda_2(\frac{\mu}{2}+ L)^2}\right), \ \phi(\lambda_2) = \frac{1}{\mu - 2\varepsilon(\lambda_2)}\sqrt{\frac{1}{8 \lambda_2 + L^2 }}.
% \]
% $\varepsilon= \frac{\mu L^2}{2(\lambda_2+ 4L^2)} < \frac{\mu}{8}$
\end{theorem}
\noindent The proof of Theorem~\ref{thm: thm1 general} is provided in Appendix~\ref{prof: thm1} with the explicit expression for $\alpha$. %✅加上with the explicit expression for $\alpha$

% 一句话概括定理讲了什么，有什么用（读者第一眼关心的应该是定理有啥用，而不是里面每一个符号在说啥）
% （1）公式说了什么：
The first part of this theorem establishes that the coherence error decays exponentially into a bounded region, as shown in \eqref{eq:thm1 general conclusion 1}. 
In particular, by applying the limit superior ($\limsup_{t \to \infty}$) to \eqref{eq:thm1 general conclusion 1}, the first term on the right-hand side vanishes, leaving the second term as a limiting bound. Then the second part of the theorem further refines this limiting bound by replacing the constant $C$ with $C_{\lim}$, as shown in \eqref{eq:thm1 general conclusion 2}, which is less conservative under disturbances with decaying $\dot \xi(t)$.
% 可能要补充的其他信息点：Clim <= C; characterize long-term behavior; 
% while the second part further gives a refined limiting bound in \eqref{eq:thm1 general conclusion 2}. 
% "limiting bound" (最终界): 这是控制理论和动态系统领域的一个更常用、更标准的术语。它特指当时间趋于无穷时，系统轨迹最终会进入并停留在的那个集合的界。例如，在输入-状态稳定 (Input-to-State Stability, ISS) 理论中，这个概念非常核心。使用 "ultimate bound" 能让熟悉这个领域的读者立刻明白你指的是系统的一个长期、稳定的性质，而不是一个单纯的数学极限。
% Comment: Could explain a bit more how the second part is obtained as a special but refined case of the first part: take t to infty and C\_lim\textbackslash{}leq C, etc. It might be very obvious to us, but not to a first-time reader.

% （2）有啥用，主要message和含义是啥：
% 其实是当bound小的时候imply，但是模糊一点应该没事。
In practice, these bounds imply that all nodal frequencies $\omega_i(t)$ quickly approach and then approximately follow the common trajectory $\omega_b(t)$ even when they do not settle to an equilibrium, with explicit characterizations of how fast they approach and how close they eventually remain. 
% 长期跟随应该换一个比subsequently更好的词语表示“长期”. % 这句的主要信息是，哪怕扰动持续时变，也可以近似同步。% 证明wb是准确的近似
Since $|\omega_i - \omega_j| \leq 2 \max_{i \in \mathcal{N}}|\omega_i - \omega_b|$, the bounds also explain why real-world power networks can exhibit approximate frequency synchronization even under time-varying disturbances. 
% 刻画解释现实现象
% 拆成两句是为了更突出结果分两层的含义，工程层面验证了近似效果，理论层面解释了自然现象。

% (3) 逐个解释公式里的每一项
%✅Now we take a closer look at how small the coherence error could be, and investigate the key factors.
We take a step further to analyze the key factors of the coherence error.
First, we examine the limiting bound of the coherence error. 
% To determine whether we should use "limiting bound" or "ultimate bound"
% ---the second term on the right-hand side of \eqref{eq:thm1 general conclusion 1} or the term on the right-hand side of \eqref{eq:thm1 general conclusion 2}. 
Here we can find two regimes where the error can be driven small: %✅can be driven
\begin{itemize}
    \item Small $C$ (or $C_{\lim}$), which means the disturbances change slowly in time. As the time-variation rate of the disturbances decreases, the nodal frequencies are able to follow $\omega_b(t)$ more coherently. In the special cases where the disturbances are constants for all $t>0$ or $t\to \infty$, we have $C= 0$ or $C_{\lim}=0$. Thus the bound reduces to zero and all the nodal frequencies are exactly synchronized to $\omega_b(\infty)$. 
    Such findings validate the intuition that in the case of slow-varying disturbances, the nodal frequencies adjust to each other gently and thus stay close together. 
    Further, recall that the power flow coupling is essentially in the form of integral control. Our results align with the fact that it suppresses low-frequency disturbances but is less responsive to rapid changes. 
    \item Large $\lambda_2$, which indicates high algebraic connectivity of the power network. The ultimate coherence error becomes arbitrarily small when $\lambda_2$ is sufficiently large, even under time-varying disturbances. This formalizes the intuition that stronger interconnection among nodes leads to more coherent behavior. 
\end{itemize}

% In addition, \eqref{eq:thm1 general conclusion 2} refines this limiting bound to be less conservative by taking into account the limiting time-variation rate of the disturbances $C_{\lim}$. It provides a better characterization of the long-term behavior and can be more useful under disturbances with decaying $\dot \xi(t)$. 
% Comment: Move this paragraph up where we explain how we get the second part of the theorem.

% ✅To compare the revised language
Having identified the regimes where the long-term coherence error is small,  we now focus on the exponential rate at which the error decays into that small region in transient, as reflected by the term $\alpha e^{-ct}$ in \eqref{eq:thm1 general conclusion 1}. 
%Notably, the decay is exponential, and 
Accordingly to \eqref{eq:exponential_rate_linear}, the rate $c$ improves with higher connectivity $\lambda_2$. Therefore, high network connectivity not only contributes to diminishing the limiting coherence error in the long run, but also accelerates the error decay process---key indicators for the capability of a power network to accommodate disturbances and maintain frequency synchronization. 
Further, the rate $c$ can get arbitrarily close to $\frac{\mu}{4}$ when $\lambda_2$ is sufficiently large, showing that in a tightly connected power network, the bottleneck actually lies in the nodal damping effect. 
% The fact that local dissipation mechanisms at buses limit overall synchronization performance offers a novel perspective for power system coherency analysis and is not yet widely recognized. 
Our result aligns with the observations in \cite{guo2018graph}, which reveals a similar dependence of the synchronization rate on both connectivity and damping, but relies on the assumption of a uniform damping-to-inertia ratio. These findings may shed light on the control design for power systems with high penetration of renewables that lack natural damping. 
%Since the exponential decay rate is a key indicator for potential power system stability, this result further supports the effectiveness of $\omega_b(t)$ as an approximation for $\omega_i(t)$.

\begin{comment}不再承接推论时候的版本
    While in the above regimes the coherence error is guaranteed to be small in the limit, for arbitrary initial states, it is generally impossible to expect that the error can be small for all $t>0$. Nevertheless, the transient term $\alpha e^{-ct}$ in \eqref{eq:thm1 general conclusion 1} ensures that the error decays rapidly with an exponential rate. The exponential decay rate is a key indicator for potential power system stability. Moreover, the rate $c$ is enhanced with higher connectivity $\lambda_2$ and gets arbitrarily close to $\frac{\mu}{4}$ when $\lambda_2$ is sufficiently large.
\end{comment}

 % Motivate的时候，不要提及任意初值下error能不能for all t>0 无限小（无论是说事实上能不能，还是bound里的alpha_1能不能，都难以严谨又简洁地论述），只说希望单独刻画扰动带来的error，并且它能实现for all t>0 无限小（至于其他初始条件下能不能实现，就不提了）
 % 这段话写的有点重复啰嗦; an state steady determined by the constant 到底是a还是the?  %steady state出现的次数太多了
While Theorem~\ref{thm: thm1 general} characterizes the coherence error under general initial states $(\omega(0),\theta(0))$, an arbitrary initial state could potentially introduce a large error in transient, mixed with the error caused by the disturbances $\xi(t)$, as indicated in the constant $\alpha$. % 其实并非从alpha表达式里一眼能看出来，需要仔细分析，所以说的比较模糊，希望读的时候也不要太纠结这个点。
To distinguish the effects of the disturbances and the role of network connectivity in suppressing them, we exclude the influence of arbitrary initial states by considering the case where the system starts from a steady state. %这句话和下一句steady state重复了，但是没想出来怎么改。不想让这句话太长，所以不能合并。换steady state的近义词显得表达太混乱了。
Specifically, suppose that prior to $t = 0$, the system has settled into a steady state determined by the constant input vector $\xi(0_-)$. This means that all frequencies and all phase angle differences remain constant and thus their time derivatives are zero, leading to the following equations for $(\omega(0),\theta(0))$: % 总觉得句子结构不太顺呢？
% This is formally defined by the conditions that the time derivatives of each frequency $\omega_i$ and each centered angle $\theta_i - (1/N) \sum_{j=1}^N \theta_j$ are zero~\cite{cui2022reinforcement}, 
% Maybe I could refer to the literature for this result.  但可能有点困难，经典文献里一般不formulate成这种nonlinear damping的形式
\begin{subequations}\label{eq:steady state condition}
    \begin{align}
          0 & =f_i(\omega_i(0)) + \xi_i(0_-) - \sum_{j \in \mathcal{N}_i} B_{ij} (\theta_i(0) - \theta_j(0)),\ \forall i \in \mathcal{N},\label{eq:steady 1}\\
  0 &= \omega_i(0) - \omega_j(0),\ \forall \{i,j\} \in \mathcal{E}\label{eq:steady 2}.
    \end{align}
\end{subequations}%虽然这里没有写成inertia-weighted 的形式，跟wb(0)的定义不太一致，但是文献里也有这样写的，况且频率全相等时，加权与否并没有影响。
% Rigorously speaking, by saying a steady state of a system, we are saying that w(0), theta(0) satisfies the equations such that (w(0), theta(0) - (1/N) 1_N^T theta(0)) is a steady state. 
% where the second line implies $\omega_i(0) = \omega_b(0)$ for all $i \in \mathcal{N}$, 不说这句了，虽然能体现这样设置确实让初始coherence error小，但更容易让人开始困惑为什么bound不是从0开始。并且放在这儿无头无尾
A solution \((\omega(0), \theta(0))\) to the equations \eqref{eq:steady state condition} exists, as demonstrated in Appendix~\ref{sec: proof of steady state}. 
For this specific initial state, the constant $\alpha$ in Theorem~\ref{thm: thm1 general} can be replaced with an explicit form $\alpha^*|\Delta \xi|^2$ that has a cleaner dependence on $\lambda_2$ and the disturbance abrupt changes, as presented in the following proposition.  % 强调这个推论与前面的形式上的不同是alpha_1换成alpha_\xi % 
\begin{proposition}\label{coro: start from steady} %✅留意所有的常数都是怎样明确的给出的，越精确越好
Let Assumption~\ref{assump1} hold. Let the constants $c$ and $K$ be given in \eqref{eq:constant_K} and \eqref{eq:exponential_rate_linear}, respectively. Suppose that $(\omega(0),\theta(0))$ is a solution to the equations \eqref{eq:steady state condition}. 
% at $t  = 0$, the system is at the steady state corresponding to the constant disturbances $\xi(0_-)$, 
If $C := \max_{i \in \mathcal{N}} \sup_{t > 0}|\dot{\xi}_i(t)|/M_i$ is finite, then for $\forall t > 0$,
    \begin{align}\label{eq:coro conclusion}
 \max_{i \in \mathcal{N}}  |\omega_i(t) - \omega_b(t) |^2 \leq \alpha^* |\Delta \xi|^2  e^{-ct} + K C^2\frac{(\lambda_2+ 4L^2)^2}{ \lambda_2^3},
\end{align}
with 
\[
\alpha^* := \frac{(\phi_1+ \phi_2)\lambda_2 + 4 \phi_2 L^2}{\lambda_2^2},
\]
where $\Delta \xi:=\xi(0_+) - \xi(0_-)$ and
\[
\phi_1 := \frac{1}{\min_i M_i^2},\ \phi_2 := \frac{16 L^2}{3\mu^2 M_b (\min_i M_i)}  .
\] 
%$c$ and $K$ are as defined in Theorem~\ref{thm: thm1 general}.
\end{proposition}
\noindent The proof of Proposition~\ref{coro: start from steady} is provided in Appendix~\ref{prof: coro1}. 
% 从定理推出来的过程
% We briefly explain how this upper bound \eqref{eq:coro conclusion} is induced from the upper bound \eqref{eq:thm1 general conclusion 1} in Theorem~\ref{thm: thm1 general}. First substitute the conditions \eqref{eq:steady state condition} into the expression of $\alpha_1$ given in the proof of Theorem~\ref{thm: thm1 general} in Appendix~\ref{prof: thm1}. Then incorporate the inequalities
% \[
% \Delta \xi^\top M^{-1}\Delta \xi \leq \frac{|\Delta \xi|^2}{\min_i M_i}, 
% \qquad 
% |\bar\xi(0_+)-\bar\xi(0_-)|^2 \leq \frac{|\Delta \xi|^2}{N}.
% \]
% This leads to an explicit bound $\alpha_\xi$ in place of $\alpha_1$, as claimed in Proposition~\ref{coro: start from steady}.

% 有什么用，含义
 % In this well-defined scenario, any coherence error that emerges for t>0 is attributable solely to the system's response to the \textbf{change in the inputs} ξ(t) from their pre-disturbance values ξ(0−).
 
% This proposition excludes the influence of arbitrary (incoherent) initial states, which could lead to arbitrarily large initial coherence error. 
This proposition explicitly reveals how the disturbance abrupt changes $\Delta \xi$ lead to a transient coherence error, as shown in the term $\alpha^*|\Delta \xi|^2 e^{-c t}$. This transient error can be effectively suppressed as the network connectivity $\lambda_2$ increases, since $\alpha^* \to 0$ as $\lambda_2 \to \infty$. 
% here arises solely from the initial abrupt changes $\Delta \xi$ of the disturbances and is more accurate. % 我其实没有想表达这个bound更tight的意思
% Specifically, this constant is proportional to $|\Delta \xi|^2$ and $\alpha_\xi \to 0$ as $\lambda_2 \to \infty$. 
This result, combined with the observations on the limiting coherence error in Theorem~\ref{thm: thm1 general}, implies that for any given tolerance $\varepsilon>0$, 
\[
 \max_{i \in \mathcal{N}}  |\omega_i(t) - \omega_b(t) |\leq \varepsilon,\ \forall t > 0
\]
holds as long as the initial disturbance jump $|\Delta \xi|$ and its subsequent variation rate $C$ are both sufficiently small relative to $\varepsilon$, or alternatively, the connectivity $\lambda_2$ is sufficiently high relative to $\varepsilon$. 
% 写epsilon这个式子是为了更加强调“可以对所有t>0都充分小”的强结论。我其实更想突出在这个特殊情况下的整个picture，毕竟扰动的影响事实上并非分成两个阶段的，并不是只有initial jump影响暂态，后续的变化就只影响长期的，只不过是bound成这样了。放在一起完整地刻画了disturbance的各种类型的变化带来的影响。
% 不过这样可能确实没有充分突出proposition跟前面的区别所在。

% 整个subsection的结论和takeaway
In summary, this subsection demonstrates that in power networks which are tightly-connected or subject to slowly varying disturbances, the frequency responses $\omega_i(t)$ can be well approximated by the trajectory $\omega_b(t)$ of the blended dynamics, under the linearized power flows. This validates our proposed reduced-order approximation and further provides insights into the high level of coherence observed in real-world networks. 
% In summary, the results in this subsection demonstrate that power networks which are tightly-connected or subject to slowly varying disturbances tend to exhibit a high level of coherence, providing insights into the coherency observed in practice. Also, in such situations, the frequency responses $\omega_i(t)$ can be well approximated by the trajectory $\omega_b(t)$ of the first-order blended dynamics, under the linear power flow approximation.
% 老师提的问题是，前面所有的结果都在讲wi(t) - wb(t)，所以承接上文应该先讲wi(t) well approximated by wb(t)，否则显得突然跳跃了
% 但这样写的难点是，想表达的其实是两条：（1）在强连通、慢扰动条件下，wi(t) well approximated by wb(t); （2）在强连通、慢扰动条件下，有high level of coherence，对现象提供了insights。但是，条件的描述本身比较长，而第（1）条的结论也很长（"the frequency responses $\omega_i(t)$ can be well approximated by the trajectory of the first-order blended dynamics $\omega_b(t)$”，想强调first-order, 而且严谨的写法只能是the trajectory of blended dynamics）。所以，把（1）放在前面会让第一句很长，第二句很短。
% 现在的做法是，用“in such situations”指代，缩短了句子的长度。

\begin{comment}
其他的可能写法：
In summary, the results in this subsection demonstrate that the trajectory $\omega_b(t)$ of the first-order blended dynamics serves as an effective approximation for the frequency responses $\omega_i(t)$'s, when the power networks are tightly-connected or subject to slowly varying disturbances, since these conditions foster a high level of coherence and guarantee small error bounds.

或者：第一句写wb is an effective approximation。第二句再写approximation is accurate under ... conditions.但是总觉得这样把因（条件）和果（良好近似）颠倒了，第一句effective其实没啥含义。
\end{comment}

\begin{remark}
Our results complement the frequency-domain coherency analysis in \cite{min2025frequency}, which connects coherence to the frequency composition of disturbance signals but fails to account for disturbances with abrupt changes. Only finite-time bounds are provided for technical reasons. Their bounds suggest that power networks are naturally coherent under sufficiently low-frequency disturbances, which is aligned with our results given sufficiently small $\dot \xi(t)$. %✅对比修改的
\end{remark}

% 为什么这个remark只讲了limiting bound呢？因为transient bound里，如果初始值差异很大，那即使dynamics同质，也需要一段时间让差异收敛到0；如果初始值也同步，加上dynamics同质，那很自然就会轨迹全都完全相等，这太trivial了没啥好说的。
\begin{remark}
    % the heterogeneity plays a role of non-vanishing perturbation, it can be seen that the ultimate boundedness results are attained
Theorem \ref{thm: thm1 general} also suggests that heterogeneity is key to the non-vanishing perturbation, leading to the long-term limiting coherence error on the right-hand side of \eqref{eq:thm1 general conclusion 1} or \eqref{eq:thm1 general conclusion 2}. 
From the proof in Appendix~\ref{prof: thm1} (particularly the term $d\tilde f/dt$ in \eqref{eq: Vdot relation with V} that is a proxy for heterogeneity),  
% To see this, define $R:= Y^T M^{\frac 1 2}$ and denote the first column of $R$ by $r_1$ so that $R = [r_1,\tilde R]$ with a matrix $\tilde R$. Then it follows from $R1_N = 0$ that $r_i = - \tilde R 1_{N-1}$. Hence
%     \begin{align*}
%     \tilde f &= RM^{-1}(f(1_N\omega_b) + \xi) = [r_1, \tilde R]\begin{bmatrix}
%         \frac{f_1(\omega_b) + \xi_1}{M_1} \\ \vdots \\  \frac{f_N(\omega_b) + \xi_N}{M_N} 
%     \end{bmatrix}\\
%     &= \tilde R\begin{bmatrix}
%         \frac{f_2(\omega_b) + \xi_2}{M_2} -  \frac{f_1(\omega_b) + \xi_1}{M_1}\\
%         \vdots\\
%            \frac{f_N(\omega_b) + \xi_N}{M_N} -  \frac{f_1(\omega_b) + \xi_1}{M_1}
%     \end{bmatrix}.
% \end{align*}
% Roughly speaking, $\tilde f$ measures the differences between nodal dynamics.
if all nodes are homogeneous in the sense that,
for all $i \in \mathcal{N}$, $f_i(\cdot) = M_i f_o(\cdot)$ and $\xi_i = M_i \xi_o$ for some $f_o(\cdot)$ and $\xi_o$, it follows from the definition of $\tilde f$ in \eqref{eq:def of f tilde linear} that $\tilde f  \equiv 0$, indicating zero heterogeneity. In this case, it can be inferred from \eqref{eq: Vdot relation with V} that $\max_{i \in \mathcal{N}}\limsup_{t \to \infty}|\omega_i(t) - \omega_b(t)| = 0$, i.e., the coherence error will eventually vanish.
\end{remark}

\subsection{Results under Nonlinear Power Flows}\label{subsec:nonlinearflow}

%In this subsection, we return to the nonlinear power flow setting to develop upper bounds on the coherence error. While the nonlinear analysis requires additional restrictions on the disturbances and the initial conditions, the insights obtained under linearized power flows hold similarly here. 
While the above insights into the coherence error are derived with the linearized power flow model, we show in this subsection most of the results generalize to the nonlinear power flow setting, under mild additional conditions.

%To simplify analysis and improve the clarity of results, 
% 不太确定能不能算equivalently，还是考虑用弱一点的词？
First, we reformulate the power flow equation \eqref{eq:power flow} as
\begin{align*}
    p_{e,i} &= \sum_{j \in \mathcal{N}_i} k B_{ij}^{0} \sin(\theta_i - \theta_j),
\end{align*}
where $B_{ij}^{0}$ is a baseline line sensitivity and $ k \in \mathbb{R}$ is a uniform scaling factor for all line parameters. This model allows us to analyze the impact of the overall network connection strength on the coherence error via the uniform scaling of a single parameter $k$.
% ✅对比修改的

% 不知道这样写行不行？段首第一句虽然确实点出了关键词“assumption”，但是第一句重心是在说文献，而不是在说我们要做xxx。%或者换成类似这种意思： We impose an additional assumption on the disturbances, which is an adaptation of the classical assumption on power flow feasibility widely used in the literature\addcite. 
%A typical assumption in nonlinear power system analysis is the feasibility of the power flow equations, which is guaranteed when the disturbances are sufficiently small~\cite{weitenberg2018robust}. In a similar spirit, we impose an additional assumption on the disturbances, with a slight adaptation tailored to our analysis.
Our analysis is predicated upon the feasibility of the nonlinear power flow equations, which generally requires that the disturbances cannot be arbitrarily large~\cite{weitenberg2018robust}. 
Denote the second-smallest eigenvalue of $L_B$ as $\lambda_2(L_B)$, or simply $\lambda_2^L$. We then impose the following assumption to bound disturbances in our case.

\begin{assumption}\label{assump:feasible}
$\xi(t)$ is bounded for all $t \geq 0$ and there exists some $\rho \in (0,\frac{\pi}{4})$ such that
  \begin{align*}
  % \label{eq:assum feasible}
   \frac{12 L (\max_{i \in \mathcal{N}} M_i) (\sup_{t  \geq 0}|\xi_b(t)| )}{\mu M_b } + 2\sup_{t \geq 0}|\xi(t) |_{\mathcal{E},\infty} 
   \leq k \lambda_2^{L} \cos(2\rho),
\end{align*}
where $|\xi(t)|_{\mathcal{E},\infty}:= \max_{\{i,j  \} \in\mathcal{E}}|\xi_i(t) - \xi_j(t)|$.
% 原始版本应该是，第一项取max over [0,t]，第二项取t时刻，该不等式对于任意t满足。为了避免写的这么绕，这里直接用了一个更强的版本。
\end{assumption}
\noindent Assumption~\ref{assump:feasible} can be readily satisfied by a more straightforward stronger assumption:
\begin{align*}
   \left (\frac{12 L (\max_{i \in \mathcal{N}} M_i) }{\mu M_b } + 4 \right)|\xi_i(t)| \leq k\lambda_2^{L} \cos(2\rho), \ \forall i \in \mathcal{N},\forall t \geq 0,
\end{align*}
i.e., each nodal disturbance is sufficiently small. Now we present the results under nonlinear power flows in the following theorem.
% 虽然exists such that重复多，但是表达的意思还基本能读懂。
% 其他可能的写法：（1）For any given $\xi(0_+)$, there exists a positive constant $\bar C$ and a non-empty set $\mathcal{X}$ such that whenever the disturbances satisfying ... and $(\omega(0),\theta(0)) \in \mathcal{X}$, the following property holds: 
% （2）结论改成the following property holds: (The equation), where alpha_1 ... are some positive constants (\footnote{The explicit expression of ... is given in ().}
\begin{theorem}
% [ISS under nonlinear power flow]
\label{thm3:nonlinear}
Let Assumption~\ref{assump1} and Assumption~\ref{assump:feasible} hold with some $\rho \in (0,\frac{\pi}{4})$.
Then there exists a positive constant $\bar C$ such that for any disturbances satisfying $C := \sup_{t > 0}\max_{i \in \mathcal{N}}|\dot{\xi}_i(t)|/M_i \leq \bar C$, there exists a non-empty set $\mathcal{X}$ 
% of admissible initial conditions. 
such that when $(\omega(0),\theta(0)) \in \mathcal{X}$, 
% If $(\omega(0),\theta(0)) \in \mathcal{X}$, % 从这里断开分成两句是错误的。后面的结论也需要C <= \bar C的前提.
the following always holds: 

\noindent
- There exist positive constants $\alpha$, $\beta$ and $c$ such that for $\forall t > 0$,
        \begin{align}
            \max_{i \in \mathcal{N}}|\omega_i(t)  - \omega_b(t)|^2 \leq \alpha e^{-ct} + \beta C^2,
        \end{align}
        % where $C := \sup_{t > 0}\max_{i \in \mathcal{N}}|\dot{\xi}_i(t)|/M_i$. 
where $c$ is strictly increasing in $k$ with $\lim_{k\to\infty} c = \Theta(\mu)$ while $\beta$ is strictly decreasing in $k$ with $\lim_{k \to \infty}\beta = 0$.
\end{theorem}
\noindent The proof of Theorem~\ref{thm3:nonlinear} is provided in Appendix~\ref{prof: thm3} with the explicit expressions for $\bar C$, $\mathcal{X}$, $\alpha$, $\beta$ and $c$.

This theorem establishes an upper bound on the coherence error that shares a similar structure to the bound in Theorem~\ref{thm: thm1 general}, which consists of a constant limiting bound $\beta C^2$ and a decaying bound $\alpha e^{-ct}$. In particular, we now analyze the dependence of these bounds on $k$, instead of on $\lambda_2$ as in the linearized power flow model, to investigate the influence of network connectivity. 
% $k$ replaces $\lambda_2$ in the linearized power flow model to indicate the degree of connectivity for nonlinear power networks.
For the limiting bound, the role of network connectivity and disturbance variation rate is both preserved, given the dependence of $\beta$ on $k$.
For the decaying bound, a tightly connected power network still contributes to improving the decaying rate as $c$ increases in $k$, with the bottleneck determined by the nodal damping effect ($\mu$).
Therefore, the key insights indeed generalize here.
%the position of $C$ and the monotonic dependence of $\alpha_2$ and $c$ on $k$ are analogous to that in Theorem~\ref{thm: thm1 general}, thus preserving the key insights on the role of network connectivity and disturbance variation rate in determining the error evolution.
Note that a refined limiting bound can be given similarly to \eqref{eq:thm1 general conclusion 2} in Theorem~\ref{thm: thm1 general}, and is not repeated here for brevity. 
%✅别加footnote ✅对比修改的

% An analogue of Corollary~\ref{coro: start from steady} for Theorem~\ref{thm: thm1 general} can also be established for Theorem~\ref{thm3:nonlinear},
An analogue of Proposition~\ref{coro: start from steady} can also be established under the nonlinear power flows, by considering the case where the system starts from a steady state determined by $\xi(0_-)$. Specifically, the conditions satisfied by $(\omega(0),\theta(0))$ are similar to \eqref{eq:steady state condition}, except that the power flows are replaced with the nonlinear counterpart, given as
\begin{subequations}\label{eq:steady state condition nonlinear}
    \begin{align}
          0 & =f_i(\omega_i(0)) + \xi_i(0_-) -k \sum_{j \in \mathcal{N}_i} B_{ij}^0 \sin(\theta_i(0) - \theta_j(0)),\ \forall i \in \mathcal{N},\label{eq:steady 1 nonlinear}\\
  0 &= \omega_i(0) - \omega_j(0),\ \forall \{i,j\} \in \mathcal{E}\label{eq:steady 2 nonlinear}.
    \end{align}
\end{subequations}
A solution \((\omega(0), \theta(0))\) to the equations \eqref{eq:steady state condition nonlinear} exists under Assumption~\ref{assump:feasible}, as shown in Appendix~\ref{sec: proof of steady state nonlinear}. Under this specific initialization, the constant $\alpha$ can be improved with explicit dependence on the parameter $k$ and the disturbance abrupt change $\Delta \xi$.
% 推论怎样写的更简短呢？我难以找到合适的方式指代定理里的  "Then there exists a positive constant $\bar C$ such that for disturbances satisfying $C := \sup_{t > 0}\max_{i \in \mathcal{N}}|\dot{\xi}_i(t)|/M_i \leq \bar C$, there exists a non-empty set $\mathcal{X}$  of admissible initial conditions. such that whenever $(\omega(0),\theta(0)) \in \mathcal{X}$, "，因为这一段既是条件又是结论。表述方式是“Then 存在这样的条件，使得（后面的结论成立）"，而“存在”又是个结论性的判断。
\begin{proposition}\label{coro: start from steady nonlinear}
Let Assumption~\ref{assump1} and Assumption~\ref{assump:feasible} hold with some $\rho \in (0,\frac{\pi}{4})$. Suppose that $(\omega(0),\theta(0))$ is a solution to the equations \eqref{eq:steady state condition nonlinear}. 
Then there exist positive constants $\bar C$ and $\bar \Delta$ such that for any disturbances satisfying $C := \sup_{t > 0}\max_{i \in \mathcal{N}}|\dot{\xi}_i(t)|/M_i \leq \bar C$ and $|\Delta \xi|:=|\xi(0_+) - \xi(0_-)|\leq \bar \Delta$, 
% there exists a non-empty set $\mathcal{X}$ such that when $(\omega(0),\theta(0)) \in \mathcal{X}$, where $(\omega(0),\theta(0))$ is the initial steady state determined by $\xi(0_-)$, 
the following always holds:

% the following property holds: 
\noindent
- There exist positive constants $\alpha^*$, $\beta$ and $c$ such that for $\forall t > 0$,
        \begin{align}
            \max_{i \in \mathcal{N}}|\omega_i(t)  - \omega_b(t)|^2 \leq \alpha^* |\Delta \xi|^2 e^{-ct} + \beta C^2,
        \end{align}
%where $\Delta \xi:=\xi(0_+) - \xi(0_-)$. 
where $\alpha^*$ is strictly decreasing in $k$ with $\lim_{k \to \infty}\alpha^* =  0$, $c$ is strictly increasing in $k$ with $\lim_{k\to\infty} c = \Theta(\mu)$, and $\beta$ is strictly decreasing in $k$ with $\lim_{k \to \infty}\beta = 0$.
\end{proposition}
\noindent The proof of Proposition~\ref{coro: start from steady nonlinear} is provided in Appendix~\ref{prof: coro2} with the explicit expressions for $\bar C$, $\bar \Delta$, $\alpha^*$, $\beta$ and $c$. 
% The derivation of Corollary~\ref{coro: start from steady nonlinear} is similar to that of Corollary~\ref{coro: start from steady}, except that the steady-state power flow equations are replaced with their nonlinear counterparts. 

In this proposition, the generic constant $\alpha$ is replaced with $\alpha^* |\Delta \xi|^2$, where $\alpha^*$ decreases with $k$ and vanishes as $k \to \infty$. This allows us to confirm that the desirable property
\[
 \max_{i \in \mathcal{N}}  |\omega_i(t) - \omega_b(t) |\leq \varepsilon,\ \forall t > 0
\]
for any tolerance $\varepsilon > 0$ can still be achieved in regimes analogous to those inferred from Proposition~\ref{coro: start from steady}, namely as long as the disturbance jump and its variation are both sufficiently small, or the connectivity (represented by $k$) is sufficiently high.

It should be noted that, unlike the global results obtained under the linear power flow approximation, the bounds in this subsection hold locally. Theorem~\ref{thm3:nonlinear} relies on upper limits on the disturbance magnitude $|\xi_i(t)|$ (in Assumption~\ref{assump:feasible}) and the variation rate $|\dot \xi_i(t)|$, as well as a specified set $\mathcal{X}$ of initial states. 
Specifically, $\mathcal{X}$ defines a neighborhood of the steady-state determined by $\xi(0_+)$, as shown in \eqref{eq: X set def} in the Appendix. This enables us to transform the restriction on initial states to upper limits on disturbance initial jumps $|\Delta \xi|$ in Proposition~\ref{coro: start from steady nonlinear}, as the initialization is specified by $\xi(0_-)$. In reality, since power grids are engineered to operate closely around the nominal frequency and disturbances are typically small relative to the overall system capacity, these local results remain highly relevant for practical operation.

% 待写remark: k越大，允许的disturbance的幅度越大。

% 待写remark: 关于角度theta_i的limiting behavior刻画

\begin{figure*}[t]
  \centering
  % 第一行
   % \hspace{0.1 em}
  \begin{subfigure}[t]{0.310\textwidth}
    \includegraphics[width=\linewidth]{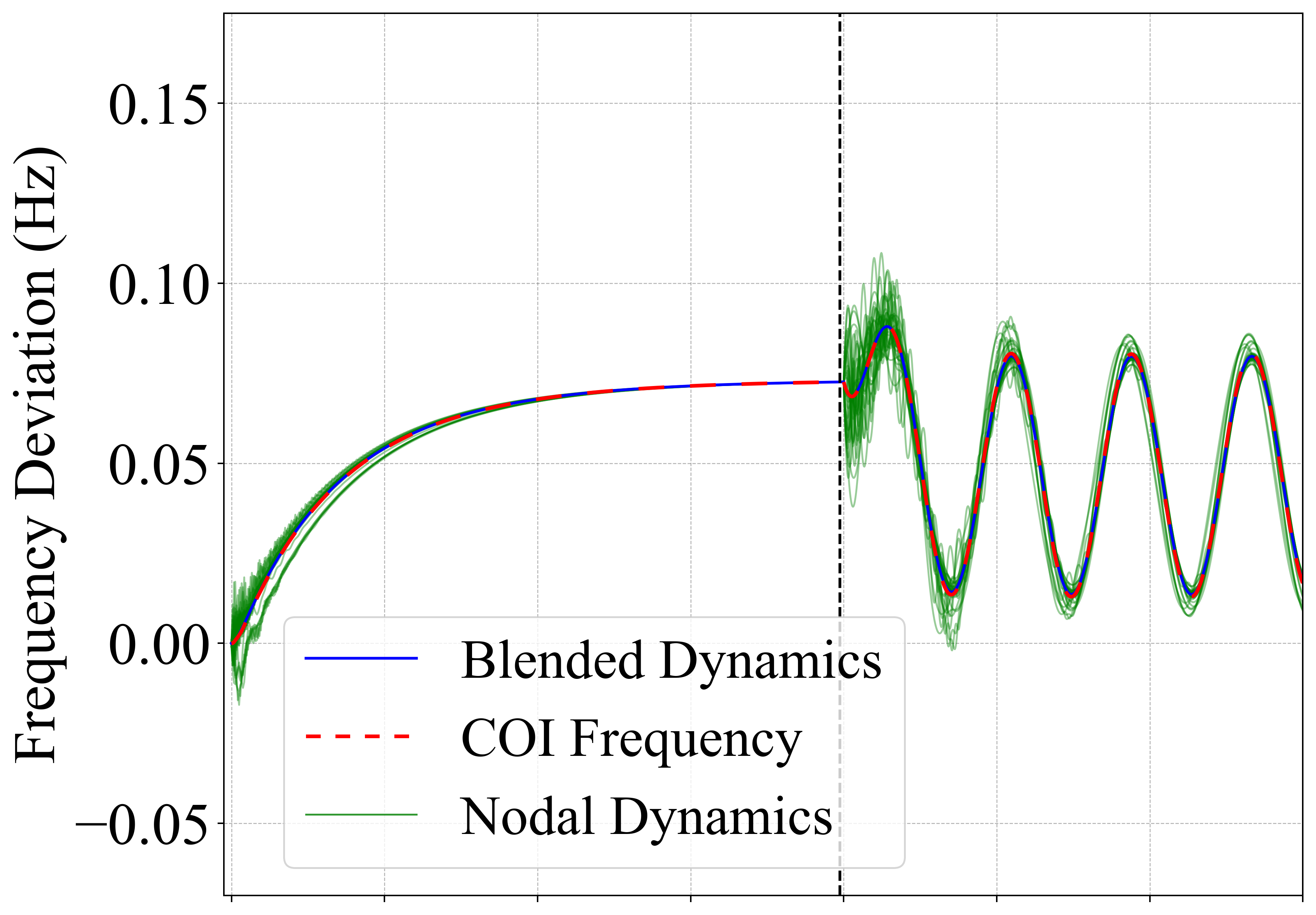}
    % \label{fig:sub1}
  \end{subfigure}
  % \hfill
  \hspace{1.9 em}
  \begin{subfigure}[t]{0.295\textwidth}
    \includegraphics[width=\linewidth]{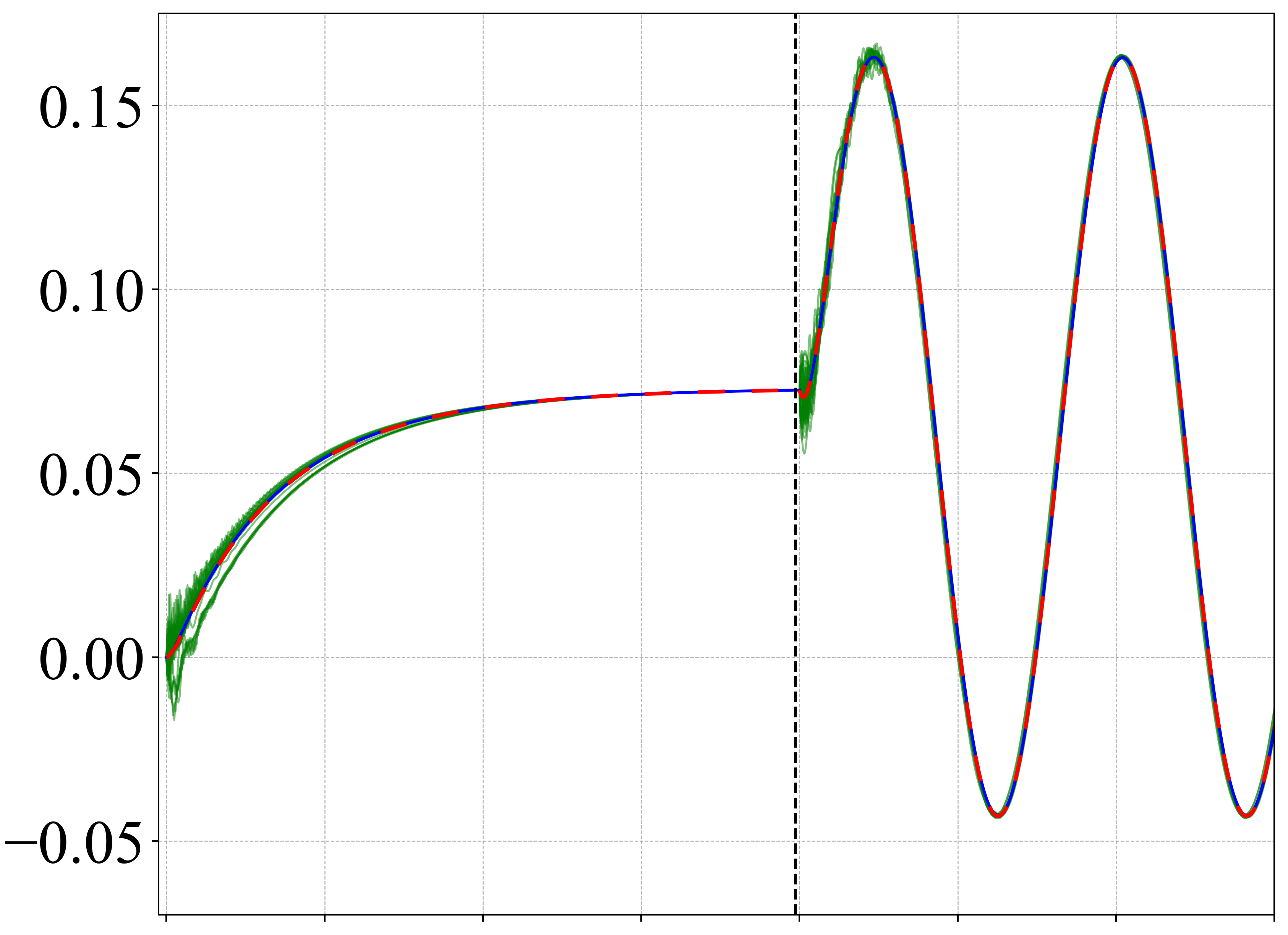}
    % \label{fig:sub2}
  \end{subfigure}
  % \hfill 
  \hspace{1.9 em}
  \begin{subfigure}[t]{0.3\textwidth}
    \includegraphics[width=\linewidth]{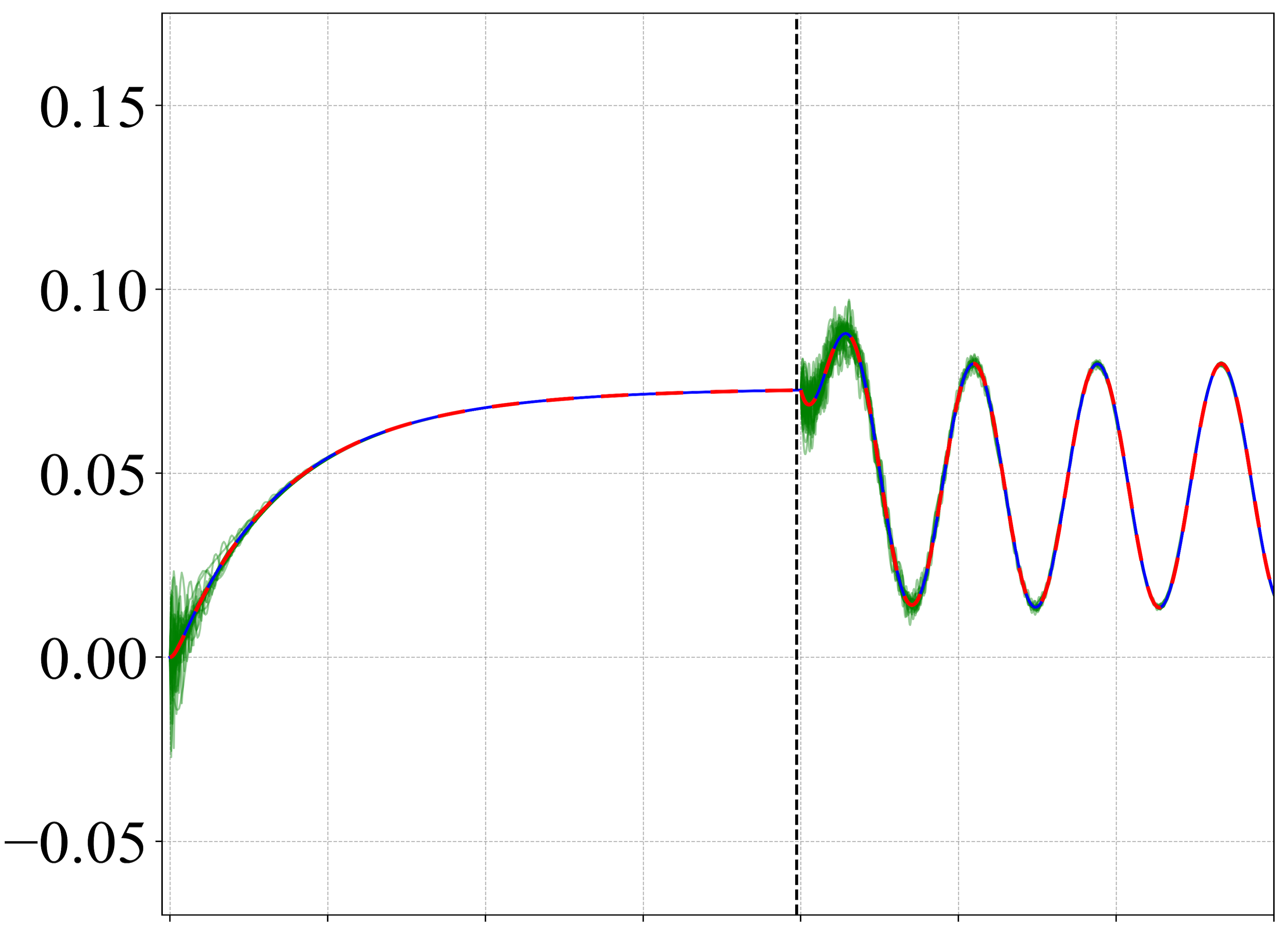}
    % \label{fig:sub3}
  \end{subfigure}
  \hspace{5.0 em}
  % 换行
  \\
  \vspace{-0.3em} % 行间距可调
  % 第二行
  \begin{subfigure}[t]{0.32\textwidth}
    \includegraphics[width=\linewidth]{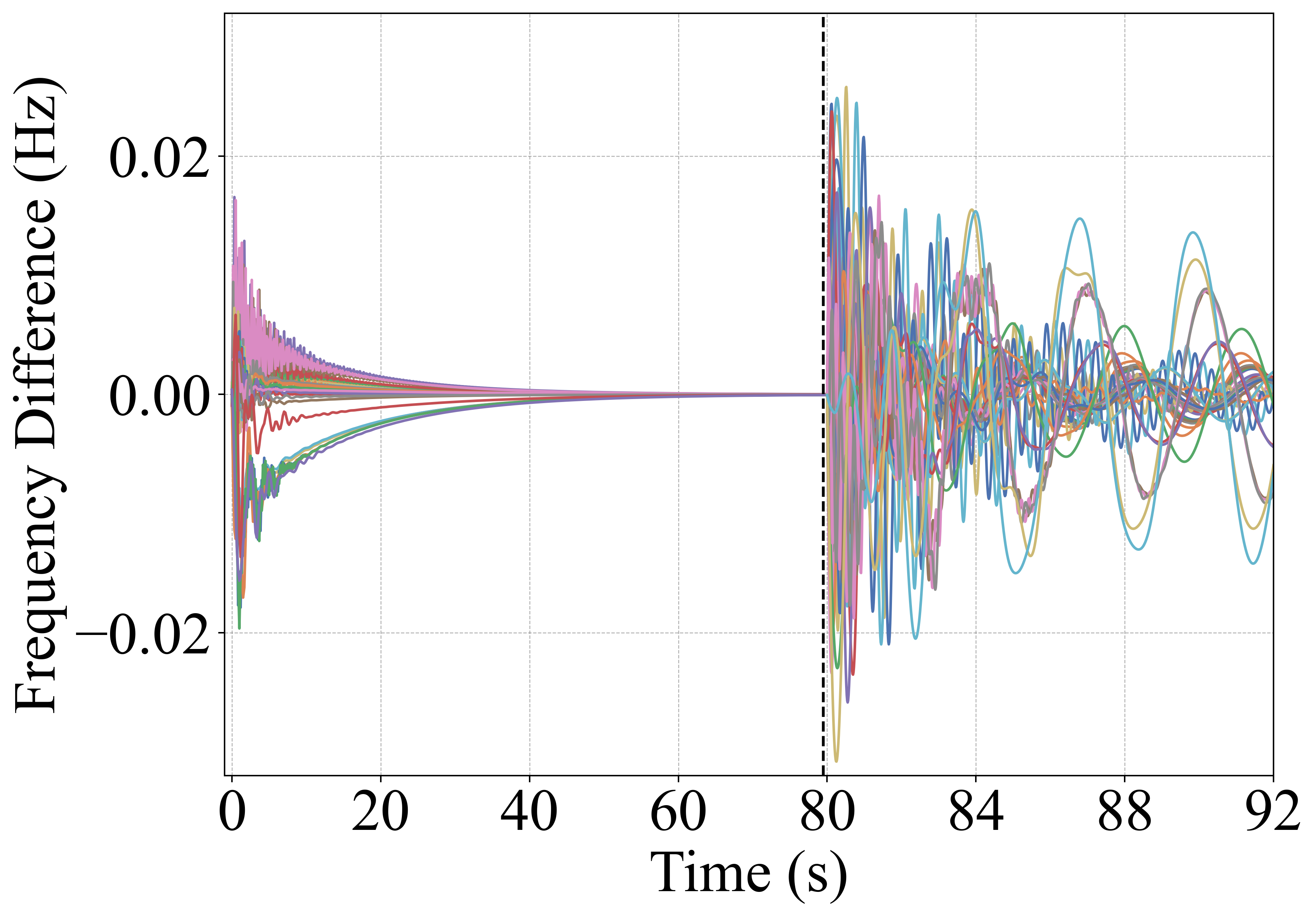}
    \caption{Case 1: Baseline disturbances $\xi_i(t)$; baseline connectivity ($B_{ij}$'s).}
    \label{fig:sub1}
  \end{subfigure}
  \hfill
  \begin{subfigure}[t]{0.307\textwidth}
    \includegraphics[width=\linewidth]{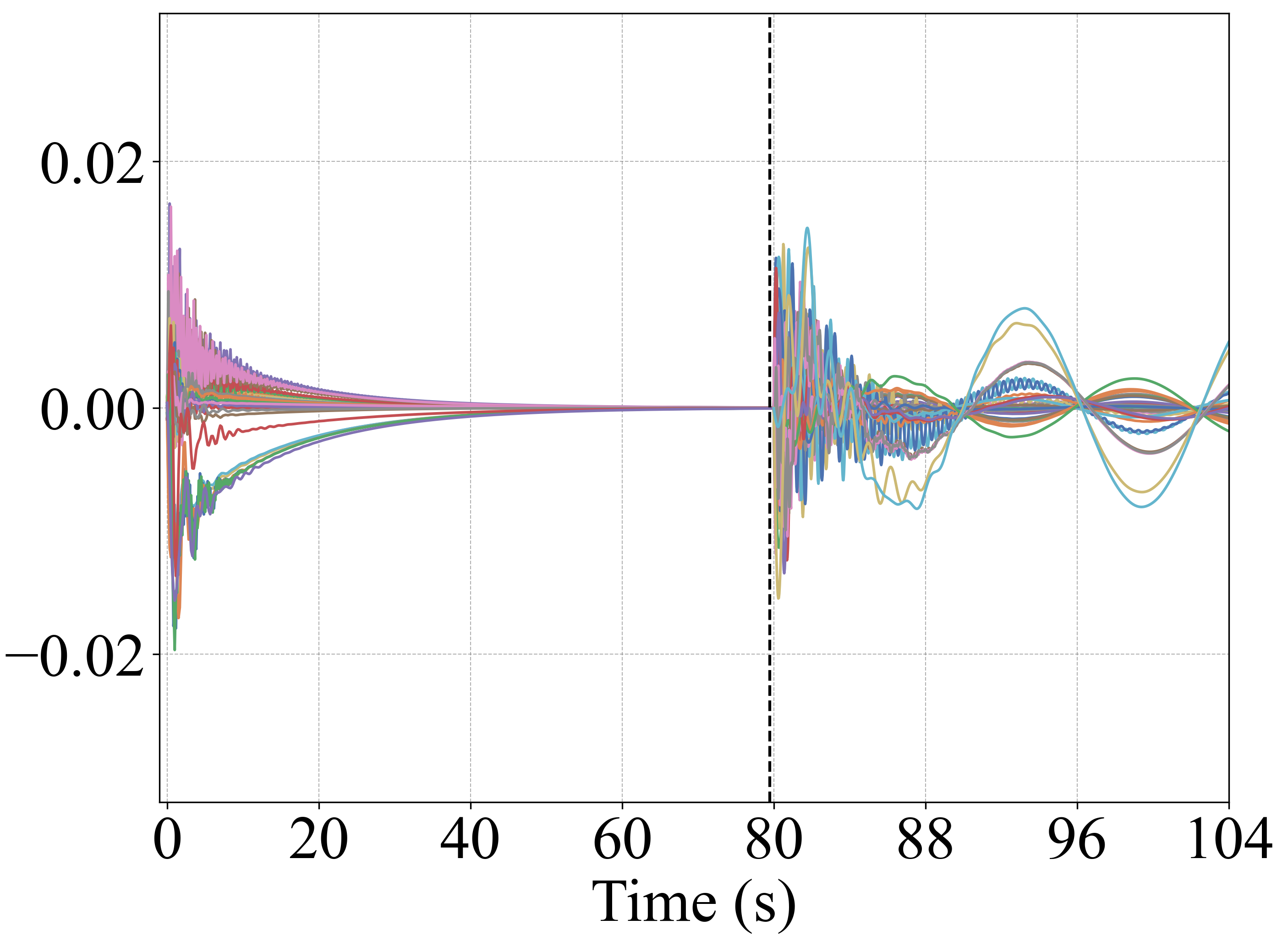}
    \caption{Case 2: Modified disturbances $\xi_i(t)$: $\Delta_i$~reduced twofold, $\Omega$ reduced fourfold and $b_i$ doubled; baseline connectivity ($B_{ij}$'s).}
    \label{fig:sub2}
  \end{subfigure}
  \hfill
  \begin{subfigure}[t]{0.307\textwidth}
    \includegraphics[width=\linewidth]{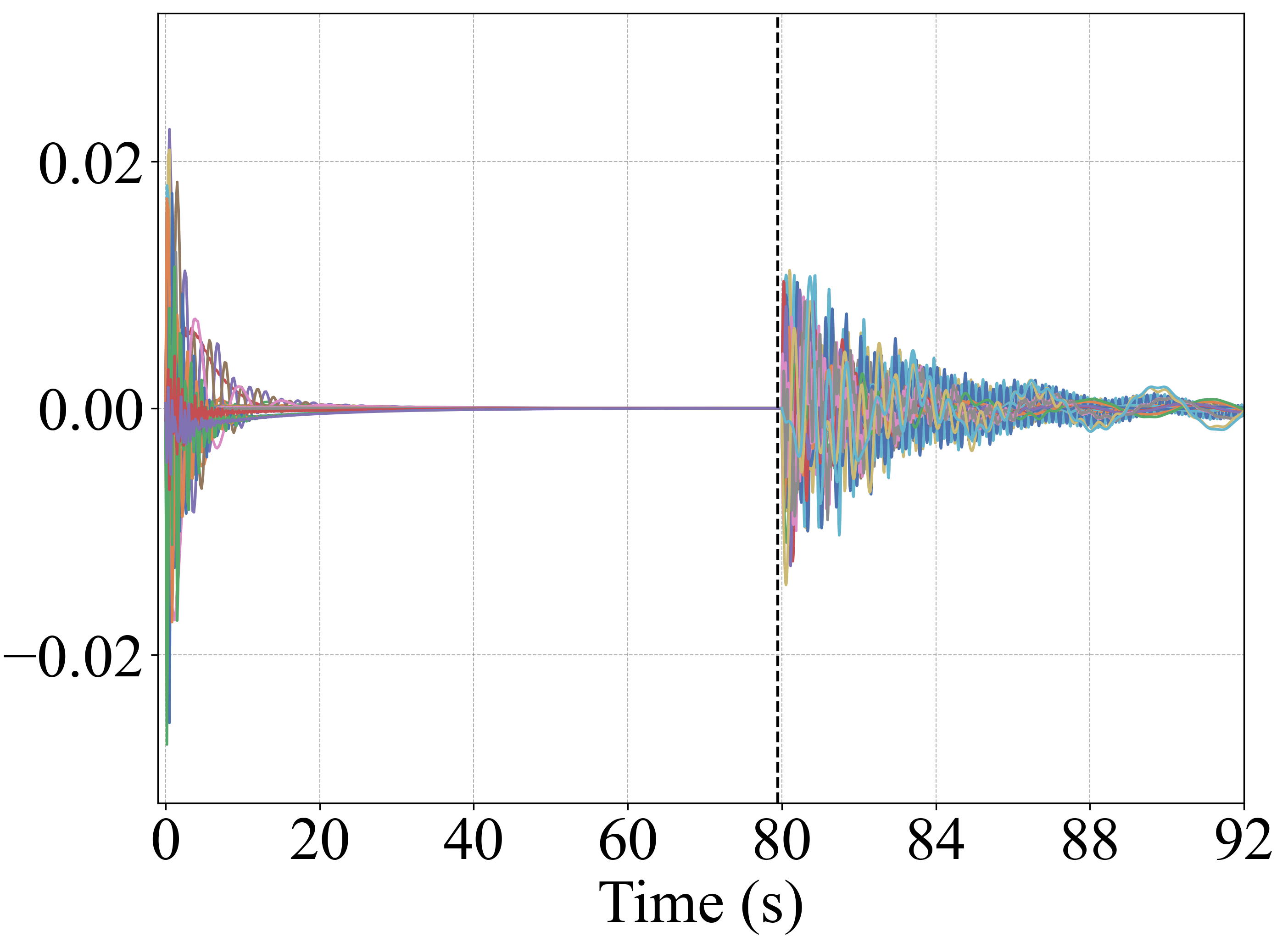}
    \caption{Case 3: Baseline disturbances $\xi_i(t)$; increased connectivity ($6 B_{ij}$'s).}
    \label{fig:sub3}
  \end{subfigure}

  \caption{Frequency responses of Icelandic power grid in three cases, where the network connectivity or the second-stage disturbances are set differently. 
  Upper row: frequency trajectories of all nodal dynamics, the blended dynamics and the COI.
Lower row: trajectories of $\omega_i(t)-\omega_b(t)$ where each line corresponds to a single node $i$. The time axis is non-uniformly spaced before and after $t=80$s for illustration purposes.}
  \label{fig:overall}
\end{figure*}
\section{Simulations}\label{sec:simulations}
In this section, we verify the theoretical analyses by numerical simulations on the Icelandic power grid~\cite{edinburgh_power_systems}. The dynamic model of the grid consists of $118$ nodes, $206$ branches and $35$ generators with heterogeneous parameters. Since our analysis focuses on the generator dynamics, we apply a Kron reduction to the network model to eliminate all non-generator nodes. %数据集里显示有189个bus的信息，但官网的数据集描述又是写的$118$ nodes，不太懂具体怎么回事，所以沿用了$118$ nodes（nodes也是官网的用词）
% where $35$ generators with heterogeneous parameters are connected to the grid. 
The parameters, including the inertia constants $M_i$, the damping coefficients $D_i$, the network topology, and the line sensitivity coefficients $B_{ij}$, are taken from~\cite{edinburgh_power_systems} following the parameter extraction procedure in~\cite{min2025frequency}. For the nonlinear nodal dynamics in \eqref{eq:agent}, we consider
\[
f_i(\omega_i) = -D_i \omega_i - 0.2D_i \tanh(\omega_i), 
\]
which incorporates a potential saturation effect in the frequency response in addition to the linear damping $D_i\omega_i$. All simulations use the nonlinear power flow equations \eqref{eq:power flow}.
% which, besides the linear damping term $D_i\omega_i$, introduces potential saturation effects into the frequency response. 
% where the hyperbolic tangent function introduces smooth saturation effects into the frequency response. 

We consider three types of disturbances by setting the following two-stage disturbance profile:
$$
\xi_i(t)=
\begin{cases}
a_i\bigl(1-e^{-r_i t}\bigr),  \ t\in[0,80),&\text{(Stage 1)}
\\
a_i + \Delta_i\, \mathbf{1}_{t\ge 80} + b_i \sin(\Omega_i (t-80)), \ t \geq 80,&\text{(Stage 2)}
\end{cases}
$$
where in the second stage, $\mathbf{1}_{t \geq 80}$ denotes a step change at $t = 80$s, followed by persistent sinusoidal oscillations. Here the parameters are randomly sampled from the uniform distributions. Specifically, $a_i\sim \mathcal{U}(-0.4,0.4)$, $r_i\sim \mathcal{U}(0.05,0.1)$, $\Delta_i\sim \mathcal{U}(-0.04,0.04)$, $b_i\sim \mathcal{U}(0,0.02)$, and $\Omega_i=2.0$. % 这里2.0应该默认是弧度而不是角度吧，还需要特别说明吗？需要写单位$\Omega_i = 2.0\,\mathrm{rad/s}$吗？但是其他的参数我也不专门写单位呀。
They are selected to emulate realistic heterogeneity in nodal disturbances while maintaining the visual clarity of simulation results. The initial states for the first stage are $\omega_i(0),\theta_i(0)\sim \mathcal{U}(-0.001,0.001)$. The second stage begins after the disturbances have become almost constant and the frequencies have settled, thereby approximating a steady-state initialization.

Fig.~\ref{fig:overall} displays the frequency response of the power network model in three different cases, each shown in a separate column. In each case, the upper row shows the frequency trajectories $\omega_i(t)$ of all generators' nodal dynamics, the trajectory $\omega_b(t)$ of the blended dynamics and the trajectory of the COI frequency, defined as $\omega_{\text{COI}} $$= (\sum_{i=1}^{N} M_i \omega_i)/(\sum_{i=1}^{N} M_i)$. The lower row shows the differences $\omega_i(t) - \omega_b(t)$, which reflects the coherence error. By comparing the three cases, we validate how the disturbance properties and network connectivity influence the level of coherence, and how well $\omega_b(t)$ approximates $\omega_i(t)$.

% where each figure includes the frequency trajectories $\omega_i(t)$ of all generators' nodal dynamics, the trajectory $\omega_b(t)$ of the blended dynamics and the trajectory of the COI frequency $w_{\text{COI}} := (\sum_{i=1}^{N} M_i \omega_i)/(\sum_{i=1}^{N} M_i)$. 
% % The right column shows the differences $\omega_i(t) - w_b(t)$, which reflects the coherence error.
% % First, in all cases, all $\omega_i(t)$'s stay close to $w_b(t)$, validating the blended dynamics as a good characterization of the coherent response. Besides, the COI trajectory, typically used for frequency response assessment, is also well approximated by $w_b(t)$.
% By comparing the three cases, we validate how the disturbance properties and network connectivity influence the level of coherence, and how well $\omega_b(t)$ approximates $\omega_i(t)$.

\noindent
\textbf{Case 1 (Fig.~(\ref{fig:sub1}))}: The nodal responses are already coherent due to the naturally high connectivity of Icelandic grid. All the nodal frequency trajectories are close to the blended-dynamics trajectory. In the first stage where $\xi_i(t) = a_i(1-e^{-r_it})$, all the frequencies eventually achieve exact synchronization with vanishing coherence error, which is consistent with the bound \eqref{eq:thm1 general conclusion 2} in Theorem~\ref{thm: thm1 general} when $C_{\lim} = 0$. This synchronized state is disrupted at $t=80$s by the disturbance jumps $\Delta_i$, which temporarily drive the nodal frequencies away from the blended dynamics and induce a large transient coherence error. This error decays with time but does not vanish eventually due to the persistent oscillations in the disturbances.  
% 第二阶段没有再说consistent with the bound，因为“与定理一致”这个表述在case1-case3里出现很多次了，不想显得很啰嗦；第二次出现blended-dynamics trajectory的时候省略了trajectory，感觉一直重复太啰嗦了。

\noindent
\textbf{Case 2 (Fig.~(\ref{fig:sub2}))}: We modify the second-stage disturbances to demonstrate their influence on coherence. For the abrupt changes, we reduce $\Delta_i$'s by half. For the sinusoidal disturbances, we double their magnitudes $b_i$'s while reducing their frequency $\Omega$ by a factor of four. 
% by reducing both the abrupt changes $\Delta_i$'s and the sinusoidal frequency $\Omega$ by a factor of four, while doubling the sinusoidal magnitudes $b_i$, thereby decreasing the time derivative of the disturbances. 
It can be observed that although all the nodal responses have a larger oscillation magnitude compared with Case 1, their entire trajectories become more closely aligned, and the coherence error becomes smaller for all $t > 80$s---due to reduced $\Delta_i$'s initially and reduced $\dot \xi(t)$ after a while. This verifies the bound in Proposition~\ref{coro: start from steady nonlinear} and further highlights that the level of coherence is more sensitive to the time-variation rate than the magnitude of the disturbances.

\noindent
\textbf{Case 3 (Fig.~(\ref{fig:sub3}))}: We show the effect of higher network connectivity by scaling up all the edge weights by a factor of six as compared with Case 1. In the first stage, the decay rate of the coherence error becomes significantly faster, 
% the rate at which the nodal frequencies approach the blended-dynamics trajectory becomes much faster, 
aligned with Theorem~\ref{thm: thm1 general} for general initial conditions. In the second stage, the coherence error is reduced 
% the overall trajectories are closer 
over the entire time period, as also suggested by Proposition~\ref{coro: start from steady nonlinear}, and the frequency synchronization of all the nodes is remarkably regular.
% 问题：Case 3好像上图和下图只需要refer其中任意一个就可以了，如果反复同时说两层含义，显得有点啰嗦没必要。但是这样会显得图片多余吗？其实是，有些cases需要上图更清楚，有些cases需要下图更清楚，所以组合起来成了2*3.

Finally, these results validate that the blended dynamics is a good approximation for the frequency responses of the full system. Moreover, as a simple first-order dynamics, it also closely approximates the COI trajectory, an indicator typically used for frequency response assessment.

% \begin{figure*}[t]
%   \centering
%   % 第一张图
%   \begin{subfigure}[t]{0.32\textwidth}
%     \includegraphics[width=\linewidth]{figures/simu_fig_low_conn_high_freq.png}
%     \caption{Case 1: Baseline disturbances $\xi_i(t)$; baseline connectivity ($B_{ij}$'s).}
%     \label{fig:sub1}
%   \end{subfigure}
%   \hfill
%   % 第二张图
%   \begin{subfigure}[t]{0.32\textwidth}
%     \includegraphics[width=\linewidth]{figures/simu_low_conn_low_freq_low_step(2).png}
%     \caption{Case 2: Modified disturbances $\xi_i(t)$ with $\Delta_i$'s and $\Omega$ reduced fourfold, and $b_i$'s doubled; baseline connectivity ($B_{ij}$'s).}
%     \label{fig:sub2}
%   \end{subfigure}
%   \hfill
%   % 第三张图
%   \begin{subfigure}[t]{0.32\textwidth}
%     \includegraphics[width=\linewidth]{figures/simu_fig_high_conn_high_freq(2).png}
%     \caption{Case 3: Baseline disturbances $\xi_i(t)$; increased connectivity ($6 B_{ij}$'s)}
%     \label{fig:sub3}
%   \end{subfigure}

%   \caption{Frequency response of Icelandic power grid in three cases where the network connectivity or the disturbances in the second stage are set differently. 
%   The blue solid line is the trajectory $w_b(t)$ obtained from the blended dynamics; the green solid lines are the nodal trajectories $\omega_i(t)$, with each line corresponding to a single generator; the red dashed line is the trajectory of Center-of-Inertia frequency.}
%   \label{fig:overall}
% \end{figure*}

\section{Conclusion}\label{sec:conclusion}
%\textcolor{olive}{(To be shortened!)}
In this paper, we develop a time-domain analysis for the coherent behavior of swing dynamics in heterogeneous nonlinear power networks subject to persistent time-varying disturbances. By extending the blended dynamics approach, we approximate the nodal frequency responses of a coherent power system by a specific trajectory governed by the weighted average of (possibly nonlinear) nodal dynamics, which highlights how heterogeneous individual nodes jointly shape their collective behavior. 
% contribute to the overall network behavior.
% highlights how individual nodes jointly shape the collective behavior. 
% captures the collective behavior in a simple and concise form.

% ❓conclusion一般可以分段吗？实在是太长了，不确定哪些信息可以删？
We analyze the differences between nodal frequency trajectories and this representative trajectory by establishing explicit upper bounds on the coherence error. 
Specifically, we identify two key factors governing coherence. On the one hand, slow time-variation rates of disturbances, including both abrupt jumps and smooth changes, are shown to be crucial for maintaining a smaller coherence error. On the other hand, high network connectivity, as a powerful synchronizing force, simultaneously accelerates the transient decay of the error and reduces its long-term limit. Moreover, for a system perturbed from a steady state, either of these factors is sufficient to ensure the error remains small for all $t>0$. 
% 已经尽力提炼了，但还是有点长，不确定哪些信息可以进一步精简。
% 关于扰动的因素只提了a smaller coherence error，连通性展开分了transient和long-term，不知道会不会让人产生疑惑。（abrupt jumps的影响其实也会进入暂态，但在Theorem1里没有具体讲。不希望把扰动的影响局限在long-term）（maintaining这个词会不会有点太偏人为控制的意味了）
% We reveal how they jointly shape the transient evolution and the long-term limit of the coherence error, for both arbitrary initialization and steady-state initialization. 
% We derive the limiting bound of the error and show how it is determined by both the time-variation rate of disturbances and network connectivity. We further provide a transient upper bound on how quickly the error approaches its limiting bound and reveal an interesting trade-off: higher connectivity accelerates the convergence rate but may also amplify the error in the early stage.
Similar results are observed under both the linearized and nonlinear power flow models---the former is instrumental in analytical insights while the latter incorporates more practical considerations. 
%shares similar structures 
%except that the latter is restricted to local properties. 
These findings offer a novel perspective of the mechanisms that underpin power system coherency and provide useful guidelines for further control design.
\begin{acks}
This work was supported by NSFC through grant 72431001.
\end{acks}

\bibliographystyle{ACM-Reference-Format}
\bibliography{main}

\appendix
\section{Proof of Theorem~\ref{thm: thm1 general}}\label{prof: thm1}
% \begin{align}
%     \dot \theta &= \omega,\\
%     \dot \omega &= M^{-1}(f(\omega) + \xi - L_B\theta).
% \end{align}

% $$\dot \omega_b = M_b^{-1}(f_b(\omega_b) + \xi_b).
% $$

% Define $\delta_\omega := \omega - \mathbb{1}_N \omega_b$, $\tilde \theta = Y^T M^{\frac 1 2}\theta$, 
% \[  
% \delta_\theta := \tilde \theta - \Lambda_P^{-1}Y^T M^{-\frac 1 2} (f(\mathbb{1}_N \omega_b)+ \xi),
% \]
% with $\Lambda_P := Y^T M^{-\frac 1 2}L_B M^{-\frac 1 2} Y$.

% Now we calculate the time derivative of $\delta_\omega$ and $\delta_\theta$ directly.
% \begin{align}
%     M \dot \delta_\omega &= f(\omega)  + \xi - L_B \theta- M_b^{-1} M \mathbb{1}_N (f_b(\omega_b) + \xi_b)
% \end{align}

% Let $M = I$. 
% $f(\omega) + \xi = \mathbb{1}_N \mathbb{1}_N^T (f(\omega) + \xi) + Y Y^T (f(\omega) + \xi)$
% $\delta_\theta := \tilde \theta - Y^TL_B^\dagger Y Y^T (f(\mathbb{1}_N \omega_b) + \xi)$

The proof consists of two major parts: First, we conduct some coordinate transformations to reformulate the system dynamics; Second, we construct a Lyapunov function in the transformed coordinates to establish the decay properties of the coherence error.

% Then we can rewrite the swing dynamics \eqref{eq:power flow} more compactly by substituting \eqref{eq:power flow linear} into \eqref{eq:agent}: \begin{subequations}\label{eq:compact form}
%         \begin{align}
%            \dot \theta &= \omega,\\
%           \dot  \omega & = M^{-1}(f(\omega) + \xi - L_B\theta),\label{eq: compact form wdot}
%         \end{align}
%    \end{subequations}
% where $f(\omega) := [f_1(\omega_1),\dots,f_N(\omega_N)]^T$ is a vector-valued function. 

For convenience of notations, substitute \eqref{eq:power flow linear} into \eqref{eq:agent} and rewrite the swing dynamics \eqref{eq:nonlinear dynamics} more compactly as \begin{subequations}\label{eq:compact form}
        \begin{align}
           \dot \theta &= \omega,\\
          \dot  \omega & = M^{-1}(f(\omega) + \xi - L_B\theta),\label{eq: compact form wdot}
        \end{align}
   \end{subequations}
% where $\theta := [\theta_1,\dots,\theta_N]^T, w := [w_1,\dots,w_N]^T,\xi := [\xi_1,\dots,\xi_N]^T, M:= \operatorname{diag}(M_1,\dots,M_N), f(\omega):= [f_1(\omega_1),\dots,f_N(\omega_N)]^T$.
where $f(\omega) := [f_1(\omega_1),\dots,f_N(\omega_N)]^T$ is a vector-valued function. 

\subsection{Coordinate Transformation}
% \noindent \textbf{Part 1: Coordinate transformation}
We begin with two steps of linear coordinate transformation to the system \eqref{eq:compact form}. 
% In the first step of transformation, we decompose the overall system into a global component and a relative component.
\begin{comment}
Define 
\[
P_e := M^{-1/2} Y Y^T M^{1/2},
\]
\[
P_c := \mathbb{1}_N \mathbb{1}_N^T \frac{M }{NM_b} 
\]
with $P_e + P_c = I_N$. 
\end{comment}

First, since the power flow term $L_B \theta$ in \eqref{eq: compact form wdot} depends only on the phase angle differences, we make the following change of coordinates to separate the (weighted) average component of the angles from the disagreement component: 
\begin{align*}
    \bar \theta :=  \mathbb{1}_N^T \frac{M}{NM_b} \theta,\quad \
             \tilde \theta := 
        Y^T M^{\frac 1 2} \theta,
\end{align*}
where $Y\in \mathbb{R}^{N \times (N-1)}$ is chosen such that the columns of $Y$ form an orthonormal basis of the null space of $\mathbb{1}_N^T M^{\frac{1}{2}}$. This transformation can be written compactly as
\begin{align}\label{eq:coordinate transform def}
    \begin{bmatrix} \bar\theta \\ \tilde\theta \end{bmatrix} = \underbrace{ \begin{bmatrix} \mathbb{1}_N^T\frac{ M}{NM_b} \\ Y^T M^{\frac 1 2} \end{bmatrix}}_{:=P} \theta.
\end{align}
% ⚠️本来好像也可以只写这个公式，不写上一行的公式。但那样似乎无法清晰的看到theta bar和theta tilde的定义是什么？
 % \begin{align}\label{eq:new var}
 %    \begin{bmatrix}
 %             \bar \theta\\
 %             \tilde \theta
 %         \end{bmatrix}
 %    := \begin{bmatrix}
 %       \mathbb{1}_N^T \frac{M}{NM_b} \\
 %        Y^T M^{\frac 1 2}
 %    \end{bmatrix} \theta,
 %        \end{align}
Then the original angle variables can be recovered by the inverse transformation:
% decomposed as 
% \begin{align*}
%     \theta = \mathbb{1}_N \bar \theta + M^{-\frac 1 2}Y\tilde \theta.
% \end{align*}
\begin{align}\label{eq:coordinate transform def inverse}
    \theta = \underbrace{\begin{bmatrix}
        \mathbb{1}_N & M^{-\frac 1 2}Y
    \end{bmatrix}}_{:= Q}\begin{bmatrix} \bar\theta \\ \tilde\theta \end{bmatrix},
\end{align}
where the fact that $PQ = I_N$ follows from $\mathbb{1}_N ^T M\mathbb{1}_N  = NM_b$, $Y^T Y = I_{N-1}$, and $Y^T M^{\frac{1}{2}} \mathbb{1}_N = 0$.
Under this transformation, the power flow term $L_B \theta = L_B M^{-\frac 1 2}Y \tilde \theta$ does not depend on $\bar \theta$. 
% Here $\bar \theta$ is the weighted average of nodal phase angles and does not affect the frequency evolution in \eqref{eq: compact form wdot} since $L_B \mathbb{1}_N = 0$. 
Therefore, the system dynamics \eqref{eq:compact form} can be rewritten as
\begin{subequations}\label{eq: transformed system 1}
    \begin{align}
    \dot \omega &= M^{-1}(f(\omega) + \xi - L_B M^{-\frac 1 2}Y \tilde \theta),
    \label{eq: omega dot after trans}\\
    \dot{\tilde \theta} & = Y^T M^{\frac 1 2}\omega, \label{eq:tilde theta dot}
    \end{align}
\end{subequations}
with the dynamics of $\bar \theta$ omitted.

% where \eqref{eq: bar omega} is obtained with ${1}_N^T L_B = 0$.  \eqref{eq: tilde omega after trans} is obtained with $L_B \mathbb{1}_N =0$.  \eqref{eq:tilde theta} is obtained with $Y^TM^{\frac 1 2}{1}_N = 0$ and $ Y^TY = I_{N-1}$. Since the frequency evolution does not depend on the state $\bar \theta$, we omit $\bar \theta$ here. In addition, it can be checked that $\Lambda_P$ is positive definite and the minimal singular value of $\Lambda_P$ is equal to the second smallest eigenvalue of $M^{-\frac 1 2}L_BM^{-\frac 1 2}$, i.e., $\sigma_m(\Lambda_P) =  \lambda_2(M^{-\frac 1 2}L_BM^{-\frac 1 2})=\lambda_2(M^{-1}L_B)$, which is positive since the network is connected.

The second step of the coordinate transformation is to define the error system, which measures the distance between the state $\omega$, $\tilde \theta$ and their anticipated limiting behavior, respectively. 
Intuitively, if these states converge, we would expect that: 
(1) $\omega$ approximately converges to $\mathbb{1}_N \omega_b$; 
(2) $\tilde \theta$ approximately converges to some $\tilde \theta^*$ which lets the right-hand side of \eqref{eq: omega dot after trans} coincide with $\mathbb{1}_N \dot \omega_b$, i.e.  
\begin{equation}\label{eq: theta* 1}
    \begin{aligned}
         M^{-1}(f(\mathbb{1}_N \omega_b)+ \xi- L_B M^{-\frac 1 2}Y \tilde \theta^*)&= \mathbb{1}_N\dot \omega_b\\ 
            &= \frac{\mathbb{1}_N \mathbb{1}_N^T}{NM_b} (f(\mathbb{1}_N \omega_b)+ \xi ).
    \end{aligned}
\end{equation}
To solve for $\tilde \theta^*$ from \eqref{eq: theta* 1}, we left-multiply this equation by the transformation matrix $P$ to address the average and the disagreement component respectively. First, after left-multiplying \eqref{eq: theta* 1} by $\mathbb{1}_N^T M$ (omitting $1/(NM_b)$), the left-hand side becomes $\mathbb{1}_N^T(f(\mathbb{1}_N \omega_b)+ \xi)$ using $\mathbb{1}_N^T L_B = 0$, which is always identical to the right-hand side, since $\mathbb{1}_N^T  M \mathbb{1}_N \mathbb{1}_N^T= N M_b  \mathbb{1}_N^T$. So it remains to solve \eqref{eq: theta* 1} by left-multiplying with $Y^T M^{\frac{1}{2}}$, which yields
\begin{equation*}
    \begin{aligned}
     % \underbrace{   Y^T M^{-\frac 1 2}(f(\mathbb{1}_N \omega_b)+ \xi)}_{:=\tilde f}-   \underbrace{Y^T M^{-\frac 1 2}L_B M^{-\frac 1 2}Y }_{:= \Lambda_P}\tilde \theta^* = 0.
    Y^T M^{-\frac 1 2}(f(\mathbb{1}_N \omega_b)+ \xi)-   Y^T M^{-\frac 1 2}L_B M^{-\frac 1 2}Y \tilde \theta^* = Y^T M^{\frac{1}{2}}\mathbb{1}_N \dot \omega_b = 0.
    \end{aligned}
\end{equation*}
For brevity, define the shorthands % 我单独拿出来主要是想更清晰地显示这俩定义，后面反复要用，怕读的时候一眼看不到
\begin{align}\label{eq:def of f tilde linear}
    \tilde f:=  Y^T M^{-\frac 1 2}(f(\mathbb{1}_N \omega_b)+ \xi),\quad \Lambda_P := Y^T M^{-\frac 1 2}L_B M^{-\frac 1 2}Y.
\end{align}
Then the above equation on $\tilde \theta^*$ is written compactly as
\begin{equation}\label{eq: theta* 3}
    \begin{aligned}
        \tilde f - \Lambda_P \tilde \theta^* = 0.
    \end{aligned}
\end{equation}
It can be checked that $\Lambda_P \in \mathbb{R}^{(N-1) \times (N-1)}$ is positive definite and its minimal singular value equals the second smallest eigenvalue of $M^{-\frac 1 2}L_B M^{-\frac 1 2}$, i.e., $\sigma_m(\Lambda_P) = \lambda_2 > 0$, since the network graph is connected. 
Besides, the matrix $L_B M^{-\frac 1 2}Y$ in the power flow term in \eqref{eq: omega dot after trans} can be rewritten in terms of $\Lambda_P$. Observe that
\begin{equation}\label{eq:LB and LambdaP trans}
    \begin{aligned}
         \Lambda_P Y^T M^{\frac 1 2} &= Y^T M^{-\frac 1 2}L_B (M^{-\frac 1 2}YY^T M^{\frac 1 2})\\
    &= Y^T M^{-\frac 1 2}L_B \left(I_N -  \mathbb{1}_N\mathbb{1}_N^T\frac{ M}{NM_b}\right)\\
&= Y^T M^{-\frac 1 2}L_B,
    \end{aligned}
\end{equation}
where the second equality follows from the expansion of the identity $QP = I_N$ and the last equality follows from $L_B \mathbb{1}_N = 0$. Thus we have $L_B M^{-\frac 1 2}Y = M^{\frac 1 2}Y \Lambda_P$.

With this in mind, we formally define the error variables
% The second transformation is to define the error system with two error $\delta_\omega$ and $\delta_\theta$: % which measure the coherence error of $\bar w$ and $\tilde \theta$ relative to the blended dynamics, respectively.
\begin{subequations}
    \begin{align}
          \delta_\omega &:= \omega-\mathbb{1}_N \omega_b, \\
  \delta_\theta &:= \tilde{\theta} - \tilde{\theta}^* = \tilde \theta - \Lambda_P^{-1}\tilde f.\label{eq: def etheta}
    \end{align}
\end{subequations}
% \noindent \textcolor{black}{These error variables measure the deviation of the system from the expected limiting behavior. Specifically, we would expect that $\bar w$ approximately converges to $\omega_b$, $\tilde w$ approximately converges to $0$, and $\tilde \theta$ approximately converges to $\Lambda_P^{-1}Y^T M^{-\frac 1 2}(f(\mathbb{1}_N \omega_b) + \xi)$ inferred from \eqref{eq: tilde omega after trans}. Note that if $\xi(t)\equiv \xi_0$ is a constant, \eqref{eq: def etheta} reduces to $\delta_\theta := \tilde \theta - \tilde \theta^*$, where 
% $\tilde \theta^* = \Lambda_P^{-1} Y^T M^{-1/2} (f(\mathbb{1}_N \omega_b^*) + \xi_0)$ is the steady-state value of $\tilde \theta$ solved by enforcing the right-hand side of \eqref{eq: tilde omega after trans} to be $0$. 
% In the subsequent analysis, including the proof of Corollary~\ref{coro1}, we show that these expected limiting behaviors are indeed valid, by establishing that the defined error variables are ultimately bounded.}
% Here the definition of $\delta_\theta$ is carefully crafted although it is not so intuitive. In particular, 
Then, the system \eqref{eq: transformed system 1} is rewritten based on the error variables $\delta_\omega$ and $\delta_\theta$. The dynamics of $\delta_\omega$ is given as
\begin{equation}
    \begin{aligned}\label{eq:err delta w}
     \dot \delta_\omega &= M^{-1}(f(\omega) + \xi - M^{\frac 1 2}Y \Lambda_P \tilde \theta) - \mathbb{1}_N \dot \omega_b
    \\
    &= M^{-1}(\underbrace{f(\omega) - f(\mathbb{1}_N \omega_b)}_{:= \Delta f}) + M^{-1}  (f(\mathbb{1}_N \omega_b)+ \xi ) - \mathbb{1}_N \dot \omega_b \\
    & - M^{-1}(M^{\frac 1 2}Y \Lambda_P \delta_\theta + M^{\frac 1 2}Y \Lambda_P\tilde \theta^*)
    \\
    &=  M^{-1}\Delta f - M^{-\frac 1 2} Y \Lambda_P \delta_\theta.
 %    \\
 %        \dot \delta_\omega &= \frac{1}{NM_b} \mathbb{1}_N\mathbb{1}_N^T \Delta f  \\
 %    &+ M^{-\frac 1 2} Y [Y^T M^{-\frac 1 2}\Delta f - 
 % \textcolor{black}{k} (\nabla U(\tilde \theta) - \nabla U(\tilde \theta^*))]\\
 % & = M^{-1} \Delta f - \textcolor{black}{k}  M^{-\frac 1 2} Y(\nabla U(\tilde \theta) - \nabla U(\tilde \theta^*)),
    \end{aligned}
\end{equation}
Here in the first equality we replace the matrix $L_B M^{-\frac 1 2}Y$ in \eqref{eq: omega dot after trans} with $M^{\frac 1 2}Y \Lambda_P$. In the lase equality, some terms are canceled out by incorporating the definition of $\tilde \theta^*$ in \eqref{eq: theta* 1}, i.e., $M^{-1}(f(\mathbb{1}_N \omega_b)+ \xi)- M^{-1}M^{\frac 1 2}Y \Lambda_P\tilde \theta^*= \mathbb{1}_N\dot \omega_b$. 
% where the cancellation in the last step follows from the definition of $\tilde \theta^*$ in \eqref{eq: theta* 1}, i.e., $M^{-1}(f(\mathbb{1}_N \omega_b)+ \xi)- M^{-1}M^{\frac 1 2}Y \Lambda_P\tilde \theta^*= \mathbb{1}_N\dot \omega_b$. 
% identity $M^{-1} -  \mathbb{1}_N \mathbb{1}_N^T/(NM_b) = M^{-\frac 1 2} Y Y^T M^{-\frac 1 2}$ together with the definition of $\tilde \theta^*$ in \eqref{eq: theta* 2}.

The dynamics of $\delta_\theta$ is given as
\begin{equation}\label{eq:err e theta}
    \begin{aligned}
         \dot \delta_\theta
      &= Y^T M^{\frac 1 2} \omega - \Lambda_P^{-1} \frac{d\tilde f}{dt}=  Y^T M^{\frac 1 2}\delta_\omega - \Lambda_P^{-1} \frac{d\tilde f}{dt}
    \end{aligned}
\end{equation}
using $Y^T M^{\frac 1 2}\mathbb{1}_N\omega_b = 0$. Here $\frac{d\tilde f}{dt}$ is the time derivative of $\tilde f$ along the blended dynamics \eqref{eq: blended dyn}.

In the following analysis, we will use an upper bound on $|d\tilde f/dt|$, which is presented in the lemma below.
\begin{lemma}\label{lm: d tilde f dt}
Let Assumption \ref{assump1} hold. 
If $C := \max_{i \in \mathcal{N}} \sup_{t > 0}|\dot{\xi}_i(t)|/M_i$ is finite, then for $\forall t > 0$,
    \begin{align*}
          \left|\frac{d\tilde f}{dt}\right|^2 \leq 2 N M_b C^2 (1+\frac{L}{\mu})^2 + \frac{2N L^2}{M_b}|f_b(\omega_b(0))+ \xi_b(0_+)|^2 e^{-2\mu t}.
    \end{align*}
If $C_{\lim} := \max_{i \in \mathcal{N}}\limsup_{t \to \infty}|\dot{\xi}_i(t)|/M_i$ is finite, then
\begin{align*}
        \limsup_{t \to \infty} \left|\frac{d \tilde f}{dt}\right|^2 
        \leq 2 N M_b C_{\lim}^2(1+\frac{L}{\mu})^2.
\end{align*}
\end{lemma}
\begin{proof}
The time derivative of $\tilde{f}$ along \eqref{eq: blended dyn} is
\begin{align}\label{eq：tilde f dot expression}
  \frac{d\tilde f}{dt} = Y^T M^{-\frac 1 2}\dot{\xi} + Y^T M^{-\frac 1 2}\frac{\partial f(\mathbb{1}_N \omega_b)}{\partial \omega_b}\dot{\omega}_b.
\end{align}
For the first term in \eqref{eq：tilde f dot expression}, using the element-wise bound $M_i^{-1}|\dot \xi_i|\leq C,\ \forall t>0$ gives
\begin{equation*}
    \begin{aligned}
      |Y^T M^{-\frac 1 2}\dot \xi| &\leq |Y||M^{\frac 1 2}M^{-1}\dot \xi|
        \leq \sqrt{N M_b}C, \ \forall t >0.
    \end{aligned}
\end{equation*}
Similarly, using $\limsup_{t \to \infty}M_i^{-1}|\dot \xi_i|\leq C_{\lim}$ yields
\begin{equation*}
    \begin{aligned}
      \limsup_{t \to \infty}|Y^T M^{-\frac 1 2}\dot \xi|   \leq \sqrt{N M_b}C_{\lim}.
    \end{aligned}
\end{equation*}

For the second term in \eqref{eq：tilde f dot expression}, using $M_i^{-1}|f_i^\prime(\omega_b)| \leq L$ from Assumption \ref{assump1} leads to%
\begin{equation}\label{eq:lm ftilde term-2}
    \begin{aligned}
        | Y^T M^{-\frac 1 2}\frac{\partial f(\mathbb{1}_N\omega_b)}{\partial \omega_b}\dot \omega_b|
       &\leq \sqrt{N M_b}L |\dot \omega_b|, \ \forall t>0,\\
       \limsup_{t \to \infty}  | Y^T M^{-\frac 1 2}\frac{\partial f(\mathbb{1}_N\omega_b)}{\partial \omega_b}\dot \omega_b|& \leq \sqrt{N M_b}L  \limsup_{t \to \infty} |\dot \omega_b|.
    \end{aligned}
\end{equation}
Now it remains to derive an upper bound on $|\dot \omega_b|$. Taking the time derivative of \eqref{eq: blended dyn},
\begin{equation*}
    M_b \ddot \omega_b = \dot{\xi}_b + f_b^\prime(\omega_b) \dot \omega_b.
\end{equation*} Define $y(t) := |\dot \omega_b(t)|$, and its dynamics is given by
\begin{equation}\label{eq:y dot}
    \begin{aligned}
       M_b \dot y
         &= \operatorname{sign}(\dot \omega_b)( \dot{\xi}_b +f_b^\prime(\omega_b) \dot \omega_b ) \\
         &\leq \left|  \dot{\xi}_b \right| + f_b^\prime(\omega_b) y \leq \left|  \dot{\xi}_b \right| -M_b\mu y, \ \   \text{almost everywhere},
    \end{aligned}
\end{equation}
where we use 
% $|\frac{\partial{f_b}}{\partial t}(t,s)| \leq M_b C$ and 
$f_b^\prime(\omega_b) \leq - M_b \mu$ obtained by combining Assumption \ref{assump1} and the definition of $f_b$.
% and Assumption \ref{assump2}. 
Applying the comparison lemma to \eqref{eq:y dot} yields $y(t)=|\dot \omega_b(t)| \leq y_0(t),\forall t > 0$, where $y_0(t)$ is the solution to $M_b\dot{y}_0  = - M_b\mu y_0+ |\dot{\bar{ \xi}}|$ with $y_0(0_+) = y(0_+) = M_b^{-1}|f_b(\omega_b(0))+ \xi_b(0_+)|$. 
Since $|\dot {\xi}_b| \leq M_b C,\ \forall t>0$ and $\limsup_{t \to \infty}|\dot \xi_b| \leq M_b C_{\lim}$, we have
\begin{align*}
    |\dot \omega_b(t)| \leq y_0(t) \leq \frac{|f_b(\omega_b(0))+ \xi_b(0_+)|}{M_b} e^{-\mu t} + \frac{C}{\mu},
\end{align*}
where we drop the negative term $-(C/\mu)e^{-\mu t}$, and
\begin{align*}
      \limsup_{t\to \infty}   |\dot \omega_b(t)| \leq \frac{C_{\lim}}{\mu}.
\end{align*}
% \begin{align}\label{eq:upper bound of st}
%     \limsup_{t \to \infty}|\dot \omega_b(t)|\leq  \frac{C_{\lim}}{\mu}
% \end{align}
% where we use $\limsup_{t \to \infty}|\dot {\xi}_b| \leq M_b C_{\lim}$ 
% obtained by combining $ \limsup_{t \to \infty}  M_i^{-1}{|\dot \xi_i|}\leq C_{\lim}$ and the definition of $\xi_b$. 
Substitute these upper bound on $|\dot \omega_b|$ back into \eqref{eq:lm ftilde term-2}, and we can derive the overall estimate
\begin{subequations}
    \begin{align}
        |\frac{d\tilde f}{dt}| &\leq \sqrt{N M_b}C(1 + \frac{L}{\mu} )+   \sqrt{\frac{N L^2}{M_b}}|f_b(\omega_b(0))+ \xi_b(0_+)|e^{-\mu t},\label{eq:lemma conclus 1}
        \\
\limsup_{t \to \infty}|\frac{d\tilde f}{dt}|  &\leq \sqrt{N M_b}C_{\lim}(1 + \frac{L}{\mu} ).\label{eq:lemma conclus 2}
    \end{align}
\end{subequations}
Then the statements of the lemma follow by applying $(a+b)^2 \le 2a^2+2b^2$ to \eqref{eq:lemma conclus 1} and, for \eqref{eq:lemma conclus 2}, squaring both sides and relaxing the right-hand side by a factor of $2$. 
\end{proof}

% \noindent \textbf{Part 2: Lyapunov function analysis} 
\subsection{Lyapunov Function Analysis}
To derive the bound in \eqref{eq:thm1 general conclusion 1} and \eqref{eq:thm1 general conclusion 2}, we proceed to construct a Lyapunov function $V$ and show that $V$ declines into a small neighborhood of the origin.
% we will show that this is indeed a valid Lyapunov function, by proving positivity outside of the origin and strict negativity of its time derivative along the solutions of (23).
% show exponential decline of the Lyapunov function V (x),
Consider the following Lyapunov function candidate%
\begin{equation}\label{eq:V def}
\begin{aligned}
V(\delta_\omega,\delta_\theta) 
    &:= W_k(\delta_\omega) + W_{p,c}(\delta_\theta,\delta_\omega),
\end{aligned}
\end{equation}
with 
\begin{align*}
    W_k(\delta_\omega):= \frac{1}{2}\delta_\omega^T M \delta_\omega
\end{align*}
representing the kinetic energy and
\begin{align*}
    W_{p,c}(\delta_\theta,\delta_\omega):= \frac{1}{2}(\delta_\theta + \eta \Lambda_P^{-1}Y^T M^{\frac 1 2}\delta_\omega)^T \Lambda_P (\delta_\theta + \eta \Lambda_P^{-1}Y^T M^{\frac 1 2}\delta_\omega)
\end{align*}
representing the potential energy together with some crafted cross terms between the kinetic and the potential energy. Here $\eta \in (0,\mu)$ is a positive parameter to design. The detailed physical intuition for a similar Lyapunov function design can be found in~\cite{weitenberg2018robust}.

It can be seen that $V$ is positive definite, since $M >0$ and $\Lambda_P > 0$. The next step is to show that $\dot V \leq - c V + \kappa$ for some $c>0 ,\kappa \geq 0$ and for all $ t  > 0$. 
To achieve this, we start by developing an upper bound of $\dot V$ term by term.

% % \begin{align}\label{eq:V def}
% %     V(\delta_\omega,\tilde w,\delta_\theta) := \frac{NM_b}{2}\delta_\omega^2 + \frac{1}{2}\tilde{w}^T\tilde{w} + \frac{1}{2}\delta_\theta^T\Lambda_P \delta_\theta.
% % \end{align}
% \begin{equation}\label{eq:V def}
% \begin{aligned}
%     &\ \ V(\delta_\omega,\delta_\theta) \\
%     &:= \frac{NM_b}{2}\delta_\omega^2 + \frac{1}{2}\tilde{w}^T(I + \eta^2  \Lambda_P^{-1} )\tilde{w} + \frac{1}{2}\delta_\theta^T\Lambda_P \delta_\theta + \eta \tilde \omega^T \delta_\theta\\
%     &=
%     % \frac{NM_b}{2}\delta_\omega^2 + \frac{1}{2}\tilde{w}^T\tilde{w} + \frac{1}{2}(\delta_\theta  + \eta\Lambda_P^{-1} \tilde{w})^T\Lambda_P (\delta_\theta  + \eta\Lambda_P^{-1} \tilde{w})\\
%     % & = 
%     \frac{NM_b}{2}\delta_\omega^2 + \frac{1}{2}\tilde{w}^T\tilde{w} + \frac{1}{2}\delta_\theta^T\Lambda_P \delta_\theta + \eta \tilde \omega^T \delta_\theta + \frac 1 2 \eta^2 \tilde \omega ^T \Lambda_P^{-1} \tilde \omega\\
%     &= \frac{NM_b}{2}\delta_\omega^2 + \frac{1}{2}\tilde{w}^T\tilde{w} + \frac 1 2 (\delta_\theta + \eta \Lambda_P^{-1}\tilde \omega)^T \Lambda_P (\delta_\theta + \eta \Lambda_P^{-1}\tilde \omega),
% \end{aligned}
% \end{equation}
% where $\eta \in (0,\mu)$ is a positive parameter to design, which aims to introduce appropriate cross terms in the Lyapunov analysis. 

% Then we have
% \ \begin{equation}
%     \begin{aligned}
%         \dot{\hat e}_\theta = 
%     \end{aligned}
% \end{equation}
For the first term $W_k(\delta_\omega)$, its time derivative is given as
\begin{align*}
   \delta_\omega^T M\dot \delta_\omega &=   \delta_\omega^T \Delta f -  \delta_\omega^T M^{\frac 1 2} Y \Lambda_P  \delta_\theta.
\end{align*}
Since $f(\omega)= [f_1(\omega_1),\dots,f_N(\omega_N)]^T$, by the mean-value theorem, there exists some $z \in \mathbb{R}^N$ such that
\begin{align}\label{eq:mean value thm}
    \Delta f = \frac{\partial f}{\partial \omega}\Bigg|_z \cdot \delta_\omega
\end{align}
where $\frac{\partial f}{\partial \omega}\big|_z  = \operatorname{diag}(\frac{\partial f_i}{\partial \omega_i}\big|_{z_i},i \in \mathcal{N})\leq -\mu M$ by Assumption~\ref{assump1}.
Then we can bound $\delta_\omega^T M\dot \delta_\omega$ as
\begin{align}\label{eq:Vdot first part}
   \delta_\omega^T M\dot \delta_\omega 
& \leq -\mu \delta_\omega^T M\delta_\omega -  \delta_\omega^T M^{\frac 1 2} Y \Lambda_P \delta_\theta.
\end{align}

For the second term $ W_{p,c}(\delta_\theta,\delta_\omega)$, for simplicity of notations, define $$\textcolor{black}{\hat \delta_\theta:= \delta_\theta + \eta \Lambda_P^{-1}Y^T M^{\frac 1 2}\delta_\omega.}$$
Then the time derivative of $ W_{p,c}(\delta_\theta,\delta_\omega) = \frac 1 2 \hat{\delta}_\theta ^T \Lambda_P \hat \delta_\theta$ is given as
\begin{equation*}
    \begin{aligned}
      &\quad \hat \delta_\theta^T \Lambda_P(\dot \delta_\theta + \eta \Lambda_P^{-1}Y^T M^{\frac 1 2}\dot \delta_\omega) \\
 &  =     \hat \delta_\theta^T \Lambda_P (Y^T M^{\frac 1 2}\delta_\omega - \Lambda_P^{-1}\frac{d\tilde f}{dt} + \eta  \Lambda_P^{-1}Y^T M^{-\frac 1 2} \Delta f -  \eta  \delta_\theta),
    \end{aligned}
\end{equation*}
where we plug in the expressions for $\dot \delta_\omega$, $\dot \delta_\theta$ and then use $ (\Lambda_P^{-1}Y^T M^{\frac 1 2})$\\$\cdot (M^{-\frac 1 2} Y \Lambda_P) = I_{N-1}$. 
% \begin{equation}
%     \begin{aligned}
%         (B) &= \hat \delta_\theta^T \Lambda_P \dot{\hat e}_\theta\\
%         &= \hat \delta_\theta^T \Lambda_P (\tilde \omega - \Lambda_P^{-1}\frac{d\tilde f}{dt} + \eta \Lambda_P^{-1}(Y^T M^{-\frac 1 2}\Delta f  - \Lambda_P \delta_\theta))\\
%         &= \hat \delta_\theta^T \Lambda_P  \tilde \omega - \hat \delta_\theta^T \frac{d\tilde f}{dt} + \eta \hat \delta_\theta^T Y^T M^{-\frac 1 2}\Delta f - \eta \hat \delta_\theta^T \Lambda_P \delta_\theta.
%     \end{aligned}
% \end{equation}
Further replacing $\Delta f$ with \eqref{eq:mean value thm} and substituting $\delta_\theta$ with $\hat \delta_\theta - \eta \Lambda_P^{-1}Y^T M^{\frac 1 2}\delta_\omega$, we rewrite the time derivative of $ W_{p,c}(\delta_\theta,\delta_\omega)$ as
\begin{equation}\label{eq:Vdot second part}
    \begin{aligned}
         &\hat \delta_\theta^T \Lambda_P Y^T M^{\frac 1 2}\delta_\omega - \hat \delta_\theta^T \frac{d\tilde f}{dt} + \eta   \hat \delta_\theta^T Y^T M^{-\frac 1 2} \frac{\partial f}{\partial \omega}\Big|_z \delta_\omega \\
    & -  \eta  \hat \delta_\theta^T \Lambda_P  (\hat \delta_\theta - \eta \Lambda_P^{-1}Y^T M^{\frac 1 2}\delta_\omega).
    \end{aligned}
\end{equation}
% \begin{equation}\label{eq:Vdot term3}
%     \begin{aligned}
%         (B) &= \hat \delta_\theta^T \Lambda_P  \tilde \omega - \hat \delta_\theta^T \frac{d\tilde f}{dt} \\
%         &+ \eta \hat \delta_\theta^T Y^T M^{-\frac 1 2}\frac{\partial f}{\partial \omega}\Big|_z \cdot (\mathbb{1}_N \delta_\omega + M^{-\frac 1 2} Y \tilde \omega) \\
%         & -\eta \hat \delta_\theta^T \Lambda_P (\hat \delta_\theta - \eta \Lambda_P^{-1}\tilde \omega)
%     \end{aligned}
% \end{equation}

% {\noindent \small
% \begin{equation}\label{eq:Vdot term3}
%     \begin{aligned}
%         \delta_\theta^T \Lambda_P \dot{e}_\theta=& 
%          \delta_\theta^T \Lambda_P  \tilde w 
%         -  \delta_\theta^T \frac{d\tilde f}{dt} + \delta_\theta^T \eta^2 \tilde w - \delta_\theta^T  \eta \Lambda_P \delta_\theta
%         \\ &+   \delta_\theta^T\eta Y^T M^{-\frac 1 2} \frac{\partial f}{\partial \omega}\Big|_z  (\mathbb{1}_N \delta_\omega + M^{-\frac 1 2} Y \tilde w).\\
%      &   \eta^2 \tilde w^T \Lambda_P^{-1}  \dot{\tilde \omega} + \delta_\theta^T \Lambda_P \dot{e}_\theta + \eta \tilde w^T \dot{e}_\theta + \eta \dot{\tilde \omega}^T \delta_\theta\\
%      & = \eta^2 \tilde w^T \Lambda_P^{-1}(Y^T M^{-\frac 1 2}\Delta f  - \Lambda_P \delta_\theta) +  \delta_\theta^T \Lambda_P ( \tilde w  -\Lambda_P^{-1} \frac{d\tilde f}{dt})\\
%      & + \eta \tilde w^T ( \tilde w  -\Lambda_P^{-1} \frac{d\tilde f}{dt}) + \eta \delta_\theta^T (Y^T M^{-\frac 1 2}\Delta f  - \Lambda_P \delta_\theta)
%     \end{aligned}
% \end{equation}
% }%
% \vspace{-0.6em plus 0.3em}
Summing the two parts \eqref{eq:Vdot first part} and \eqref{eq:Vdot second part} above and substituting $\delta_\theta$ in \eqref{eq:Vdot first part} with $\hat \delta_\theta - \eta \Lambda_P^{-1}Y^T M^{\frac 1 2}\delta_\omega$, we obtain
% Combining \eqref{eq: Vdot term 1 and 2}, \eqref{eq: Vdot term1 and 2 refine} and \eqref{eq:Vdot term3}, we have
\begin{equation}\label{eq:Vdot allsumup}
    \begin{aligned}
         \dot V \leq & -\mu \delta_\omega^T M\delta_\omega -  \delta_\omega^T M^{\frac 1 2}Y \Lambda_P(\hat \delta_\theta - \eta \Lambda_P^{-1}Y^T M^{\frac 1 2}\delta_\omega)\\
         &+ \hat \delta_\theta^T \Lambda_P Y^T M^{\frac 1 2}\delta_\omega
         - \hat \delta_\theta^T \frac{d\tilde f}{dt} 
         + \eta   \hat \delta_\theta^T Y^T M^{-\frac 1 2} \frac{\partial f}{\partial \omega}\Big|_z \delta_\omega \\
    & -  \eta  \hat \delta_\theta^T \Lambda_P  (\hat \delta_\theta - \eta \Lambda_P^{-1}Y^T M^{\frac 1 2}\delta_\omega)
         \\ %以下是化简后的（原来的结果，把etheta加上hat）
         = & - \mu \delta_\omega^T M\delta_\omega + \eta \delta_\omega^T M^{\frac 1 2}YY^T M^{\frac 1 2}\delta_\omega - \eta \hat \delta_\theta^T \Lambda_P \hat \delta_\theta
         \\
         &+ \eta   \hat \delta_\theta^T Y^T M^{-\frac 1 2} \frac{\partial f}{\partial \omega}\Big|_z \delta_\omega 
         + \eta^2  \hat \delta_\theta^T Y^T M^{\frac 1 2} \delta_\omega 
          - \hat \delta_\theta^T \frac{d\tilde f}{dt},
    \end{aligned}
\end{equation}
where the last equality in \eqref{eq:Vdot allsumup} is derived by canceling the term $\hat\delta_\theta^T \Lambda_P Y^T M^{\frac 1 2}\delta_\omega$ with its negative counterpart and rearranging the order.

Now, we further use Young inequalities to bound the cross terms and first-order term in \eqref{eq:Vdot allsumup}. 
Since
\begin{align*}
    \left|M^{-\frac 1 2} \frac{\partial f}{\partial \omega}\Big|_z \delta_\omega \right| = \left|M^{- 1} \frac{\partial f}{\partial \omega}\Big|_z M^{\frac 1 2} \delta_\omega \right| 
    \leq L  \left|M^{\frac 1 2} \delta_\omega \right|
\end{align*}
by $M^{-1}\frac{\partial f}{\partial \omega}\big|_z  = \operatorname{diag}(\frac{\partial f_i}{\partial \omega_i}\big|_{z_i},i \in \mathcal{N})\geq -L I$ from Assumption~\ref{assump1}, we have
    \begin{subequations}\label{eq:thm 1 bound of all cross terms}
    \begin{align}
    % 第一项
    &\eta   \hat \delta_\theta^T Y^T M^{-\frac 1 2} \frac{\partial f}{\partial \omega}\Big|_z \delta_\omega  \notag\\
    \leq &\eta L |\hat \delta_\theta| |M^{\frac 1 2}\delta_\omega|
    \leq  \frac{\eta L^2}{\sigma_m(\Lambda_P)}|M^{\frac 1 2}\delta_\omega|^2
    + \frac{\eta \sigma_m(\Lambda_P)}{4}|\hat \delta_\theta|^2
    \label{thm1 cross 1},\\
    % 第二项
     &\eta^2  \hat \delta_\theta^T Y^T M^{\frac 1 2} \delta_\omega \notag\\
    \leq &\eta^2 |\hat \delta_\theta| |M^{\frac 1 2}\delta_\omega|
    \leq \frac{\eta^3}{\sigma_m(\Lambda_P)}|M^{\frac 1 2}\delta_\omega|^2
    + \frac{\eta \sigma_m(\Lambda_P)}{4}|\hat \delta_\theta|^2,
    \label{thm1 cross 2}\\
    % 第三项
    & |\hat \delta_\theta|  \left| \frac{d\tilde f}{dt} \right| 
    \leq \frac{ \left| \frac{d\tilde f}{dt} \right|^2 }{ \eta \sigma_m(\Lambda_P) } 
    + \frac{\eta \sigma_m(\Lambda_P)}{4} |\hat \delta_\theta|^2. \label{thm1 cross 3}
\end{align}
\end{subequations}
Now substitute \eqref{eq:thm 1 bound of all cross terms} back into \eqref{eq:Vdot allsumup} together with $\eta \delta_\omega^T M^{\frac 1 2}YY^T M^{\frac 1 2}\delta_\omega \leq \eta \delta_\omega^T M \delta_\omega$. Then sum up the coefficients of all quadratic terms, which gives
\begin{align*}
% \label{eq:Vdot upper bound with phi}
     \dot V &\leq - \phi_1 \delta_\omega^T M \delta_\omega - \frac{\eta}{4}\hat \delta_\theta^T \Lambda_P \hat \delta_\theta + \frac{ \left| \frac{d\tilde f}{dt} \right|^2 }{ \eta \sigma_m(\Lambda_P) },
    \end{align*}
    where 
    \begin{align*}
        \phi_1:= \mu - \eta - \frac{\eta L^2}{\sigma_m(\Lambda_P)} -  \frac{\eta^3}{\sigma_m(\Lambda_P)}.
    \end{align*}
Set 
$$\eta = \eta^* := \frac{\mu \sigma_m(\Lambda_P)}{2(\sigma_m(\Lambda_P)+ 4L^2)} \in (0,\frac{\mu}{2}),$$ then it can be checked that $\phi_1 > \frac{3}{8}\mu > \frac{\eta}{4}$. Recall that $V = \frac 1 2 \delta_\omega^T M \delta_\omega + \frac 1 2\hat \delta_\theta^T \Lambda_P \hat \delta_\theta$. Thus we obtain
% \small 
\begin{equation}\label{eq: Vdot relation with V}
    \begin{aligned}
         \dot V &\leq -\frac{\eta^*}{2} V + \frac{| \frac{d\tilde f}{dt}|^2}{ \eta^* \sigma_m(\Lambda_P) }\\
    & \leq -\frac{\eta^*}{2} V + 
    \underbrace{\frac{2 N L^2 |f_b(\omega_b(0))+ \xi_b(0_+)|^2}{{\eta^*}\sigma_m(\Lambda_P)  M_b}
    }_{:=\beta_1} e^{-2\mu t}\\
  & +\underbrace{\frac{2 N M_b C^2}{{\eta^*} \sigma_m(\Lambda_P)}(1+\frac{L}{\mu})^2}_{:= \beta_2},
    \end{aligned}
\end{equation} 
where the second step inserts the upper bound on $|d \tilde f/dt|$ from Lemma~\ref{lm: d tilde f dt}.
% Since $V = \frac 1 2 N M_b \delta_\omega^2 + \frac 1 2\tilde w^T \tilde w + \frac 1 2 \hat \delta_\theta^T\Lambda_P \hat \delta_\theta$, we can always find a positive constant $c_1$ such that 
% \begin{align*}
%     \dot V \leq -c_1 V + \frac{| \frac{d\tilde f}{dt}|^2}{ \frac{\mu \sigma_m(\Lambda_P)^2}{\sigma_m(\Lambda_P)+ 4L^2} }.
% \end{align*}

Applying the comparison lemma to the inequality above yields
\begin{align*}
    V(t)
&\leq e^{-\frac{\eta^*}{2}t} V(0_+)
  + \frac{2\beta_1}{4\mu-\eta^*}\left(e^{-\frac{\eta^*}{2}t}-e^{-2\mu t}\right)
  + \frac{2\beta_2}{\eta^*}\left(1-e^{-\frac{\eta^*}{2}t}\right).
\end{align*}
To simplify the expression, we drop the negative terms $-e^{-2\mu t}$ and $-e^{-\frac{\eta^*}{2}t}$ and use $4\mu - \eta^* > 3 \mu$, which leads to
\begin{align*}
% V(t) &\leq (V(0_+) + \frac{2\beta_1}{4 \mu - \eta^*})e^{-\frac{\eta^*}{2}t} + \frac{2 \beta_2}{\eta^*} \\
V(t) & \leq (V(0_+) + \frac{2\beta_1}{3 \mu})e^{-\frac{\eta^*}{2}t} + \frac{2 \beta_2}{\eta^*},
\ \forall t>0.
\end{align*}
% which further leads to
% \begin{align}
%   V(t) & \leq \left(V(0_+) -\frac{ 2\sup_{t >0} | \frac{d\tilde f}{dt}|^2}{{\eta^*}^2 \sigma_m(\Lambda_P) } \right)e^{-\frac {\eta^*}{2}t} + \frac{ 2\sup_{t >0} | \frac{d\tilde f}{dt}|^2}{{\eta^*}^2 \sigma_m(\Lambda_P) }, \notag \\
%   & \leq V(0_+) e^{-\frac {\eta^*}{2}t} + \frac{4N L^2 |f_b(\omega_b(0))+ \xi_b(0_+)|^2}{{\eta^*}^2\sigma_m(\Lambda_P)  M_b} e^{-2\mu t}\\
%   & +\frac{4 N M_b C^2}{{\eta^*}^2 \sigma_m(\Lambda_P)}(1+\frac{L}{\mu})^2,
% \end{align}
% where 
% by the comparison lemma. 
Finally, since for each $i \in \mathcal{N}$, 
\[
M_i |\omega_i(t) -  \omega_b(t)|^2 \leq \delta_\omega^T (t)M \delta_\omega(t) \leq 2 V(t),
\]
we arrive at
\begin{subequations}\label{eq:thm1 proof final step}
\begin{align}
      |\omega_i(t) -  \omega_b(t)|^2 & \leq 
    \underbrace{\frac{2}{\min_i M_i} (V(0_+) + \frac{2\beta_1}{3 \mu })}_{:=\alpha}e^{-\frac{\eta^*}{2}t}
    \\
    & + \frac{4 \beta_2}{(\min_i M_i)\eta^*} \label{eq:thm1 proof final step constant 1}
    \\
    & \leq  \alpha  e^{-\frac {\eta^*}{2}t} \\
    &+ \underbrace{\frac{32 NM_b (1+\frac L \mu)^2 }{(\min_i M_i ) \mu^2}}_{:= K}C^2 \frac{(\lambda_2 + 4L^2)^2}{\lambda_2^3},\label{eq:thm1 proof final step constant 2}
\end{align}
\end{subequations}
% \begin{equation}
%     \begin{aligned}
%          |\omega_i(t) -  \omega_b(t)|^2 & \leq 
%     \underbrace{\frac{2}{\min_i M_i} (V(0_+) + \frac{2\beta_1}{3 \mu })}_{:=\alpha}e^{-\frac{\eta^*}{2}t}
%     \\& + \frac{4 \beta_2}{(\min_i M_i)\eta^*}\\
%     & \leq  \alpha  e^{-\frac {\eta^*}{2}t} \\
%     &+ \underbrace{\frac{32 NM_b (1+\frac L \mu)^2 }{(\min_i M_i ) \mu^2}}_{:= K}C^2 \frac{(\lambda_2 + 4L^2)^2}{\lambda_2^3},
%     \end{aligned}
% \end{equation}
where from \eqref{eq:thm1 proof final step constant 1} to \eqref{eq:thm1 proof final step constant 2}, we incorporates the explicit expressions for $\beta_2$, $\eta^*$ and substitutes $\sigma_m(\Lambda_P) = \lambda_2$. 
Note that the rate
\begin{align*}
\frac{\eta^*}{2} &= \frac{\mu \lambda_2}{4(\lambda_2+ 4L^2)}:= c.
\end{align*}
Then we complete the proof of the bound \eqref{eq:thm1 general conclusion 1} in the first part of Theorem~\ref{thm: thm1 general}.

% \section{Proof of Theorem~\ref{thm: thm2 limiting}}\label{prof: thm2}
% \textcolor{olive}{(Should move to the proof above)} 
The proof of the second part follows the same line of argument as the first part, up to the inequality
\begin{align*}
         \dot V &\leq -\frac{\eta^*}{2} V + \frac{|\frac{d\tilde f}{dt}|^2}{\eta^* \sigma_m(\Lambda_P)}.
\end{align*}
Using the comparison lemma yields
\begin{align*}
        \limsup_{t \to \infty} V(t) &\leq \frac{2}{(\eta^*)^2 \sigma_m(\Lambda_P)} 
        \limsup_{t \to \infty} \left|\frac{d \tilde f}{dt}\right|^2\\
        &\leq \frac{4{N M_b}C_{\lim}^2}{(\eta^*)^2 \sigma_m(\Lambda_P)} \Big(1+\frac{L}{\mu}\Big)^2,
\end{align*}
where the second inequality follows from the upper bound on $\limsup_{t \to \infty} |d\tilde f/dt|^2$ in Lemma~\ref{lm: d tilde f dt}. Finally, plug in the explicit expression for $\eta^*$, replace $\sigma_m(\Lambda_P)$ with $\lambda_2$ and then substitute this bound into
% It remains to derive an upper bound on $\limsup_{t \to \infty} |d\tilde f/dt|^2$. Analogously to the derivation of Lemma~\ref{lm: d tilde f dt}, we obtain
% \begin{align*}
%         \limsup_{t \to \infty} \Big|\tfrac{d \tilde f}{dt}\Big|^2 
%         \leq \sqrt{N M_b}\, C_{\lim} + \sqrt{N M_b}\,L \limsup_{t \to \infty}|\dot \omega_b|.
% \end{align*}
% where $|\dot \omega_b|$ is bounded by the comparison variable $y_0$ satisfying
% \[
% M_b \dot y_0 = - M_b \mu y_0 + |\dot{\xi}_b|.
% \]
% Since $\limsup_{t \to \infty}|\dot {\xi}_b(t)| \leq M_b C_{\lim}$, we have $\limsup_{t \to \infty} |\dot \omega_b(t)| \leq C_{\lim}/\mu$. 
% Substituting this bound back gives
\[
\limsup_{t \to \infty} |\omega_i(t) -  \omega_b(t)|^2 
\leq \frac{2}{\min_i M_i}\limsup_{t \to \infty} V(t).
\]
This leads to the desired conclusion \eqref{eq:thm1 general conclusion 2} in the second part of Theorem~\ref{thm: thm1 general}. Then the proof is completed.

\section{Existence of Solutions to the Steady State Conditions \eqref{eq:steady state condition} }\label{sec: proof of steady state}
In this section, we show that there always exist solutions $(\omega(0),\theta(0))$ to the steady state conditions \eqref{eq:steady state condition} under linear power flows. 

First, \eqref{eq:steady 1} implies $\omega_i(0) = \omega_s(0)$,$ \forall i \in \mathcal{N}$ for some synchronous frequency $\omega_s(0)$, since the network graph is connected. Regarding the equation \eqref{eq:steady 2}, we write its compact form as
\begin{equation}\label{eq:compact_form_proof}
    0 = f(\mathbb{1}_N\omega_s(0)) + \xi(0_-) - L_B \theta(0).
\end{equation}
Similar to the procedure used to address the equation \eqref{eq: theta* 1} in Appendix~\ref{prof: thm1}, we resolve \eqref{eq:compact_form_proof} here along the basis directions of $[\mathbb{1}_N,M^{-\frac 1 2}Y]$. First, left-multiplying \eqref{eq:compact_form_proof} by $\mathbb{1}_N^T$ yields 
\[
    \mathbb{1}_N^T f(\mathbb{1}_N\omega_s(0)) + \mathbb{1}_N^T \xi(0_-) = Nf_b(\omega_s(0))+ N\xi_b(0_-) =  0,
\]
which admits a unique solution $$\omega_s(0) = f_b^{-1}(-\xi_b(0_-)),$$ since $f_b^\prime(\omega_b) \leq - M_b \mu <0 $ by Assumption~\ref{assump1}. 
Next, we left-multiply \eqref{eq:compact_form_proof} by $Y^T M^{-\frac 1 2}$, leading to
\begin{equation}\label{eq: theta tilde 0 solution}
    \begin{aligned}
        &  Y^T M^{-\frac 1 2}(f(\mathbb{1}_N\omega_s(0))+\xi(0_-))- \Lambda_P \tilde \theta(0)=0,\\
    \Leftrightarrow \ & \tilde \theta(0) =  \Lambda_P^{-1} Y^T M^{-\frac 1 2}(f(\mathbb{1}_N\omega_s(0))+\xi(0_-)),
    \end{aligned}
\end{equation}
where $\tilde \theta(0) = Y^TM^{\frac 1 2}\theta(0)$ is the transformed coordinate defined in \eqref{eq:coordinate transform def}, and $ \Lambda_P>0$ is the matrix as defined in \eqref{eq: theta* 3}. Therefore, a solution $\theta(0)$ to the equations \eqref{eq:steady state condition} exists and is unique up to a uniform shift. Specifically, $\theta(0) = \mathbb{1}_N\bar \theta(0) + M^{-\frac 1 2}Y \tilde \theta(0)$ for the $\tilde \theta(0)$ specified in \eqref{eq: theta tilde 0 solution} and any scalar $\bar \theta(0)$. 

\section{Proof of Proposition~\ref{coro: start from steady}}\label{prof: coro1}
The proof is nearly identical to that of Theorem~\ref{thm: thm1 general} up to the definition of $\alpha$ in \eqref{eq:thm1 proof final step}. The key difference is that $\alpha$ can be further simplified by substituting the specific initial state $(\omega(0),\theta(0))$ whose expression is derived in Appendix~\ref{sec: proof of steady state} as 
\begin{align*}
    \omega(0) &= \mathbb{1}_N \omega_s(0) =  \mathbb{1}_N f_b^{-1}(-\xi_b(0_-)),\\
    \tilde \theta(0) &= \Lambda_P^{-1} Y^T M^{-\frac 1 2} (f(\mathbb{1}_N  \omega_b(0)) + \xi(0_-)).
\end{align*}
Recall the expression of $\alpha$: 
\begin{align*}
\alpha = \frac{2}{\min_i M_i}\left(V(0_+) + \frac{2\beta_1}{3\mu}\right).
\end{align*}
We proceed by calculating the terms $V(0_+)$ and $\beta_1$. For $V(0_+) = V(\delta_\omega(0), \delta_\theta(0_+))$, we have $\delta_\omega(0) = \mathbb{1}_N\omega_s(0) - \mathbb{1}_N\omega_b(0)= 0$ and $\delta_\theta(0_+)$ can be written as
\begin{align*}
& \quad \tilde \theta(0) - \tilde \theta^*(0_+) \\
&= \Lambda_P^{-1} Y^T M^{-\frac 1 2} \left( (f(\mathbb{1}_N  \omega_s(0)) + \xi(0_-)) - (f(\mathbb{1}_N  \omega_b(0)) + \xi(0_+)) \right) \\
&= -\Lambda_P^{-1} Y^T M^{-\frac 1 2}\Delta \xi.
\end{align*}
Substituting these into the definition of $V$ in \eqref{eq:V def} yields the initial value $V(0_+)$ as
\begin{align*}
V(0_+) &= \frac 1 2\delta_\theta^T (0_+) \Lambda_P \delta_\theta(0_+) \\
& = \frac 1 2 \Delta \xi^TM^{-\frac 1 2} Y \Lambda_P^{-1} Y^T M^{-\frac 1 2}\Delta \xi \\
&\leq  \frac{1}{2\lambda_2} \frac{|\Delta \xi|^2}{\min_i M_i}.
\end{align*}
Next, we calculate the term $\beta_1$. Recall its definition and substitute the steady-state condition $f_b(\omega_b(0))=-\xi_b(0_-)$: 
\begin{align*}
 \beta_1 &=   \frac{2 N L^2 | f_b(\omega_b(0))+  \xi_b(0_+)|^2}{{\eta^*}\lambda_2  M_b} \\
&= \frac{2 N L^2 | -\xi_b(0_-)+  \xi_b(0_+)|^2}{{\eta^*}\lambda_2  M_b}.
\end{align*}
Using the property $|\xi_b(0_+)-\xi_b(0_-)|^2 \leq {|\Delta \xi|^2}/N$ and the definition of $\eta^* =\mu \lambda_2/(2(\lambda_2 + 4L^2)) $, we obtain
\begin{align*}
\beta_1 \leq \frac{4 L^2(\lambda_2+4L^2)}{\mu M_b\lambda_2^2}|\Delta\xi|^2.
\end{align*}
Finally, incorporating these new bounds for $V(0_+)$ and $\beta_1$ into the expression of $\alpha$ leads to
\begin{align*}
\alpha &\leq \frac{2}{\min_i M_i}\left(\frac{1}{2\lambda_2} \frac{|\Delta \xi|^2}{\min_i M_i} + \frac{2}{3\mu}\frac{4 L^2(\lambda_2+4L^2)}{\mu M_b\lambda_2^2}|\Delta\xi|^2\right)\\
&= \alpha^* |\Delta \xi|^2,
\end{align*}
where $\alpha^*$ is a constant defined as
\begin{align*}
\alpha^* &:=\underbrace{ \frac{1}{\min_i M_i^2}}_{:= \phi_1}\frac{1}{\lambda_2}  + \underbrace{\frac{16 L^2}{3 \mu^2 M_b (\min_i M_i)}}_{:= \phi_2}\frac{\lambda_2+ 4L^2}{\lambda_2^2}.
\end{align*}
This completes the proof.

\section{Proof of Theorem~\ref{thm3:nonlinear}}\label{prof: thm3}
% \textcolor{olive}{Need to remove the repetitive sentences with proof of Theorem~\ref{thm: thm1 general}. }

The proof follows the same line of arguments as in Theorem \ref{thm: thm1 general}. 

To rewrite the original system \eqref{eq:nonlinear dynamics} into a vector form, we assign an arbitrary but fixed orientation to each edge in $\mathcal{E}$, based on which we define an node-edge incidence matrix $A \in \mathbb{R}^{N \times E}$. Specifically, for an edge $l \in \{1,\dots,E\}$ corresponding to the pair $\{i,j\}$, if the orientation is assigned from $i$ to $j$, then the $l$-th column of $A$ has entries $A_{il} = 1$ and $A_{jl} = -1$, with all other entries being zero. Let $\Gamma := \operatorname{diag}(B_{ij}^0,\{i,j\}\in \mathcal{E})$ collect the edge weights. Then the vector of power flows, governed by the nonlinear equations \eqref{eq:power flow}, can be expressed as the gradient of a magnetic energy function
\begin{align*}
    U_0(\theta) := -\mathbb{1}_E^T \Gamma \cos(A^T \theta).
\end{align*}
The power flow vector is then given by
\begin{align*}
    \nabla U_0(\theta) &= A\Gamma \sin(A^T \theta)\\
    &= \left[\sum_{j=1}^N B_{ij}^0\sin(\theta_i - \theta_j), i \in \mathcal{N}\right]^T.
\end{align*}
To this end, we obtain the following compact form of the system:
\begin{subequations}\label{eq:compact form nonlinear}
        \begin{align}
           \dot \theta &= \omega,\\
          \dot \omega & = M^{-1}(f(\omega) + \xi - \textcolor{black}{k}\nabla U_0(\theta)).\label{eq:compact form nonlinear omega}  
        \end{align}
   \end{subequations}

\subsection{Coordinate Transformation}
As in the proof of Theorem~\ref{thm: thm1 general}, we decompose $\theta$ into an average component $\bar \theta$ and a disagreement component $\tilde \theta$, with $\theta = \mathbb{1}_N \bar \theta + M^{-\frac 1 2}Y\tilde \theta.$ Since $A^T\theta = A^T(\mathbb{1}_N \bar \theta + M^{-\frac{1}{2}}Y \tilde \theta) = A^T M^{-\frac{1}{2}}Y \tilde \theta$ using $A^T \mathbb{1}_N = 0$, the magnetic energy function $U_0(\theta)$ can be expressed in the new coordinate as a function of $\tilde \theta$, given by
\begin{align*}
    U(\tilde \theta):= -\mathbb{1}_E^T \Gamma \cos(A^T M^{-\frac{1}{2}}Y \tilde \theta) \equiv -\mathbb{1}_E^T \Gamma \cos(A^T  \theta) =U_0(\theta).
\end{align*}
And the gradient of $U(\tilde \theta)$ can be expressed as
\begin{align*}
    \nabla U(\tilde \theta) =   Y^T M^{-\frac 1 2} A\, \Gamma\, \sin(A^T M^{-\frac 1 2} Y \tilde{\theta}) \equiv Y^T M^{-\frac 1 2} \nabla U_0(\theta).
\end{align*}
Similar to the derivation of \eqref{eq:LB and LambdaP trans} under the linear power flows, the power flow term $\nabla U_0(\theta)$ in \eqref{eq:compact form nonlinear omega} can also be rewritten in terms of $\nabla U(\tilde \theta)$. To see this, note that
\begin{equation*}
    \begin{aligned}
          \nabla U^T(\tilde \theta)Y^T M^{\frac 1 2}&= (\nabla U_0^T(\theta)M^{-\frac 1 2}Y )Y^T M^{\frac 1 2}
          \\
    &=\nabla U_0^T(\theta) \left(I_N -  \mathbb{1}_N\mathbb{1}_N^T\frac{ M}{NM_b}\right)\\
&= \nabla U_0^T(\theta).
    \end{aligned}
\end{equation*}
Here the second equality follows from the expansion of the identity $QP = I_N$, where $P$ and $Q$ are the coordinate transformation matrix defined in \eqref{eq:coordinate transform def} and \eqref{eq:coordinate transform def inverse}. The last equality follows from $\nabla U_0^T(\theta) \mathbb{1}_N = 0$. This allows us to substitute $\nabla U_0(\theta)$ with $M^{\frac 1 2}Y\nabla U(\tilde \theta)$.

% \begin{align*}
%     M^{-\frac 1 2}Y\nabla U(\tilde \theta) &= M^{-\frac 1 2} Y Y^T M^{-\frac 1 2} \nabla U_0(\theta)\\
%     &= (M^{-1} - \frac{1}{NM_b}\mathbb{1}_N \mathbb{1}_N^T )\nabla U_0(\theta)\\
%     &= M^{-1}U_0(\theta)
% \end{align*}
% $\mathbb{1}_N \mathbb{1}_N^T \frac{I}{NM_b} + M^{-\frac 1 2} Y Y^T M^{-\frac 1 2} = M^{-1}$
Thus the system dynamics becomes%
% \begin{subequations}\label{eq: transformed system 1 nonlinear}
%     \begin{align}
%     \dot {\bar w} &=\mathbb{1}_N^T \frac{M}{NM_b}  M^{-1}(f(w ) + \xi -k\nabla U_0(\theta))= \frac{1}{NM_b}\mathbb{1}_N^T (f(\omega) + \xi),\label{eq: bar w nonlinear}
%     \\ \dot {\tilde \omega} 
%     &= Y^T M^{\frac 1 2}M^{-1}(f(\omega) + \xi -k\nabla U_0(\mathbb{1}_N \bar \theta + M^{-\frac{1}{2}}Y \tilde \theta) )
% \notag  \\
%     &= Y^T M^{-\frac 1 2}(f(\omega) + \xi ) - 
%  \textcolor{black}{k}  Y^T M^{-\frac 1 2} C\, \Gamma\, \sin(A^T M^{-\frac 1 2} Y \tilde{\theta}), \label{eq: tilde w after trans nonlinear}\\
%     \dot{\tilde \theta} & = Y^T M^{\frac 1 2}w  =Y^T M^{\frac 1 2}  ({1}_N \bar w + M^{-\frac{1}{2}}Y \tilde w)= \tilde w,\label{eq:tilde theta nonlinear}
%     \end{align}
% \end{subequations}
\begin{subequations}\label{eq: transformed system 1 nonlinear}
    \begin{align}
   \dot \omega &= M^{-1}(f(\omega)+ \xi - \textcolor{black}{k}M^{\frac 1 2}Y \nabla U(\tilde \theta)),\label{eq: w nonlinear}
    \\ 
    \dot{\tilde \theta} & =Y^T M^{\frac 1 2}  \omega,\label{eq:tilde theta nonlinear}
    \end{align}
\end{subequations}
with the dynamics of $\bar \theta$ omitted.

% In addition, it can be checked that $\Lambda_P$ is positive definite and the minimal singular value of $\Lambda_P$ is equal to the second smallest eigenvalue of $M^{-\frac 1 2}L_BM^{-\frac 1 2}$, i.e., $\sigma_m(\Lambda_P) =  \lambda_2(M^{-\frac 1 2}L_BM^{-\frac 1 2})=\lambda_2(M^{-1}L_B)$, which is positive since the network is connected.

% \subsection{Define the error system}

The second step of the coordinate transformation is to define the error variables, following the same principle as in the proof of Theorem~\ref{thm: thm1 general}. We anticipate that $\omega$ approximately converges to $\mathbb{1}_N \omega_b$ and $\tilde \theta$ approximately converges to some $\tilde \theta^*$ that enforces the right-hand side of \eqref{eq: w nonlinear} to coincide with $\mathbb{1}_N \dot \omega_b$, i.e., 
\begin{equation}\label{eq: theta* 1 nonlinear}
    \begin{aligned}
         M^{-1}(f(\mathbb{1}_N \omega_b)+ \xi- k M^{\frac 1 2}Y \nabla U(\tilde \theta^*))&= \mathbb{1}_N\dot \omega_b\\ 
            &= \frac{\mathbb{1}_N \mathbb{1}_N^T}{NM_b} (f(\mathbb{1}_N \omega_b)+ \xi ).
    \end{aligned}
\end{equation}
To this end, we first establish the existence and uniqueness of the solution $\tilde \theta^*$ to the equation~\eqref{eq: theta* 1 nonlinear} within a proper region. Specifically, to ensure the sin nonlinearities are well-behaved, we define a safety set where the angle differences $A^T M^{-\frac 1 2}Y \tilde \theta$ are bounded away from $\pm \pi/2$. For any $\rho \in (0,\frac{\pi}{2})$, define
\begin{align*}
    \mathbb{S} (\rho):= \{\tilde \theta \in \mathbb{R}^{N-1}:  |A^T M^{-\frac 1 2}Y \tilde \theta|_{\infty} < \frac{\pi}{2} - \rho\},
\end{align*}
where $|x|_{\infty}:= \max_{i}|x_i|$ for any vector $x$. Then we show that the equation \eqref{eq: theta* 1 nonlinear} admits a unique solution in $ \mathbb{S}(2\rho) $ when $\omega_b(0)$ is in a proper region, which is given in the following lemmas. 
% Since $A^T\theta = A^T M^{-\frac{1}{2}}Y \tilde \theta$, we have the following equivalent representation
% \begin{align*}
%     \mathbb{S}(\rho) = \{  Y^T M^{\frac 1 2}\theta: \theta\in \mathbb{R}^N,|\theta_i - \theta_j|< \frac \pi 2 - \rho,\ \forall \{i,j\}\in \mathcal{E}  \}.
% \end{align*}
% Then we show that the equation \eqref{theta* 1 nonlinear} admits a unique solution $\tilde \theta^*$ in $ \mathbb{S}(2\rho) $ when $\omega_b(0)$ is in a proper region, given in the following lemma. 

\begin{lemma}
% [power flow feasibility]
\label{lm:power flow feasi}
Let Assumption~\ref{assump1} and~\ref{assump:feasible} hold with some $\rho \in (0,\frac{\pi}{4})$. 
% If $\omega(0)$ satisfies $ \mu |\mathbb{1}_N^T M \omega(0)| \leq N \max_{\tau \geq 0}|\xi_b(\tau)|$\footnote{Such a $\omega(0)$ always exists, since $\omega(0) = \mathbb{1}_N f_b^{-1}(-\xi_b(0_+))$ satisfies this condition.}, 
Suppose that 
\[
|\omega_b(0)| = \frac{|\mathbb{1}_N^T M \omega(0)|}{N M_b} \leq \frac{|\xi_b(0_+)|}{\mu M_b } + \frac{k \lambda_2^L \cos(2\rho)}{8 L (\max_i M_i)}.
\]
% Choose $\omega_b(0_+)$ such that $f_b(\omega_b(0_+)) + \xi_b(0_+) = 0$. 
Then for each $t > 0$, there exists a unique $\tilde \theta^*(t)$ in $\mathbb{S}(\textcolor{black}{2\rho})$ that is a solution to the equation \eqref{eq: theta* 1 nonlinear} at time $t$. 
% \footnote{The equation here is slightly different from the steady state equation in \eqref{eq:steady state condition nonlinear}. The latter can also be written as $k \nabla U(\tilde \theta(0)) = Y^T M^{-\frac 1 2}(f(\mathbb{1}_N \omega_s(0)) + \xi(0_-))$, but with $\omega_s(0) = f_b^{-1}(-\xi(0_-))$ as shown in Appendix~\ref{sec: proof of steady state}. }
    % \begin{align}\label{eq:assump condi}
    % \textcolor{black}{k} \nabla U( \tilde \theta^*(t))= Y^T M^{-\frac 1 2}(f(\mathbb{1}_N \omega_b(t)) + \xi(t) ).
    % \end{align}
\end{lemma}
\begin{proof}
We begin by rearranging the equation \eqref{eq: theta* 1 nonlinear}. 
Define $g(t) := f(\mathbb{1}_N \omega_b(t)) + \xi(t)$. In what follows we sometimes omit the explicit time index $t$ when no confusion arises. Then the equation \eqref{eq: theta* 1 nonlinear} is written compactly as
\begin{align}
    & M^{-1}(g- k M^{\frac 1 2}Y \nabla U(\tilde \theta^*))= \frac{\mathbb{1}_N \mathbb{1}_N^T}{NM_b} g, \notag \\
     \Leftrightarrow  \quad & k M^{\frac 1 2}Y \nabla U(\tilde \theta^*) = g - M \frac{\mathbb{1}_N \mathbb{1}_N^T}{NM_b}g. \label{eq:lemma feasible eq1}
\end{align}
Here the left-hand side is precisely the power flow vector determined by $\tilde \theta^*$, and the right-hand side is a vector with a zero average, since
\[
\mathbb{1}_N^T (g - M \frac{\mathbb{1}_N \mathbb{1}_N^T}{NM_b}g) = \mathbb{1}_N^T g - \frac{\mathbb{1}_N^T M \mathbb{1}_N }{NM_b}\mathbb{1}_N^Tg = 0.
\]
This enables us to use the phase cohesiveness condition in \cite{dorfler2013synchronization}, which states that the equation \eqref{eq:lemma feasible eq1} admits a unique solution $\tilde \theta^*\in \mathbb{S}(\textcolor{black}{2\rho})$ if 
% \[
% |L_B^\dagger (g  - \mathbb{1}_N \bar g)|_{\mathcal{E},\infty} = |L_B^\dagger g|_{\mathcal{E},\infty} \leq k \cos(2\rho),
% \]
% where $L_B^\dagger$ is the pseudo-inverse of $L_B$. 
% Furthermore, a sufficient condition for the above inequality to hold is
\begin{align}\label{eq:assump condi3}
    \left|g - M \frac{\mathbb{1}_N \mathbb{1}_N^T}{NM_b}g\right|_{\mathcal{E},\infty} \leq k \lambda_2^L \cos(2\rho),
\end{align}
where $\lambda_2^L$ is the second smallest eigenvalue of $L_B$.

In the remaining part of the proof, we are going to show that \eqref{eq:assump condi3} holds under Assumption~\ref{assump1} and~\ref{assump:feasible}. Define 
\[
 \tilde g:= g - M \frac{\mathbb{1}_N \mathbb{1}_N^T}{NM_b}g.
\]
Note that the $i$-th component of the vector $\tilde g$ can be written as
\begin{align*}
    f_i(\omega_b) + \xi_i -  \frac{M_i(f_b(\omega_b) + \xi_b)}{M_b}.
\end{align*}
Thus we have
\begin{align*}
    \left|\tilde g\right|_{\mathcal{E},\infty} &\leq 2 \max_{i \in \mathcal{N}} |f_i(\omega_b)| +|\xi|_{\mathcal E,\infty}\\
    & + 2 (\max_i M_i) \left(\frac{|f_b(\omega_b)|}{M_b} + \frac{\sup_{\tau \geq 0}|\xi_b(\tau)|}{M_b} \right).
\end{align*}
It remains to derive an upper bound for $|f_i(\omega_b)|$ and $|f_b(\omega_b)|$. 
By the mean-value theorem, there exists $\omega_i^\circ$ between 0 and $\omega_b$ such that $f_i(\omega_b)=f_i’(\omega_i^\circ)\,\omega_b + f_i(0)$. Since $f_i(0)=0$ and $|f_i^\prime(\omega)|\le L M_i$ for all $\omega$ by Assumption~\ref{assump1}, we obtain
\begin{align*}
     \max_{i \in \mathcal{N}} |f_i(\omega_b)| \leq  L(\max_i M_i)|\omega_b|.
\end{align*}
Following similar arguments, we have
\begin{align*}
    |f_b(\omega_b)| \leq L M_b |\omega_b|.
\end{align*}
Substituting these inequalities back gives the following upper bound on $\left|\tilde g\right|_{\mathcal{E},\infty} $:
\begin{equation}\label{eq:lemma theta* 2}
    \begin{aligned}
          \left|\tilde g\right|_{\mathcal{E},\infty} &\leq 4L (\max_i M_i)|\omega_b| +|\xi|_{\mathcal E,\infty}\\
       &+ \frac{ 2 (\max_i M_i) }{M_b}\sup_{\tau \geq 0}|\xi_b(\tau)|.
    \end{aligned}
\end{equation}
To further control $|\omega_b|$, define $z(t):= |\omega_b(t)|$, then the dynamics of $z$ is given by
\begin{align*}
   M_b \dot z&= \operatorname{sign}(\omega_b)(f_b(\omega_b) + \xi_b),\quad \text{almost everywhere}.
   % \\
    % &= \operatorname{sign}(\omega_b)(f_b^\prime(w^\circ)\omega_b+ \xi_b),
    \end{align*}
Again by the mean-value theorem, there exists $w_b^\circ$ between 0 and $\omega_b$ such that $\operatorname{sign}(\omega_b)f_b(\omega_b) = f_b^\prime (\omega_b^\circ)\operatorname{sign}(\omega_b)\omega_b \leq - M_b \mu  |\omega_b|$, which uses $f_b'(w_b^\circ) \leq -\mu M_b$ by Assumption~\ref{assump1}. Then we obtain
    \begin{align*}
        M_b \dot z & \leq -M_b \mu z + |\xi_b|.
    \end{align*}
    Applying the comparison lemma, we obtain that for all $t \geq  0$:
\begin{align*}
% \label{eq:assump condi4}
    |\omega_b(t)|&\leq e^{-\mu t} |\omega_b(0)| + \int_{0}^t e^{-\mu(t - \tau)}\frac{|\xi_b(\tau)|}{M_b}\,d\tau,\\
    & \leq e^{-\mu t} |\omega_b(0)| + \frac{\sup_{\tau \geq 0}|\xi_b(\tau)|}{M_b \mu}(1-e^{-\mu t}).
\end{align*}
Substituting the bound on $|\omega_b(t)|$ back into \eqref{eq:lemma theta* 2}, we conclude that
\begin{align*}
   |g|_{\mathcal{E},\infty} & \leq A_1 + A_2
\end{align*}
where
\begin{align*}
    A_1 &:= 4L(\max_i M_i) e^{- \mu t} \left(|\omega_b(0)| - \frac{\sup_{\tau  \geq 0}|\xi_b(\tau)|}{M_b \mu}\right)\\
   & \leq \frac 1 2 k \lambda_2^L \cos(2\rho)
\end{align*}
by using the condition on $|\omega_b(0)|$, and 
\begin{align*}
    A_2 &:=  \frac{4L(\max_i M_i)\sup_{\tau  \geq 0}|\xi_b(\tau)|}{M_b \mu} + |\xi|_{\mathcal{E},\infty}\\
   & + \frac{ 2 (\max_i M_i) }{M_b}\sup_{\tau \geq 0}|\xi_b(\tau)|\\
   & \leq \frac{6L(\max_i M_i)\sup_{\tau  \geq 0}|\xi_b(\tau)|}{M_b \mu} + |\xi|_{\mathcal{E},\infty},\\
   & \leq \frac 1 2 k \lambda_2^L \cos(2\rho),
\end{align*}
where the first inequality follows from $L \geq \mu$, and the last inequality follows from the restriction on $\xi$ from Assumption~\ref{assump:feasible}. This confirms that $\left|\tilde g\right|_{\mathcal{E},\infty} \leq k \lambda_2^L \cos(2\rho)$ and thus completes the proof. 
\end{proof}
Having established the existence of a unique solution $\tilde\theta^*(t)$ in $\mathbb{S}(2\rho)$ for all $t>0$, we now derive a more explicit form for $\nabla U( \tilde \theta^*(t))$, which will facilitate the computation of $\dot {\tilde \theta}^*(t)$. 
Specifically, we resolve the equation \eqref{eq: theta* 1 nonlinear} following the same procedure that was used to derive \eqref{eq: theta* 3} under linear power flows. The equation \eqref{eq: theta* 1 nonlinear} left-multiplied by $\mathbb{1}_N^T M$ holds true for any value of $\tilde \theta^*$, using $\mathbb{1}_N^T M^{\frac 1 2}Y =0$. Thus $\tilde\theta^*$ is determined by left-multiplying \eqref{eq: theta* 1 nonlinear} by $Y^T M^{\frac 1 2}$, which leads to
\begin{align}\label{eq:assump condi derivation}
\underbrace{Y^T M^{-\frac 1 2}(f(\mathbb{1}_N \omega_b) + \xi ) }_{:=\tilde f}-     k \nabla U( \tilde \theta^*) = 0,\ \forall t>0,
\end{align}
where we use $Y^T Y = I_{N-1}$. 
Taking the time derivative over both sides leads to
    \begin{align}\label{eq:theta tilde dot equation}
        \textcolor{black}{k} \nabla^2 U(\tilde \theta^*) \dot{\tilde \theta}^* = \frac{d \tilde f}{dt},
    \end{align}
    where \begin{align*}
        \nabla^2 U(\tilde \theta^*) = Y^TM^{-\frac 1 2} A \operatorname{diag}(\Gamma \cos(A^TM^{-\frac 1 2}Y \tilde \theta^*))A^T M^{-\frac 1 2}Y,
    \end{align*}
    and $\frac{d\tilde f}{dt}$ is the time derivative of $\tilde f$ along the blended dynamics \eqref{eq: blended dyn}.
For subsequent analysis, we need the following lemma on the eigenvalues of the Hessian matrix.
\begin{lemma}
% [Properties of the Hessian matrix]
\label{lm: Hessian}
Given any $\rho \in (0, \frac{\pi}{2})$, for all $\tilde \theta \in \mathbb{S}(\rho)$,
\begin{align}
    \sin(\rho)\lambda_2 I \leq \nabla^2 U(\tilde \theta) \leq \lambda_N I,
\end{align}
where $\lambda_N$ is the largest eigenvalue of $M^{-1}L_B$.
\end{lemma}
\begin{proof}\label{remark 1}
For any $x \in \mathbb{R}^{N-1}$, consider the quadratic form
    \begin{align*}
   & x^T \nabla^2 U(\tilde \theta) x\\
   &= x^T Y^T M^{-\tfrac12} A \, \operatorname{diag}\!\big(\Gamma \cos(A^T M^{-\tfrac12} Y \tilde\theta)\big) A^T M^{-\tfrac12} Y x \\
   &= \sum_{l=1}^E \Gamma_{ll}\, \cos\!\big((A^T M^{-\tfrac12}Y \tilde\theta)_l\big)\, \big((A^T M^{-\tfrac12}Y x)_l\big)^2 \\
   &\ge \sin(\rho)\sum_{l=1}^E \Gamma_{ll}\,\big((A^T M^{-\tfrac12}Y x)_l\big)^2 \\
   &= \sin(\rho)\, x^T Y^T M^{-\tfrac12} A \Gamma A^T M^{-\tfrac12} Y x \\
   &= \sin(\rho)\, x^T Y^T M^{-\tfrac12} L_B M^{-\tfrac12} Y x,
   % \\&= \sin(\rho)\, x^T \Lambda_P x,
\end{align*}
where $(A^T M^{-\tfrac12}Y x)_l$ is the $l$th entry of the vector, and  the inequality follows from $\tilde \theta \in \mathbb{S}(\rho)$ and the last step uses $A \Gamma A^T = L_B$. 
Therefore, the smallest eigenvalue of $\nabla^2 U(\tilde \theta)$ is at least $\sin(\rho)\lambda_2 > 0$. 
     
Similarly,  
       \begin{align*}
   & x^T \nabla^2 U(\tilde \theta) x\\
   &\le \sum_{l=1}^E \Gamma_{ll}\,\big((A^T M^{-\tfrac12}Y x)_l\big)^2 \\
   &= x^T Y^T M^{-\tfrac12} L_B M^{-\tfrac12} Y x.
\end{align*}
 Therefore, the largest eigenvalue of $\nabla^2 U(\tilde \theta)$ is at most $\lambda_N $. 
\end{proof}
Since $\tilde \theta^*(t) \in \mathbb{S}(2\rho)$, it follows from Lemma~\ref{lm: Hessian} that $\nabla^2 U(\tilde \theta^*(t))$ is positive definite and the time derivative of $\tilde \theta^*(t)$ can be derived explicitly from \eqref{eq:theta tilde dot equation} as
\begin{align*}
    \dot{\tilde \theta}^* = [\textcolor{black}{k} \nabla^2 U(\tilde \theta^*) ]^{-1} \frac{d\tilde f}{dt}.
\end{align*}

In addition, by Lemma~\ref{lm: Hessian} and noting that $\mathbb{S}(\rho)$ is convex, we can obtain the following inequalities for all $\tilde \theta, \tilde \theta^\prime$ in $\mathbb{S}(\rho)$ with any $\rho \in (0,\frac \pi 2)$:
\begin{align}
            \sin(\rho)\lambda_{2}|\tilde \theta  - \tilde \theta^\prime | & \leq |\nabla U(\tilde \theta) - \nabla U(\tilde \theta^\prime)|  \leq   \lambda_N |\tilde \theta  - \tilde \theta^\prime |,
            \label{eq: grad U ineq}
            \\
           \frac 1 2  \sin(\rho)\lambda_{2} |\tilde \theta  - \tilde \theta^\prime |^2 &\leq  U(\tilde \theta) -   U(\tilde \theta^\prime) - \nabla U(\tilde \theta^\prime)^T (\tilde \theta - \tilde \theta^\prime)\notag  \\
           &\leq \frac 1 2\lambda_N |\tilde \theta  - \tilde \theta^\prime |^2 .\label{eq: Wp ineq}
        \end{align}

With these in mind, we formally define the error variables 
% \begin{subequations}
    \begin{align*}
       \delta_\omega(t) &:= \omega(t) - \mathbb{1}_N \omega_b(t),\\
          % e_\omega(t)&:= \bar{w}(t) - \omega_b(t), \\
  \delta_\theta(t) &:= \tilde{\theta}(t) - \tilde \theta^*(t).
    \end{align*}
% Define the shorthands
% \begin{align*}
%     \Delta f &:= f(\omega) - f(\mathbb{1}_N \omega_b),\\
%       \tilde{f}&:= Y^T M^{-\frac 1 2}  (f(\mathbb{1}_N \omega_b) + \xi).
% \end{align*}
% \end{subequations}
Then the dynamics of $\delta_\omega$ is given as
\begin{equation*}
    \begin{aligned}
    % \label{eq:err delta w nonlinear}
     % \dot \delta_\omega &= M^{-1}(f(\omega) + \xi - \textcolor{black}{k}M^{\frac 1 2}Y \nabla U(\tilde \theta)) - \frac{\mathbb{1}_N \mathbb{1}_N^T}{NM_b} (f(\mathbb{1}_N \omega_b)+ \xi )
     \dot \delta_\omega &= M^{-1}(f(\omega) + \xi - \textcolor{black}{k}M^{\frac 1 2}Y \nabla U(\tilde \theta)) - \mathbb{1}_N \dot \omega_b \\
     & = M^{-1}(\underbrace{f(\omega) - f(\mathbb{1}_N \omega_b)}_{:= \Delta f})+M^{-1} (f(\mathbb{1}_N \omega_b)+ \xi )  - \mathbb{1}_N \dot \omega_b \\
     & - \textcolor{black}{k} M^{-\frac 1 2}Y (\nabla U(\tilde \theta)  - \nabla U(\tilde \theta^*)) - \textcolor{black}{k} M^{-\frac 1 2}Y \nabla U(\tilde \theta^*)\\
    % \\
    % &= M^{-1}(\underbrace{f(\omega) - f(\mathbb{1}_N \omega_b)}_{:= \Delta f}) + M^{-1}(I -\frac{\mathbb{1}_N \mathbb{1}_N^TM}{NM_b} ) (f(\mathbb{1}_N \omega_b)+ \xi ) \\
    % & - \textcolor{black}{k} M^{-\frac 1 2}Y (\nabla U(\tilde \theta)  - \nabla U(\tilde \theta^*)) - \textcolor{black}{k} M^{-\frac 1 2}Y \nabla U(\tilde \theta^*)
    % \\
    &=  M^{-1}\Delta f - \textcolor{black}{k} M^{-\frac 1 2}Y (\nabla U(\tilde \theta)  - \nabla U(\tilde \theta^*)),
 %    \\
 %        \dot \delta_\omega &= \frac{1}{NM_b} \mathbb{1}_N\mathbb{1}_N^T \Delta f  \\
 %    &+ M^{-\frac 1 2} Y [Y^T M^{-\frac 1 2}\Delta f - 
 % \textcolor{black}{k} (\nabla U(\tilde \theta) - \nabla U(\tilde \theta^*))]\\
 % & = M^{-1} \Delta f - \textcolor{black}{k}  M^{-\frac 1 2} Y(\nabla U(\tilde \theta) - \nabla U(\tilde \theta^*)),
    \end{aligned}
\end{equation*}
where the cancellation in the last step follows from the definition of $\tilde \theta^*$ in \eqref{eq: theta* 1 nonlinear}, i.e., $M^{-1}(f(\mathbb{1}_N \omega_b)+ \xi)- k M^{-\frac 1 2}Y \nabla U(\tilde \theta^*)= \mathbb{1}_N\dot \omega_b$. 
% the identity $M^{-1} -  \mathbb{1}_N \mathbb{1}_N^T/(NM_b) = M^{-\frac 1 2} Y Y^T M^{-\frac 1 2}$ together with the definition of $\tilde \theta^*$ in \eqref{eq:assump condi}.

The dynamics of $\delta_\theta$ is given as
\begin{equation*}
% \label{eq:err e theta nonlinear}
    \begin{aligned}
         \dot{e}_\theta
      &= Y^T M^{\frac 1 2} (\omega - \mathbb{1}_N \omega_b) - [\textcolor{black}{k}\nabla^2 U(\tilde \theta^*) ]^{-1} \frac{d\tilde f}{dt}\\
      &=  Y^T M^{\frac 1 2}\delta_\omega - [\textcolor{black}{k}\nabla^2 U(\tilde \theta^*) ]^{-1} \frac{d\tilde f}{dt}.
    \end{aligned}
\end{equation*}
% Then the system \eqref{eq:error dyn} can be rewriten as
% \begin{subequations}\label{eq:error dyn new}
%     \begin{align}
%     \dot \delta_\omega &= \Delta f - k Y(\nabla U(\tilde \theta) - \nabla U(\tilde \theta^*)),
% \\
%      \dot{e}_\theta
%       &= \textcolor{black}{ Y^T \delta_\omega  - {[k\nabla^2 U(\tilde \theta^*) ]^{-1} \frac{d\tilde f}{dt}}},\label{err-3}
% \end{align}
% \end{subequations}
% where 
% \begin{align}\label{eq:mean value thm 2}
%     \Delta f = \frac{\partial f}{\partial \omega}\Big|_z \cdot \delta_\omega.
% \end{align}
\subsection{Lyapunov Function Analysis}
% \noindent \textbf{Part 2: Lyapunov function analysis} 

The next step is to construct a Lyapunov function $V$ and show that $V$ declines into a small neighborhood of the origin.
% we will show that this is indeed a valid Lyapunov function, by proving positivity outside of the origin and strict negativity of its time derivative along the solutions of (23).
% show exponential decline of the Lyapunov function V (x),
Consider the following Lyapunov function candidate%
% \begin{align}\label{eq:V def}
%     V(e_\omega,\tilde w,\delta_\theta) := \frac{NM_b}{2}e_\omega^2 + \frac{1}{2}\tilde{w}^T\tilde{w} + \frac{1}{2}\delta_\theta^T\Lambda_P \delta_\theta.
% \end{align}
\begin{equation}\label{eq:V def nonlinear}
\begin{aligned}
    \ \ V &
    := W_k(\delta_\omega) + W_p(\tilde \theta,\tilde \theta^*) +\eta W_c (\delta_\omega,\tilde \theta,\tilde \theta^*)
\end{aligned}
\end{equation}
with 
\begin{align*}
W_k(\delta_\omega) &:= \frac 1 2\delta_\omega^T M\delta_\omega,\\
    W_p (\tilde \theta,\tilde \theta^*)&:=  \textcolor{black}{k} \left(U(\tilde \theta) -   U(\tilde \theta^*) - \nabla U(\tilde \theta^*)^T (\tilde \theta - \tilde \theta^*) \right),\\
    W_c (\delta_\omega,\tilde \theta,\tilde \theta^*)&:= \left(\nabla U(\tilde \theta) -   \nabla U(\tilde \theta^*)\right)^T Y^T M^{\frac 1 2}\delta_\omega.
\end{align*}
Here $\eta >0$ is a positive parameter to design, which aims to introduce appropriate cross terms in the Lyapunov analysis. 
Note that the design of the cross terms $W_c$ here is slightly different from that in the proof of Theorem\ref{thm: thm1 general}, thus the selection of $\eta$ will be adjusted accordingly.

In the following lemma, we show that $V$ is a well-defined Lyapunov function when $\tilde \theta$ and $\tilde \theta^*$ belong to $\mathbb{S}(\rho)$ and $\eta$ is properly chosen.
% Define $x = \operatorname{col}(\delta_\omega,\delta_\theta)$. We have
\begin{lemma}
% [positivity of $V$]
\label{lm:posi V}
   Given any $\rho \in (0,\frac{\pi}{2})$, for all $\tilde \theta, \tilde \theta^*$ in $\mathbb{S}(\rho)$,
   % and any $\eta \in (0,\sqrt{\textcolor{black}{k}\lambda_2 \sin(\rho)/\lambda_N^2})$, 
   the function $V$ in \eqref{eq:V def nonlinear} satisfies
       \begin{align}\label{eq: V upper bound in lemma}
           V&\leq \frac{3}{4}\delta_\omega^T M \delta_\omega+ (\frac 1 2  k  \lambda_N + \eta^2 \lambda_N^2)|\delta_\theta|^2, 
       \end{align}
       and 
       \begin{align}\label{eq: V lower bound in lemma}
           V & \geq \frac{1}{4} \delta_\omega^T M \delta_\omega + (\frac{1}{2} k \lambda_2\sin(\rho) - \eta^2\lambda_N^2)|\delta_\theta|^2.
       \end{align}
    % \[
    % \beta_1 |x|^2 \leq V \leq \beta_2 |x|^2
    % \]
    % with \begin{subequations}
    %     \begin{align}
    %         \beta_1 &:= \min\{\frac{1}{4}M_{\min}, \textcolor{black}{k} \lambda_{2}\sin(\rho) - \eta^2 \lambda_N^2 \} >0,\\
    %         \beta_2 &:= M_{\max}+ \textcolor{black}{k} \lambda_N + \eta^2 \lambda_N^2,
    %     \end{align}
    % \end{subequations}
    % where $M_{\min}:= \min_{i \in \mathcal{N}}M_i$ and $M_{\max}:= \max_{i \in \mathcal{N}}M_i$.
\end{lemma}
\begin{proof}
    % First, for the quadratic terms in $e_\omega$ and $\tilde w$,I
    It follows from \eqref{eq: Wp ineq} that
    \begin{align*}
       \frac { k \sin(\rho)\lambda_{2}}{ 2} |\delta_\theta|^2 \leq  W_p (\tilde \theta,\tilde \theta^*) \leq  \frac { k \lambda_N}{ 2} |\delta_\theta|^2.
    \end{align*}
    Besides, use the Young inequalities to obtain
   \begin{align*}
       \eta W_c (\delta_\omega,\tilde \theta,\tilde \theta^*) &\leq  \frac{1}{4}|Y^T M^{\frac 1 2}\delta_\omega|^2 + \eta^2 |\nabla U(\tilde \theta) -   \nabla U(\tilde \theta^*)|^2\\
      & \leq \frac{1}{4}\delta_\omega^T M \delta_\omega
     + \eta^2 \lambda_N^2 |\delta_\theta|^2,\\
      \eta W_c (\delta_\omega,\tilde \theta,\tilde \theta^*)  
      & \geq -\frac{1}{4}|Y^T M^{\frac 1 2}\delta_\omega|^2 - \eta^2 |\nabla U(\tilde \theta) -   \nabla U(\tilde \theta^*)|^2\\
      & \geq -\frac{1}{4}\delta_\omega^T M \delta_\omega - \eta^2 \sin^2(\rho)\lambda_2^2  |\delta_\theta|^2,
    \end{align*}
    where we use $|Y^T M^{\frac 1 2}\delta_\omega|^2  \leq |M^{\frac 1 2}\delta_\omega|^2 = \delta_\omega^T M \delta_\omega$ and the inequality \eqref{eq: grad U ineq}. Putting the above inequalities into the definition of $V$, we arrive at \eqref{eq: V upper bound in lemma} and \eqref{eq: V lower bound in lemma}.
\end{proof}
According to Lemma~\ref{lm:posi V}, we have the following requirement on the choice of $\eta$:
\begin{align}\label{eq: eta require nonlinear}
\eta^2 < \frac{\textcolor{black}{k}\lambda_2 \sin(\rho)}{2\lambda_N^2}.
\end{align}

The next step is to show that $\dot V \leq - c V + \kappa$ with some $c>0 ,\kappa \geq 0$ as long as $\tilde \theta(t) \in  \mathbb{S}(\rho),\ \forall t > 0$. 
To achieve this, we start by developing an upper bound of $\dot V$ term by term.

For the first term $W_k(\delta_\omega)$, its time derivative is given as
\begin{align*}
   \delta_\omega^T M\dot \delta_\omega &=   \delta_\omega^T \Delta f - k \delta_\omega^T M^{\frac 1 2} Y (\nabla U(\tilde \theta) - \nabla U(\tilde \theta^*)).
\end{align*}
Since $f(\omega)= [f_1(\omega_1),\dots,f_N(\omega_N)]^T$, by the mean-value theorem, there exists some $z \in \mathbb{R}^N$ such that
\begin{align*}
% \label{eq:mean value thm nonlinear}
    \Delta f = \frac{\partial f}{\partial \omega}\Bigg|_z \cdot \delta_\omega
\end{align*}
where $\frac{\partial f}{\partial \omega}\big|_z  = \operatorname{diag}(\frac{\partial f_i}{\partial \omega_i}\big|_{z_i},i \in \mathcal{N})\leq -\mu M$ by Assumption~\ref{assump1}. 
% Thus we have
% \begin{align*}
%     \delta_\omega^T \Delta f \leq -\mu \delta_\omega^T M\delta_\omega.
% \end{align*}
% Besides, we can use \eqref{eq: grad U ineq} to obtain 
Then we can bound $\delta_\omega^T M\dot \delta_\omega$ as
\begin{align*}
   \delta_\omega^T M\dot \delta_\omega 
& \leq -\mu \delta_\omega^T M\delta_\omega -\textcolor{black}{k} \delta_\omega^T M^{\frac 1 2} Y (\nabla U(\tilde \theta) - \nabla U(\tilde \theta^*)).
\end{align*}

For the second term $W_p (\tilde \theta,\tilde \theta^*)$, its time derivative is given as 
\begin{align*}
  &  k [\nabla U(\tilde \theta)^T \dot{\tilde \theta} - \nabla U(\tilde \theta^*)^T \dot{\tilde \theta}^*  \\
  &- \nabla U(\tilde \theta^*)^T (\dot{\tilde \theta} - \dot{\tilde \theta}^*)
    - (\tilde \theta -  \tilde \theta^*)^T \nabla^2 U(\tilde \theta^*) \dot{\tilde \theta}^*
]\\
&= \textcolor{black}{k}[      (\nabla U(\tilde \theta) - \nabla U(\tilde \theta^*) ) ^T Y^T M^{\frac 1 2}\delta_\omega   -\textcolor{black}{\frac 1 k}  (\tilde \theta -  \tilde \theta^*)^T \frac{d\tilde f}{dt}
],
\end{align*}
where we use $\dot {\tilde \theta} = Y^T M^{\frac 1 2}\omega =  Y^T M^{\frac 1 2} \delta_\omega$ and $\textcolor{black}{k}\nabla^2  U(\tilde \theta^*) \dot{\tilde \theta}^* =  d\tilde f/dt$.

For the last term $\eta W_c (\delta_\omega,\tilde \theta,\tilde \theta^*)$, its time derivative is given as
\begin{align*}
% 第一项
& \eta [\nabla U(\tilde \theta) -  \nabla  U(\tilde \theta^*)]^T Y^T M^{-\frac 1 2} \Delta f  \\
&- \textcolor{black}{k} \eta [\nabla U(\tilde \theta) -  \nabla U(\tilde \theta^*)]^T[\nabla U(\tilde \theta) - \nabla U(\tilde \theta^*)]\\
% 第二项
&+ \eta \delta_\omega^T M^{\frac 1 2}Y (\nabla^2 U(\tilde \theta) \dot{\tilde \theta} - \nabla^2 U(\tilde \theta^*) \dot{\tilde \theta}^* )\\
% 化简1
&= \eta [\nabla U(\tilde \theta) -  \nabla  U(\tilde \theta^*)]^T Y^T M^{-\frac 1 2} \frac{\partial f}{\partial \omega}\Big|_z  \delta_\omega \\
& - \textcolor{black}{k} \eta |\nabla U(\tilde \theta) - \nabla U(\tilde \theta^*)|^2 \\
% 化简2
& + \eta \delta _\omega^T M^{\frac 1 2} Y\nabla^2 U(\tilde \theta) Y^T M^{\frac 1 2} \delta_\omega \\
& - \eta (\textcolor{black}{\frac 1 k}) \delta_\omega^T M^{\frac 1 2}Y \frac{d \tilde f}{dt}.
\end{align*}
Summing up the above terms, we arrive at
\begin{equation}\label{eq: Vdot upper bound nonlinear 1}
    \begin{aligned}
         \dot V & \leq -\mu   \delta_\omega^T  M \delta_\omega  +  \eta \delta _\omega^T M^{\frac 1 2} Y\nabla^2 U(\tilde \theta) Y^T M^{\frac 1 2} \delta_\omega   -\textcolor{black}{ k} \eta |\nabla U(\tilde \theta) - \nabla U(\tilde \theta^*)|^2 \\
    & + \eta [\nabla U(\tilde \theta) -  \nabla  U(\tilde \theta^*)]^T Y^T M^{-\frac 1 2} \frac{\partial f}{\partial \omega}\Big|_z\delta_\omega\\
    &- \delta_\theta^T \frac{d \tilde f}{dt}- \eta (\textcolor{black}{\frac 1 k}) \delta_\omega^T M^{\frac 1 2}Y \frac{d \tilde f}{dt}.
    \end{aligned}
\end{equation}
When $\tilde \theta(t) \in  \mathbb{S}(\rho)$, we can incorporate the maximum eigenvalue of $ \nabla^2 U(\tilde \theta)$ in Lemma~\ref{lm: Hessian} and the bounds of $\nabla U(\tilde \theta) - \nabla U(\tilde \theta^*) $ in \eqref{eq: grad U ineq}. Then we can further derive upper bounds on the terms in \eqref{eq: Vdot upper bound nonlinear 1} as:
\begin{align*}
    \eta \delta _\omega^T M^{\frac 1 2} Y\nabla^2 U(\tilde \theta) Y^T M^{\frac 1 2} \delta_\omega 
    \leq \eta \lambda_N|Y^T M^{\frac 1 2}\delta_\omega|^2 &\leq \eta \lambda_N \delta_\omega^T M \delta_\omega,
   \\
   -\textcolor{black}{ k} \eta |\nabla U(\tilde \theta) - \nabla U(\tilde \theta^*)|^2  & \leq -\textcolor{black}{ k} \eta\sin^2(\rho)\lambda_2^2 |\delta_\theta|^2,
   \\
    \eta [\nabla U(\tilde \theta) -  \nabla  U(\tilde \theta^*)]^T Y^T M^{-\frac 1 2} \frac{\partial f}{\partial \omega}\Big|_z\delta_\omega & \leq \eta \lambda_N|\delta_\theta|L |M^{\frac  1 2}\delta_\omega|,
    \\
    - \eta (\textcolor{black}{\frac 1 k}) \delta_\omega^T M^{\frac 1 2}Y \frac{d \tilde f}{dt} & \leq   \eta (\textcolor{black}{\frac 1 k}) |M^{\frac  1 2}\delta_\omega| |\frac{d \tilde f}{dt}|.
\end{align*}
Substituting these inequalities into \eqref{eq: Vdot upper bound nonlinear 1} gives
\begin{equation*}
    \begin{aligned}
       \dot V & \leq -(\mu - \eta \lambda_N) \delta_\omega^T M\delta_\omega    - \textcolor{black}{k} \eta \sin^2(\rho)\lambda_2^2 |\delta_\theta|^2 \\
    & +\eta \lambda_N|\delta_\theta|L |M^{\frac  1 2}\delta_\omega|\\
    &+ |\delta_\theta| | \frac{d \tilde f}{dt}| +  \eta (\textcolor{black}{\frac 1 k}) |M^{\frac  1 2}\delta_\omega| |\frac{d \tilde f}{dt}|.
    \end{aligned}
\end{equation*}

Now, we further use the Young inequalities to bound the cross terms and first-order terms in the above upper bound of $\dot V$. 
\begin{subequations}
    \begin{align}
        \eta \lambda_NL  |\delta_\theta| |M^{\frac  1 2}\delta_\omega| & \leq \frac{\textcolor{black}{k}\eta\lambda_2^2\sin^2(\rho)}{4}|\delta_\theta|^2 + \frac{\eta \lambda_N^2 L^2 }{\textcolor{black}{k} \lambda_2^2\sin^2(\rho)} |M^{\frac  1 2}\delta_\omega|^2,\\
        |\delta_\theta||\frac{d\tilde f}{dt}| & \leq \frac{\textcolor{black}{k}\eta\lambda_2^2\sin^2(\rho)}{4}|\delta_\theta|^2  + \frac{1}{\textcolor{black}{k} \eta \lambda_2^2\sin^2(\rho)} |\frac{d\tilde f}{dt}|^2,\\
        \frac{\eta}{\textcolor{black}{k}} |M^{\frac  1 2}\delta_\omega||\frac{d\tilde f}{dt}| & \leq \frac{\eta \lambda_2^2\sin^2(\rho)}{4\textcolor{black}{k}\lambda_N^2 L^2}|\frac{d\tilde f}{dt}|^2 + \frac{\eta\lambda_N^2 L^2}{\textcolor{black}{k}\lambda_2^2\sin^2(\rho)}|M^{\frac  1 2}\delta_\omega|^2.
    \end{align}
\end{subequations}
Using the inequalities above and collecting the coefficients of the quadratic terms, we obtain
\begin{align*}
    \dot V & \leq -(\mu - \textcolor{black}{\eta}\lambda_N - \frac{2\textcolor{black}{\eta} \lambda_N^2 L^2}{\textcolor{black}{k} \lambda_2^2\sin^2(\rho)} )\delta_\omega^T M \delta_\omega\\
    & - \frac 1 2 \textcolor{black}{k} \textcolor{black}{\eta} \lambda_2^2 \sin^2(\rho)|\delta_\theta|^2 \\
    &  +  (\frac{1}{ \textcolor{black}{k} \textcolor{black}{\eta} \lambda_2^2\sin^2(\rho)} + \frac{\textcolor{black}{\eta} \lambda_2^2\sin^2(\rho)}{4 \textcolor{black}{k} \lambda_N^2 L^2})|\frac{d\tilde f}{dt}|^2.
\end{align*}

To ensure that $\dot V$ is negative definite with a decay rate that can be explicitly certified, we impose the requirement:\footnote{Although it suffices to require the coefficient of $\delta_\omega^T M \delta_\omega$ to be positive, we enforce a margin of $\mu/2$ to derive a cleaner decay rate estimation without any minimum-type expressions, which can be seen later.}
\[
\mu - \textcolor{black}{\eta}\lambda_N - \frac{2\textcolor{black}{\eta} \lambda_N^2 L^2}{\textcolor{black}{k} \lambda_2^2\sin^2(\rho)} > \frac{\mu}{2},
\]
that is,
\begin{align*}
    \textcolor{black}{\eta} \leq \frac{\mu}{2\lambda_N \left(1 + \frac{2 \lambda_N L^2}{\textcolor{black}{k} \lambda_2^2 \sin^2(\rho)}\right)}.
\end{align*}
Together with the requirement $\eta^2 < \textcolor{black}{k}\lambda_2 \sin(\rho)/(2\lambda_N^2)$ in \eqref{eq: eta require nonlinear}, we have a convenient explicit choice of $\eta$ as
\begin{align}\label{eq: eta* nonlinear}
    \textcolor{black}{\eta}^* :  = 
\frac{1}{{\frac{2\lambda_N }{\mu}\left(1 + \frac{2 \lambda_N L^2}{\textcolor{black}{k} \lambda_2^2 \sin^2(\rho)}\right)} + \sqrt{\frac{2\lambda_N^2}{\textcolor{black}{k}\lambda_2\sin(\rho)}}},
\end{align}
which uses $\frac{1}{1/A + 1/B} < \min\{A,B\}$ for any $A,B > 0$. 
Then the upper bound of $\dot V$ can be updated as
\begin{equation}
    \begin{aligned}
        \dot V &\leq  - \frac{\mu}{2} \delta_\omega^T M \delta_\omega - \frac 1 2\textcolor{black}{k} \textcolor{black}{\eta}^* \lambda_2^2 \sin^2(\rho)|\delta_\theta|^2\\
        &+(\underbrace{\frac{1}{ \textcolor{black}{k} \textcolor{black}{\eta}^* \lambda_2^2\sin^2(\rho)}}_{:=\varphi_1(k)} + \underbrace{\frac{\textcolor{black}{\eta}^* \lambda_2^2\sin^2(\rho)}{4 \textcolor{black}{k} \lambda_N^2 L^2} }_{:= \varphi_2(k)}
        )|\frac{d\tilde f}{dt}|^2.
    \end{aligned}
\end{equation}
Note from Lemma~\ref{lm:posi V} that
\begin{align*}
V &\leq \frac{3}{4}\delta_\omega^T M \delta_\omega+ (\frac{1}{2}\textcolor{black}{k} \lambda_N + \textcolor{black}{{\eta}^*}^2\lambda_N^2)|\delta_\theta|^2 \\
    & \leq  \frac{3}{4}\delta_\omega^T M \delta_\omega  + \frac{1}{2}\textcolor{black}{k}(\lambda_N+ 
\lambda_2 \sin(\rho))|\delta_\theta|^2,
\end{align*}
where the second step uses ${\eta^*}^2 < \textcolor{black}{k}\lambda_2 \sin(\rho)/(2\lambda_N^2)$. Comparing the upper bounds of $\dot V$ and $V$, we arrive at
\begin{equation}\label{eq: Vdot upper bound nonlinear 2}
    \begin{aligned}
        \dot V \leq - c V +    (\varphi_1(k) + \varphi_2(k))|\frac{d\tilde f}{dt}|^2,
    \end{aligned}
\end{equation}
with 
\begin{align*}
    c&:= \min\{ \frac 2 3\mu ,  \frac{\textcolor{black}{\eta}^* \lambda_2^2 \sin^2(\rho)}{\lambda_N+ 
\lambda_2 \sin(\rho)}\}\\
    & =\frac{\textcolor{black}{\eta}^* \lambda_2^2 \sin^2(\rho)}{\lambda_N+ 
\lambda_2 \sin(\rho)}.
\end{align*}
Here the $\min\{\cdot,\cdot\}$ operator is removed by observing that
\begin{align*}
    \frac{\textcolor{black}{\eta}^* \lambda_2^2 \sin^2(\rho)}{\lambda_N+ 
\lambda_2 \sin(\rho)} < \frac{\mu}{2\lambda_N} \frac{\lambda_2^2 \sin^2(\rho)} {\lambda_N} \leq \frac{1}{2} \mu< \frac{2}{3}\mu,
\end{align*} using $\textcolor{black}{\eta}^* < \frac{\mu}{2\lambda_N} $.

Recall that $\tilde f= Y^T M^{-\frac 1 2}(f(\mathbb{1}_N \omega_b) + \xi)$, where the dynamics of $\omega_b$ and the signal $\xi$ in the nonlinear power flow setting remain identical to that in the linear case. This allows us to substitute the upper bound of $|d\tilde f/dt|^2$  in Lemma~\ref{lm: d tilde f dt} into \eqref{eq: Vdot upper bound nonlinear 2}. Applying the comparison lemma then yields
\begin{equation}\label{eq: Vdot upper nonlinear 3}
    \begin{aligned}
         V(t) & \leq \left( \bar V(0_+)- \frac{2[\varphi_1(k)+ \varphi_2(k)]NM_b C^2 (1+ L/\mu)^2}{c} \right)e^{-ct} \\
    &+ \frac{2[\varphi_1(k)+ \varphi_2(k)]NM_b C^2 (1+ L/\mu)^2}{c}, \ \forall t > 0,
    \end{aligned}
\end{equation}
where
\begin{align*}
  \bar V(0_+):= V(0_+) + \frac{ 2NL^2 [\varphi_1(k)+ \varphi_2(k)]   |f_b(\omega_b(0))+ \xi_b(0_+)|^2 }{M_b (2\mu - c)}.
\end{align*}
% \begin{align*}
%  V(0_+)&\leq \bar V(0_+):=\frac{3}{4}\delta_\omega(0)^T M \delta_\omega(0)+ (\frac 1 2 k \lambda_N + \frac{\mu^2}{4})|\delta_\theta(0_+)|^2.
% \end{align*}
% \begin{align}
%     \limsup_{t \to \infty} V(t) & \leq \frac{\frac{1}{k \textcolor{black}{\eta}^* \lambda_2^2\sin^2(\rho)} + \frac{\textcolor{black}{\eta}^* \lambda_2^2\sin^2(\rho)}{4k\lambda_N^2 L^2}}{ c_1}\limsup_{t \to \infty}|\frac{d\tilde f}{dt}|^2\\
%     & \leq \frac{\frac{1}{k \textcolor{black}{\eta}^* \lambda_2^2\sin^2(\rho)} + \frac{\textcolor{black}{\eta}^* \lambda_2^2\sin^2(\rho)}{4k\lambda_N^2 L^2}}{ c_1} \sqrt{N}C_{\lim}(1+\frac{L}{\mu}).
% \end{align}

For the inequality \eqref{eq: Vdot upper nonlinear 3} to lead to the claimed convergence, we must guarantee that the solutions $\tilde \theta(t)$ would not leave $\mathbb{S}(\rho)$. To do so, we study the sublevel set of $V$ and find one that is contained in $\mathbb{S}(\rho)$.
Define
\begin{align}\label{eq:Vc def}
    V_c :=\frac{( k \lambda_2\sin(\rho) -{\eta^*}^2\lambda_N^2)\rho^2}{2|A^T M^{-\frac 1 2}Y|^2_{2 \to \infty}},
\end{align}
where $|\cdot|_{2 \to \infty}$ is the induced $2 \to \infty$ operator norm. For all $\delta_\omega, \tilde \theta, \tilde \theta^*$ that satisfy $V \leq V_c$, we have
% Whenever $V \leq V_c$, we have
\begin{align*}
    |\delta_\theta|^2 \leq \frac{V}{\frac 1 2 ( k \lambda_2\sin(\rho) -{\eta^*}^2\lambda_N^2)} \leq \frac{\rho^2}{|A^T M^{-\frac 1 2}Y|^2_{2 \to \infty}}.
\end{align*}
Recall in Lemma~\ref{lm:power flow feasi} that $\tilde\theta^* \in \mathbb{S}(2\rho)$, thus the above inequality implies 
\begin{align*}
   | A^T M^{-\frac 1 2}Y \tilde \theta |_{\infty} &\leq | A^T M^{-\frac 1 2} Y \tilde \theta^* |_{\infty} + | A^T M^{-\frac 1 2}Y \delta_\theta |_{\infty}\\
   & \leq \frac{\pi}{2} - 2\rho + |A^TM^{-\frac 1 2} Y|_{2 \to \infty} |\delta_\theta|\\
   & \leq \frac{\pi}{2} - 2\rho +  \rho =  \frac{\pi}{2} - \rho .
\end{align*}
Therefore, to ensure $\tilde\theta(t)\in \mathbb{S}(2\rho)$, it suffices to guarantee $V(t) \leq V_c, \ \forall t  > 0$, which requires the upper bound of $V(t)$ in \eqref{eq: Vdot upper nonlinear 3} to lie below $V_c$ both in the limit and at $t = 0_+$.
% both the limiting bound of $V(t)$ and its initial condition to lie below $V_c$.  说V(t)的initial condition是不对的，应该是upper bound在0时间的值
On the one hand, the limiting bound is below $V_c$ whenever
\begin{align*}
  C=  \sup_{t > 0}\max_{i\in\mathcal N}\frac{|\dot\xi_i(t)|}{M_i} \;\le\; \bar C,  
\end{align*}
where
\begin{align*}
    \bar C:= \sqrt{ \frac{c\, \rho^2 (k\lambda_2\sin(\rho) -{\eta^*}^2\lambda_N^2)}
    {4 N M_b(1+ L/\mu)^2 \|A^T M^{-\frac 1 2}Y\|_{2 \to \infty}^2\,[\varphi_1(k)+ \varphi_2(k)] }}.
\end{align*}
On the other hand, for any given $\xi(0_+)$ (which determines $\tilde \theta^*(0_+)$), the admissible set of the initial states is defined as \begin{align*}
    \mathcal{X}_c:= \left\{(\omega(0),\theta(0)): \bar V(0_+)  \leq V_c
    \right\}.
\end{align*}
By the above construction, if the disturbance $\xi(t)$ satisfies $\max_{i\in\mathcal N}$ \\ $|\dot\xi_i(t)|/M_i \le \bar C$ for all $t>0$ and the initial state lies in $\mathcal{X}_c$, then $V(t)\le V_c$ for all $t>0$, and consequently $\tilde\theta(t)\in\mathbb{S}(\rho)$.

In addition, recall that for Lemma~\ref{lm:power flow feasi} to hold, we require \[
\frac{|\mathbb{1}_N^T M \omega(0)|}{N M_b} \leq \frac{|\xi_b(0_+)|}{\mu M_b } + \frac{k \lambda_2^L \cos(2\rho)}{8 L (\max_i M_i)}:= \varphi_3.
\]
% $ \mu |\mathbb{1}_N^T M \omega(0)| \leq N \max_{\tau \geq 0}|\xi_b(\tau)|$. 
Therefore, the initial states should be further restricted in the following set
\begin{align}\label{eq: X set def}
     \mathcal{X}:= \mathcal{X}_c\cap\left\{(\omega(0),\theta(0)):\frac{|\mathbb{1}_N^T M \omega(0)|}{N M_b} \leq \varphi_3
    \right\}.
\end{align}
Note that the set $\mathcal{X}$ is non-empty. Consider an initial state defined by $ \omega_i(0)  =\omega_b(0) = f_b^{-1}(-\xi_b(0_+)),\ \forall i \in \mathcal{N}$ and $\tilde \theta(0) = \tilde \theta^*(0_+)$. For this choice, we have $\bar V(0_+) = 0 \leq V_c$ and $|\mathbb{1}_N^T M \omega(0)|/(NM_b) = |\omega_b(0)|\leq |\xi_b(0_+)|/(\mu M_b) $ using the mean-value theorem and Assumption~\ref{assump1}. Thus this initial state satisfies both requirements for membership in $\mathcal{X}$, confirming that $\mathcal{X}$ is non-empty. As will become evident from Appendix~\ref{sec: proof of steady state nonlinear}, such initial states are precisely the steady states determined by $\xi(0_+)$, and $\mathcal{X}$ actually restricts $(\omega(0),\theta(0))$ to be not too far from these steady states. 
% \textcolor{black}{For all solutions to stay within $\Omega_c$, we have two requirements:
% \begin{enumerate}
%     \item $C_{\lim} \leq C_{\lim}^*$ such that $\frac 1 {c_1}({\frac{1}{k \eta^* \lambda_2^2\sin^2(\rho)} + \frac{\eta^* \lambda_2^2\sin^2(\rho)}{4k\lambda_N^2 L^2}}) \sqrt{N}C_{\lim}(1+\frac{L}{\mu}) \leq \frac{\beta_1 \rho^2}{|A^T Y|^2_{2 \to \infty}}$.
%     \item The initial solution $(\delta_\omega(0),\delta_\theta(0))$ is within $\Omega_c$.
% \end{enumerate}
% }

Now we are able to use the upper bound of $V(t)$ in \eqref{eq: Vdot upper nonlinear 3} to obtain
\begin{equation}\label{eq: final conclusion in proof nonlinear}
    \begin{aligned}
        & |\omega_i(t) - \omega_b(t)|^2\\
   \leq &\frac{1}{M_i }\delta_\omega^T (t)M \delta_\omega(t)\\
     \leq &\frac{4}{M_i }V(t)\\
       \leq &\alpha e^{-c t} + \beta C^2,\ \forall t > 0,
    \end{aligned}
\end{equation}
where 
\begin{align*}
    \alpha & := \frac{4}{\min_i M_i} \left( \bar V(0_+)- \frac{2[\varphi_1(k)+ \varphi_2(k)]NM_b C^2 (1+ L/\mu)^2}{c} \right),\\
    \beta &:= \frac{8[\varphi_1(k)+ \varphi_2(k)]NM_b (1+ L/\mu)^2}{c\min_i M_i}.
\end{align*}

In particular, inserting the explicit expressions for $\eta^*$ (as defined in  \eqref{eq: eta* nonlinear}) into $\varphi_1(k)$ and $\varphi_2(k)$ yields
\begin{align*}
    \varphi_1(k)&=\frac{1}{ \textcolor{black}{k} \textcolor{black}{\eta}^* \lambda_2^2\sin^2(\rho)}\\
    & = {\frac{2\lambda_N }{\mu \textcolor{black}{k}\lambda_2^2\sin^2(\rho)}\left(1 + \frac{2 \lambda_N L^2}{\textcolor{black}{k}  \lambda_2^2 \sin^2(\rho)}\right)} + \frac{\sqrt{2}\lambda_N}{ \textcolor{black}{k} \lambda_2^2\sin^2(\rho)\sqrt{\textcolor{black}{k} \lambda_2\sin(\rho)}},\\
     \varphi_2(k)&= \frac{\textcolor{black}{\eta}^* \lambda_2^2\sin^2(\rho)}{4k\lambda_N^2 L^2} \\
    & =  \frac{\lambda_2^2\sin^2(\rho)}{4\textcolor{black}{k} \lambda_N^2 L^2 \left[{\frac{2\lambda_N }{\mu}\left(1 + \frac{2 \lambda_N L^2}{\textcolor{black}{k}  \lambda_2^2 \sin^2(\rho)}\right)} + \sqrt{\frac{2\lambda_N^2}{\textcolor{black}{k} \lambda_2\sin(\rho)}}\right]},
\end{align*}
which shows that $\varphi_1(k)$ and $\varphi_2(k)$ are both strictly decreasing in $k$ and tend to $0$ as $k \to \infty$.

Now we can analyze the dependence of $c$ and $\beta$ on $k$. Since
\begin{align*}
    c=\frac{\textcolor{black}{\eta}^* \lambda_2^2 \sin^2(\rho)}{\lambda_N+ 
\lambda_2 \sin(\rho)},
\end{align*}
in which $\eta^*$ is strictly increasing in $k$, we obtain that $c$ is strictly increasing in $k$, and thus $\beta$ is strictly decreasing in $k$. Moreover, since $\textcolor{black}{\eta}^* \to \frac{\mu}{2\lambda_N}$ as $k \to \infty$, we have
\begin{align*}
   \lim_{k \to \infty} c = \frac{ \lambda_2^2 \sin^2(\rho)}{2\lambda_N(\lambda_N+ 
\lambda_2 \sin(\rho))}\mu,
\end{align*}
and thus $\lim_{k \to \infty} \beta = 0$. This completes the proof.

\section{Existence of Solutions to the Steady State Conditions \eqref{eq:steady state condition nonlinear} }\label{sec: proof of steady state nonlinear}
In this section, we show that there exist solutions $(\omega(0),\theta(0))$ to the steady state conditions \eqref{eq:steady state condition nonlinear} under nonlinear power flows when $\xi(0_-)$ is restricted by Assumption~\ref{assump:feasible}.

The existence and uniqueness of the synchronized frequency solution $$\omega_i(0) = \omega_s(0) = f_b^{-1}(-\xi_b(0_-))$$ is established using the exact same arguments as in the linear power flow case. That is because the sum of the sine coupling terms over the entire network is also zero. 

To solve for $\theta(0)$, we rewrite the condition \eqref{eq:steady 1 nonlinear} in a compact form based on the coordinate transformation in Appendix~\ref{prof: thm3}. Following a similar derivation to that of \eqref{eq:lemma feasible eq1}, this condition can be expressed as
\begin{equation}\label{eq:compact_form_proof nonlinear}
    k M^{\frac 1 2}Y \nabla U(\tilde \theta(0)) = g_0,
\end{equation}
where $g_0 = f(\mathbb{1}_N\omega_s(0)) + \xi(0_-)$, $\tilde \theta(0) = Y^TM^{\frac 1 2}\theta(0)$ is the transformed coordinate from \eqref{eq:coordinate transform def}, and $\nabla U(\cdot)$ is the same gradient function as defined in \eqref{eq:lemma feasible eq1}. 

We follow the same line of arguments as the proof of Lemma~\ref{lm:power flow feasi} to show that the above equation admits a unique solution $\tilde \theta(0)$ in $\mathbb{S}(2\rho)$. Similarly using the phase cohesiveness condition in \cite{dorfler2013synchronization}, we are required to show that
% This choice of $\omega(0)$ satisfies both \eqref{eq:steady 2 nonlinear} and \eqref{eq:steady 1 nonlinear} after the latter is left-multiplied by $\mathbb{1}_N^T$. Thus it remains to find a $\theta(0)$ that solves \eqref{eq:steady 1 nonlinear} when left-multiplied by $Y^T M^{-\frac 1 2}$. Following a similar derivation to that of \eqref{eq:assump condi derivation}, this condition can be written in the compact form:
% \begin{equation}
%     Y^T M^{-\frac 1 2}(f(\mathbb{1}_N\omega_s(0)) + \xi(0_-)) - k \nabla U(\tilde \theta(0)) = 0,
% \end{equation}
% where $\tilde \theta(0) = Y^TM^{\frac 1 2}\theta(0)$ is the transformed coordinate from \eqref{eq:coordinate transform def}, and $\nabla U(\cdot)$ is the same gradient function as defined in \eqref{eq:assump condi derivation}.
% Following the same line of arguments as the proof of Lemma~\ref{lm:power flow feasi}, we will show that \eqref{eq:compact_form_proof nonlinear} admits a unique solution $\tilde \theta(0)$ in $\mathbb{S}(2\rho)$. Similarly using the phase cohesiveness condition in \cite{dorfler2013synchronization}, we are required to show that
\begin{align*}
  |f(\mathbb{1}_N\omega_s(0)) + \xi(0_-)|_{\mathcal{E},\infty} \leq k \lambda_2^L \cos(2\rho).
 \end{align*}
 To achieve this, the only difference from the derivation in Lemma~\ref{lm:power flow feasi} is that we bound $|\omega_s(0)|$ instead of $|\omega_b(t)|$. By the mean-value theorem, since $f_b(0)=0$ and $|(f_b^{-1})^\prime|\leq 1/(M_b \mu)$ from Assumption~\ref{assump1}, we have
\begin{align*}
    |\omega_s(0)| \leq \frac{|\xi_b(0_-)|}{M_b \mu}.
\end{align*}
Substituting this bound leads to
\begin{align*}
    |f(\mathbb{1}_N\omega_s(0)) + \xi(0_-)|_{\mathcal{E},\infty} & \leq 2 L (\max_i M_i)\frac{|\xi_b(0_-)|}{M_b \mu} + |\xi|_{\mathcal{E},\infty}\\
    & \leq k \lambda_2^L \cos(2\rho),
\end{align*}
where the second inequality follows from Assumption~\ref{assump:feasible}. This confirms that \eqref{eq:compact_form_proof nonlinear} admits a unique solution $\tilde \theta(0)$ in $\mathbb{S}(2\rho)$. With the unique values of $\omega(0)$ and $\tilde \theta(0)$ determined, any $\theta(0)$ of the form $\mathbb{1}_N\bar \theta(0) + M^{-\frac 1 2}Y \tilde \theta(0)$ for any scalar $\bar \theta(0)$ solves the equations \eqref{eq:steady state condition nonlinear}. 

We can further derive a more explicit expression for $\tilde \theta(0)$. Similar to how we derive \eqref{eq:assump condi derivation}, we first left-multiply \eqref{eq:compact_form_proof nonlinear} by $\mathbb{1}_N^T $, which always holds true due to the definition of $\omega_s(0)$. Next we left-multiply \eqref{eq:compact_form_proof nonlinear} by $Y^T M^{-\frac 1 2}$, which yields:
\begin{align}
    k \nabla U(\tilde \theta(0)) &= Y^T M^{-\frac 1 2}(f(\mathbb{1}_N\omega_s(0)) + \xi(0_-)).
\end{align}
This relation will be further used in the proof of Proposition~\ref{coro: start from steady nonlinear} in Appendix~\ref{prof: coro2}.

\section{Proof of Proposition~\ref{coro: start from steady nonlinear}}\label{prof: coro2}

The proof follows that of Theorem~\ref{thm3:nonlinear} up to the final step in \eqref{eq: final conclusion in proof nonlinear}. The main differences here are twofold: First, we replace the constant $\alpha$ with a more specific form using the steady-state initialization. Second, we transform the requirement $(\omega(0), \theta(0)) \in \mathcal{X}$ into constraints on the initial abrupt changes $|\Delta \xi|$ of disturbances.

\subsection{Replacement of $\alpha$}
We start from the expression of $\alpha$ given in \eqref{eq: final conclusion in proof nonlinear}. Omitting the negative term leads to 
\begin{align*}
    \alpha & \leq  \frac{4}{\min_i M_i}  \bar V(0_+),
\end{align*}
where $\bar V(0_+)$ is previously defined as
\[
 \bar V(0_+)= V(0_+) + \frac{ 2NL^2 (\varphi_1(k)+ \varphi_2(k))   |f_b(\omega_b(0))+  \xi_b(0_+)|^2 }{M_b (2\mu - c)}.
\]
Now we proceed to calculate the two main terms in $\bar V(0_+)$ by incorporating the specific $\omega(0)$ and $\theta(0)$, which is defined in Appendix~\ref{sec: proof of steady state nonlinear} as
\begin{align*}
    \omega(0) &= \mathbb{1}_N \omega_s(0) =  \mathbb{1}_N f_b^{-1}(-\xi_b(0_-)),\\
    k \nabla U(\tilde \theta(0)) &= Y^T M^{-\frac 1 2}(f(\mathbb{1}_N\omega_s(0)) + \xi(0_-)),
\end{align*}
where $\tilde \theta(0) \in \mathbb{S}(2\rho)$. 
Since $\delta_\omega(0) = \omega(0) - \mathbb{1}_N\omega_b(0) = 0$, $V(0_+)$ is simplified to
\[
V(0_+) = k \left(U(\tilde \theta(0)) -    U(\tilde \theta^*(0_+)) - \nabla U(\tilde \theta^*(0_+))^T (\tilde \theta(0) - \tilde \theta^*(0_+)) \right).
\]
Using the inequalities for $U(\cdot)$ provided in \eqref{eq: grad U ineq} and \eqref{eq: Wp ineq}, we can bound $V(0_+)$ as follows:
\begin{align*}
    V(0_+) & \leq k\left(\frac{1}{2}\lambda_{N} |\tilde \theta(0) - \tilde \theta^*(0_+)|^2\right) \\
    & \leq k\left(\frac{1}{2}\lambda_{N} \frac{|\nabla U(\tilde \theta(0)) - \nabla U(\tilde \theta^*(0_+))|^2}{(\sin(\rho)\lambda_2)^2}\right).
\end{align*}
According to the definition of $\tilde \theta(0)$ and $\tilde\theta^*(0_+)$, the term involving the gradient difference can be expressed in terms of the disturbance jump $\Delta\xi$:
\begin{align*}
 & \quad   |\nabla U(\tilde \theta(0)) - \nabla U(\tilde \theta^*(0_+))|^2 \\
    &= \frac{1}{k^2}|Y^T M^{-\frac 1 2}(f(\mathbb{1}_N \omega_s(0)) + \xi(0_-)-f(\mathbb{1}_N \omega_b(0)) - \xi(0_+))|^2 \\
    &= \frac{1}{k^2}|Y^T M^{-\frac 1 2}\Delta \xi|^2 \quad (\text{since } \omega_s(0) = \omega_b(0)) \\
    &\leq \frac{1}{k^2 \min_i M_i }|\Delta \xi|^2.
\end{align*}
Substituting this back gives a bound for $V(0_+)$:
\begin{align*}
V(0_+) \leq \frac{\lambda_{N}}{2k(\min_i M_i)(\sin(\rho)\lambda_2)^2}|\Delta \xi|^2.
\end{align*}
In addition, the second term in $\bar V(0_+)$ depends on $|f_b(\omega_b(0))+  \xi_b(0_+)|^2$, which is now equal to $|\xi_b(0_+)-\xi_b(0_-)|^2 \leq |\Delta \xi|^2/N$. 

Substituting the above bounds into the definition of $\bar V(0_+)$ yields
\begin{align*}
 \bar V(0_+) 
 &\leq \underbrace{\left( \frac{\lambda_{N}}{2k(\min_i M_i)(\sin(\rho)\lambda_2)^2} + \frac{ 2L^2 (\varphi_1(k)+ \varphi_2(k)) }{M_b (2\mu - \mu)} \right)}_{:=\zeta_1} |\Delta\xi|^2,
\end{align*}
where we use $2\mu - c > 2\mu -\mu$ by the definition of $c$. Finally, since $ \alpha  \leq  {4\bar V(0_+)}/{\min_i M_i}$ as previously stated, we obtain
\begin{align*}
    \alpha \leq  \frac{4}{\min_i M_i} \zeta_1 |\Delta \xi|^2 := \alpha^* |\Delta \xi|^2.
\end{align*}
Note that since $\varphi_1(k)$ and $\varphi_2(k)$ are strictly decreasing functions of $k$ that approach zero as $k \to \infty$, we conclude that $\alpha^*$ also strictly decreases with $k$ and $\alpha^* \to 0$ as $k\to\infty$.

\subsection{Constraints on $|\Delta \xi|$}
Next, we show that the initial state requirement $(\omega(0), \theta(0)) \in \mathcal{X}$ translates into an upper bound on $|\Delta\xi|$. Recall that the two conditions for membership in $\mathcal{X}$ are:
\begin{enumerate}
    \item $\bar V(0_+) \leq V_c$, where $V_c>0$ is given in \eqref{eq:Vc def}. 
    \item $\frac{|\mathbb{1}_N^T M \omega(0)|}{N M_b} \leq \frac{|\xi_b(0_+)|}{\mu M_b } + \frac{k \lambda_2^L \cos(2\rho)}{8 L (\max_i M_i)}$.
\end{enumerate}
For the first condition, we use the bound on $\bar V(0_+)$ derived above, which requires
\begin{align*}
    \bar V(0_+) \leq \zeta_1|\Delta \xi|^2 \leq V_c.
\end{align*}
For the second condition, the left-hand side equals $|\omega_s(0)|$ by the definition of $\omega(0)$. This is further bounded by $|\xi_b(0_-)|/(\mu M_b)$ due to the mean-value theorem and Assumption~\ref{assump1}. Thus, the condition is satisfied if
\[
\frac{|\xi_b(0_-)|}{\mu M_b} \leq \frac{|\xi_b(0_+)|}{\mu M_b } + \frac{k \lambda_2^L \cos(2\rho)}{8 L (\max_i M_i)}.
\]
Since $|\xi_b(0_-)| -  |\xi_b(0_+)| \leq {|\Delta\xi|}/{\sqrt{N}} $, it is sufficient to impose
\[
\frac{|\Delta\xi|}{\sqrt{N}} \leq \frac{k \lambda_2^L \cos(2\rho)\mu M_b}{8 L (\max_i M_i)} := \zeta_2.
\]
In summary, for the initial state to be in $\mathcal{X}$, both conditions are guaranteed if $|\Delta\xi| \leq \bar{\Delta}$ with
\begin{align*}
    \bar \Delta:= \min\left\{ \sqrt{\frac{V_c}{\zeta_1}}, \sqrt{N}\zeta_2 \right\}.
\end{align*}
In such cases, the conclusion \eqref{eq: final conclusion in proof nonlinear} holds with the newly specified constant $\alpha^*|\Delta\xi|^2$. This completes the proof.

\end{document}